\documentclass[10pt]{article}
\usepackage{amsfonts}
\usepackage{amssymb}
\usepackage{amsmath}
\usepackage{amsthm}
\usepackage{geometry}
\usepackage{commath}
\usepackage{booktabs}
\usepackage{lscape}
\usepackage{float}

\usepackage{bigstrut}
\usepackage{multirow}
\usepackage{array}
\usepackage{tabularx}
\usepackage{xr}
\usepackage{fancyvrb,booktabs,multirow,array,float}

\usepackage[authoryear]{natbib}

\newcommand{\myref}[2]{\hyperref[#1]{#2}}
\numberwithin{equation}{section}

\usepackage{latexsym}
\usepackage{graphicx} % include figures
\usepackage{enumerate}
\usepackage{setspace}
\usepackage[toc,title,titletoc,header]{appendix}
\usepackage[bottom]{footmisc}   % forces the footnotes to go at the bottom of page.
\usepackage[pdfborder={0 0 0}]{hyperref}
\hypersetup{
  colorlinks=true,linkcolor=blue, citecolor=red, filecolor=magenta, 
  linkbordercolor={0 1 1}, citebordercolor={1 0 0},
}

\usepackage{lineno}
\setpagewiselinenumbers
\VerbatimFootnotes 

\newtheorem{theorem}{Theorem}[section]
\newtheorem{lemma}{Lemma}[section]
\newtheorem{corollary}{Corollary}[section]

\newcounter{assumptionM}
\newcounter{assumptionA}
\def\theassumptionM{A.\arabic{assumptionM}'}
\def\theassumptionA{A.\arabic{assumptionA}}
\newenvironment{assumptionA}[1][]{\refstepcounter{assumptionA}\medskip\noindent{\textbf{Assumption \theassumptionA. #1}} \em \rmfamily}{\medskip}

\hypersetup{
  colorlinks=true,linkcolor=blue, citecolor=red, filecolor=magenta,
  linkbordercolor={0 1 1}, citebordercolor={1 0 0},
}
\begin{document}
\sloppy

%-------- Beginning Title Page -----------------------------
\hypersetup{pageanchor=false}
% \title{\Large Iterated minimum distance estimation of dynamic discrete choice games
\title{On the iterated estimation of dynamic discrete choice games\thanks{We thank the editor and three anonymous referees for comments and suggestions that have greatly improved the manuscript. We also thank Peter Arcidiacono, Patrick Bayer, Joe Hotz, Matt Masten, Arnaud Maurel, Takuya Ura, and seminar participants at various institutions for useful comments and suggestions. Of course, all errors are our own. The research of the first author was supported by the NIH Grant 40-4153-00-0-85-399, and the research of both authors was supported by NSF Grant SES-1729280.}
% \\
% PRELIMINARY AND INCOMPLETE
}
\author{Federico A. Bugni\\Department of Economics\\Duke University\\ \href{mailto:federico.bugni@duke.edu}{\texttt{federico.bugni@duke.edu}}
\and Jackson Bunting\\Department of Economics\\Duke University\\ \href{mailto:jackson.bunting@duke.edu }{\texttt{jackson.bunting@duke.edu}}} 
\date{\today{\\% \vspace{0.5cm}PRELIMINARY AND INCOMPLETE\\PLEASE DO NOT DISTRIBUTE WITHOUT PERMISSION
}}
\maketitle

\begin{abstract}
We study the first-order asymptotic properties of a class of estimators of the structural parameters in dynamic discrete choice games. We consider $K$-stage policy iteration (PI) estimators, where $K$ denotes the number of policy iterations employed in the estimation. This class nests several estimators proposed in the literature. By considering a ``pseudo likelihood'' criterion function, our estimator becomes the $K$-PML estimator in \citet{aguirregabiria/mira:2002,aguirregabiria/mira:2007}. By considering a ``minimum distance'' criterion function, it defines a new $K$-MD estimator, which is an iterative version of the estimators in \cite{pesendorfer/schmidt-dengler:2008} and \cite{pakes/ostrovsky/berry:2007}.

First, we establish that the $K$-PML estimator is consistent and asymptotically normal for any $K \in \mathbb{N}$. This complements findings in \cite{aguirregabiria/mira:2007}, who focus on $K=1$ and $K$ large enough to induce convergence of the estimator. Furthermore, we show under certain conditions that the asymptotic variance of the $K$-PML estimator can exhibit arbitrary patterns as a function of $K$.

Second, we establish that the $K$-MD estimator is consistent and asymptotically normal for any $K \in \mathbb{N}$. For a specific weight matrix, the $K$-MD estimator has the same asymptotic distribution as the $K$-PML estimator. Our main result provides an optimal sequence of weight matrices for the $K$-MD estimator and shows that the optimally weighted $K$-MD estimator has an asymptotic distribution that is invariant to $K$. The invariance result is especially unexpected given the findings in \cite{aguirregabiria/mira:2007} for $K$-PML estimators. Our main result implies two new corollaries about the optimal $1$-MD estimator (derived by \cite{pesendorfer/schmidt-dengler:2008}). First, the optimal $1$-MD estimator is efficient in the class of $K$-MD estimators for all $K \in \mathbb{N}$. In other words, additional policy iterations do not provide first-order efficiency gains relative to the optimal $1$-MD estimator. Second, the optimal $1$-MD estimator is more or equally efficient than any $K$-PML estimator for all $K \in \mathbb{N}$. Finally, the appendix provides appropriate conditions under which the optimal $1$-MD estimator is efficient among regular estimators.

\begin{description}
	\item[Keywords:] dynamic discrete choice problems, dynamic games, pseudo maximum likelihood estimator, minimum distance estimator, estimation, optimality, efficiency.
	\item[JEL Classification Codes: C13, C61, C73]
\end{description}
\end{abstract}
\vfill

\thispagestyle{empty} % no page number \clearpage
\pagebreak % The paper begins in a new page
\setcounter{page}1 % Set the next page number to 1

% \tableofcontents

\pagebreak % The paper begins in a new page
%-------- End Title Page -----------------------------
\hypersetup{pageanchor=true}

%%%%%%%% DIVIDER %%%%%%%%%%%%
\section{Introduction}\label{sec:Introduction}
%%%%%%%% DIVIDER %%%%%%%%%%%%

This paper investigates the first-order asymptotic properties of a broad class of estimators of the structural parameters in a dynamic discrete choice game, i.e., a dynamic game with discrete actions. Given an i.i.d.\ sample of games $i=1,\dots,n$, we consider the class of $K$-stage policy iteration (PI) estimators, where $K$ denotes the number of policy iterations employed in the estimation. The $K$-stage PI estimator is defined by
\begin{equation}
\hat{\alpha} _{K} ~~ \equiv ~~\underset{\alpha \in \Theta_{\alpha}}{\arg \max}~ \hat{Q}( {\alpha,\hat{P}}_{K-1}) ,
\label{eq:est_intro}
\end{equation}
where $\alpha $ is the structural parameter of interest with a true value equal to $\alpha^*\in \Theta_{\alpha}\subseteq \mathbb{R}^{d_{\alpha}}$, $\hat{Q}$ is a sample criterion function, and $\hat{P}_{k}$ is the $k$-stage estimator of the conditional choice probabilities (CCPs), which is defined iteratively as follows. The preliminary estimator of the CCPs is denoted by $ \hat{P}_{0}$. One possible choice of $ \hat{P}_{0}$ is the sample frequency estimator of the CCPs, although this is not required. Then, for any $k=1,\dots, K-1$,
\begin{equation}
\hat{P}_{k}~~\equiv~~ \Psi (\hat{\alpha}_{k},\hat{P}_{k-1}),
\label{eq:CCPMap_intro}
\end{equation}
where $\Psi$ is the best response CCP mapping of the structural game. Given any set of beliefs $P$, optimal or not, $\Psi(\alpha,P) $ indicates the corresponding optimal CCPs when the structural parameter is $\alpha$. The idea of using iterations to estimate dynamic discrete choice problems was introduced in the seminal contributions of \citet{aguirregabiria/mira:2002,aguirregabiria/mira:2007}. As argued in \citet[Page 21]{aguirregabiria/mira:2007}, relative to their non-iterative counterparts, these iterative estimators can be more efficient, have smaller finite sample bias due to a more precise initial estimator of the CCPs, and are robust to inconsistent choices of the initial estimator of the CCPs. Here and throughout this paper, we use ``efficiency'' to denote first-order asymptotic efficiency, and we say that an estimator is ``optimal'' within a certain class of estimators if it first-order asymptotically efficient within said class.

Our $K$-stage PI estimator nests most of the estimators proposed in the dynamic discrete choice games literature. By appropriate choice of $\hat{Q}$ and $K$, our $K$-stage PI estimator coincides with the pseudo maximum likelihood (PML) estimator in \citet{aguirregabiria/mira:2002,aguirregabiria/mira:2007}, the asymptotic least squares estimators in \cite{pesendorfer/schmidt-dengler:2008}, or the so-called simple estimators in \cite{pakes/ostrovsky/berry:2007}.

To implement the $K$-stage PI estimator, the researcher must determine the number of policy iterations $K$. This choice poses several related research questions. How should researchers choose $K$? Does it make a difference? If so, what is the optimal choice of $K$? The literature provides arguably incomplete answers to these questions. The main contribution of this paper is to answer these questions. Before describing our results, we review the main related findings in the literature. 

\citet{aguirregabiria/mira:2002,aguirregabiria/mira:2007} propose $K$-stage PML estimators of the structural parameters in dynamic discrete choice problems. The earlier paper considers single-agent problems whereas the later one generalizes the analysis to multiple-agent problems, i.e., games. In both of these papers, the objective is to maximize the pseudo log-likelihood criterion function $\hat{Q}=\hat{Q}_{PML}$, defined by
\begin{equation}
\hat{Q}_{PML}( \alpha ,P) ~\equiv~\frac{1}{n} \sum_{i=1}^{n}\ln \Psi ( \alpha,P) ( a_{i}|x_{i}).
\label{eq:ML_criterion_intro}
\end{equation}
In this paper, we refer to the resulting estimator as the $K$-PML estimator. One of the main contributions of \citet{aguirregabiria/mira:2002,aguirregabiria/mira:2007} is to study the effect of the number of iterations $K$ on the asymptotic distribution of the $K$-PML estimator.

In single-agent dynamic problems, \cite{aguirregabiria/mira:2002} show that the asymptotic distribution of the $K$-PML estimators is invariant to $K$. In other words, any additional round of policy mapping iteration has no first-order effect on the asymptotic distribution. This striking result is a consequence of the so-called ``zero Jacobian property'', that naturally occurs when a single agent makes optimal decisions. The zero Jacobian property typically does not hold in dynamic problems with multiple players, as each player makes optimal choices according to their preferences, which may not be aligned with their competitors' preferences. Thus, in dynamic discrete choice games, one might expect the asymptotic distribution of $K$-PML estimators to change with $K$.

In multiple-agent dynamic games, \cite{aguirregabiria/mira:2007} show that the asymptotic distribution of the $K$-PML estimators is {\it not} invariant to $K$. They consider two specific choices of $K$. On the one hand, they consider the $1$-PML estimator, which they refer to as the two-step pseudo maximum likelihood (PML) estimator. On the other hand, they propose a sequential nested pseudo likelihood (NPL) algorithm, which consists in increasing $K$ until the $K$-PML estimator converges (i.e., $\hat{\alpha}_{(K-1)-PML}=\hat{\alpha}_{K-PML}$). We refer to the estimator resulting from the convergence the NPL algorithm as the $\infty$-PML estimator.\footnote{\cite{aguirregabiria/mira:2007} propose the sequential NPL algorithm as a way of computing their NPL estimator. To define the NPL estimator, we first define the NPL fixed points as all the pairs of $(\hat{\alpha},\hat{P})$ that satisfy two conditions: (a) given $\hat{P}$, $\hat{\alpha}$ maximizes the PML criterion function in Eq.\ \eqref{eq:ML_criterion_intro} and (b) given $\hat{\alpha}$, $\hat{P}$ is a fixed point of the best response CCP mapping in Eq.\ \eqref{eq:CCPMap_intro}. Then, the NPL estimator is the $\hat{\alpha}$ of the NPL fixed point that maximizes the PML criterion function.} Under some conditions, \cite{aguirregabiria/mira:2007} show that the $1$-PML and $\infty$-PML estimators are consistent and asymptotically normal estimators of $\alpha^* $, i.e.,
\begin{align}
\sqrt{n}( \hat{\alpha}_{1-PML}-\alpha^{\ast }) &~\overset{d}{\to}~ N( {\bf 0}_{d_{\alpha}\times 1},\Sigma_{1-PML}) \notag\\
\sqrt{n}( \hat{\alpha}_{\infty-PML}-\alpha^{\ast }) & ~\overset{d}{\to}~  N({\bf 0}_{d_{\alpha}\times 1},\Sigma_{\infty-PML}). \label{eq:AsyAM2007}
\end{align}
Importantly, under additional conditions, \cite{aguirregabiria/mira:2007} show that $\Sigma_{1-PML}-\Sigma_{\infty-PML}$ is positive definite, that is, the $\infty$-PML estimator is more efficient than the $1$-PML estimator. So, although iterations of the policy mapping may be burdensome, they can improve efficiency within the $K$-PML class.

In later work, \cite{pesendorfer/schmidt-dengler:2010} indicate that the sequential algorithm used to compute the $\infty$-PML estimator may be inconsistent in certain games with unstable equilibria. The intuition for this is as follows. Recall that the $\infty$-PML estimator is defined as the limit of the $K$-PML estimator when $K$ is increased until its convergence. For a given data sample, sampling error implies that the estimator of the CCPs $\hat{P}_{K}$ differs from the equilibrium CCPs. In an unstable equilibrium, increasing $K\to\infty $ can derail $\hat{P}_{K}$ away from the (unstable) equilibrium CCPs, regardless of the sample size $n$. As a consequence of this, the algorithm used to compute $\infty$-PML may produce inconsistent results in the relevant asymptotic framework in which we first consider $ K\to\infty $ and then consider $n\to\infty$.\footnote{Note that this inconsistency would disappear under an asymptotic framework in which we first consider $n\to\infty$ and then consider $K\to\infty$. Unfortunately, by definition, this asymptotic framework is not the relevant one to analyze the asymptotic properties of the $\infty$-PML.} In this paper, we avoid the problem raised by \cite{pesendorfer/schmidt-dengler:2010} because we consider an asymptotic analysis for $K$-stage PI estimators with $K\in \mathbb{N}$ and fixed as $n\to\infty$.

\cite{pesendorfer/schmidt-dengler:2008} consider the estimation of dynamic discrete choice games using a class of minimum distance (MD) estimators. Specifically, their objective is to minimize the sample criterion function $\hat{Q}=\hat{Q}_{MD}$, given by
\begin{equation*}
{\hat{Q}}_{MD}( {\alpha,\hat{P}_{0}}) ~\equiv ~( \hat{P} _{0}-\Psi (\alpha,\hat{P}_{0}) )^{\prime }\hat{W} ( \hat{P}_{0}-\Psi (\alpha,\hat{P}_{0}) ) ,
\end{equation*}
where $\hat{W}$ is a weight matrix that converges in probability to a limiting weight matrix  $W$. This is a single-stage estimator and, consequently, we refer to it as the $1$-MD estimator. \cite{pesendorfer/schmidt-dengler:2008} show that the $1$-MD estimator is a consistent and asymptotically normal estimator of $\alpha $, i.e.,
\begin{align*}
\sqrt{n}( \hat{\alpha}_{1-MD}-\alpha^{\ast }) ~ \overset{d}{\to} ~ N( {\bf 0}_{d_{\alpha}\times 1},\Sigma_{1-MD}(W) ).
\end{align*}
\cite{pesendorfer/schmidt-dengler:2008} show that an appropriate choice of $\hat{W}$ implies that the $1$-MD is asymptotically equivalent to the $1$-PML estimator in \cite{aguirregabiria/mira:2007} or the simple estimators in \cite{pakes/ostrovsky/berry:2007}. Furthermore, \cite{pesendorfer/schmidt-dengler:2008} characterize the optimal choice of $W$, denoted by $W_{1-MD}^{\ast }$. In general, $\Sigma_{1-PML}-\Sigma_{1-MD}( W_{1-MD}^{\ast })$ is positive semidefinite, i.e., the optimal $1$-MD estimator is more or equally efficient than the $1$-PML estimator.

In the light of the results in \cite{aguirregabiria/mira:2007} and \cite{pesendorfer/schmidt-dengler:2008}, it is natural to inquire whether an iterated version of the MD estimator could yield efficiency gains relative to the $1$-MD or the $K$-PML estimators. The consideration of iterated MD estimator opens several important research questions. How should we define the iterated version of the MD estimator? Does this strategy result in consistent and asymptotically normal estimators of $\alpha^*$? If so, how should we choose the weight matrix $W$? What about the number of iterations $K$? Finally, does iterating the MD estimator produce efficiency gains as \cite{aguirregabiria/mira:2007} find for the $K$-PML estimators? This paper answers these questions.

We now summarize the main findings of our paper. We consider a standard dynamic discrete-choice game as in \cite{aguirregabiria/mira:2007} or \cite{pesendorfer/schmidt-dengler:2008}. In this context, we investigate the asymptotic properties of $K$-PML and $K$-MD estimators.

First, we establish that the $K$-PML estimator is consistent and asymptotically normal for any $K\in \mathbb{N}$. See also \cite{aguirregabiria:2004} for related results. This complements findings in \cite{aguirregabiria/mira:2007}, who focus on $K=1$ and $K$ large enough to induce convergence of the estimator. Under certain conditions, we show that the asymptotic variance of the $K$-PML estimator can exhibit arbitrary patterns as a function of $K$. In particular, depending on the parameters of the dynamic problem, the asymptotic variance could increase, decrease, or even be non-monotonic with $K$.

Second, we also establish that the $K$-MD estimator is consistent and asymptotically normal for any $K\in \mathbb{N}$. This is a novel contribution relative to \cite{pesendorfer/schmidt-dengler:2008} or \cite{pakes/ostrovsky/berry:2007}, who focus on non-iterative $1$-MD estimators. The asymptotic distribution of the $K$-MD estimator depends on the choice of the weight matrix. For a specific weight matrix, the $K$-MD has the same asymptotic distribution as the $K$-PML. We investigate the optimal choice of the weight matrix for the $K$-MD estimator. 

Our main result, Theorem \ref{thm:MD_main}, shows that an optimal $K$-MD estimator has an asymptotic distribution that is invariant to $K$. This appears to be a novel result in the literature on PI estimation for games, and it is particularly surprising given the findings in \cite{aguirregabiria/mira:2007} for $K$-PML estimators. Our main result implies two important corollaries regarding the optimal $1$-MD estimator (derived by \cite{pesendorfer/schmidt-dengler:2008}):
\begin{enumerate}
	\item The optimal $1$-MD estimator is efficient in the class of $K$-MD estimators. In other words, additional policy iterations do not provide efficiency gains relative to the optimal $1$-MD estimator. 
	\item The optimal $1$-MD estimator is more or equally efficient than any $K$-PML estimator.
\end{enumerate}
We also show in Section \ref{sec:MLE} that, under suitable conditions, the optimal $1$-MD estimator has the same asymptotic distribution as the maximum likelihood estimator (MLE), and so it is efficient within the class of all regular estimator. Finally, we reiterate that our asymptotic analysis focuses on first-order terms and ignores high-order approximation errors. In finite samples, these terms may generate high-order efficiency gains for the optimal $K$-MD as we vary the number of iterations. Besides exploring this in simulations, we study this issue theoretically in Section \ref{sec:highorder}, where we characterize these high-order terms.

The remainder of the paper is organized as follows. Section \ref{sec:Setup} describes the dynamic discrete choice game used in the paper, introduces the PI estimators and the main assumptions, and provides an illustrative example of the econometric model. Section \ref{sec:ML} studies the asymptotic properties of the $K$-PML estimator. Section \ref{sec:MD} introduces the $K$-MD estimation method, relates it to the $K$-PML method, and studies its asymptotic distribution. Section \ref{sec:MCsimulations} presents a Monte Carlo simulation, and Section \ref{sec:Conclusions} concludes. The appendix of the paper collects the proofs, intermediate results, and complementary findings.

%%%%%%%% DIVIDER %%%%%%%%%%%%
\section{Setup}\label{sec:Setup}
%%%%%%%% DIVIDER %%%%%%%%%%%%

This section describes the econometric model, introduces the estimator and the assumptions, and provides an illustrative example.

%%%%%%%% DIVIDER %%%%%%%%%%%%
\subsection{Econometric model}\label{sec:Model}
%%%%%%%% DIVIDER %%%%%%%%%%%%

We consider a standard dynamic discrete-choice game as described in \cite{aguirregabiria/mira:2007} or \cite{pesendorfer/schmidt-dengler:2008}. The game has discrete time $t=1,\ldots ,T\equiv\infty$, and a finite set of players indexed by $j \in J\equiv \{1,\ldots ,|J|\}$. In each period $t$, every player $j$ observes a vector of state variables $s_{jt}$ and chooses an action $a_{jt}$ from a finite and common set of actions $ A\equiv \{0,1,\ldots ,|A|-1\}$ (with $|A|>1$) to maximize his expected discounted utility. The action denoted by $0$ is referred to as the outside option, and we denote $\tilde{A}\equiv \{1,\ldots ,|A|-1\}$. All players choose their action simultaneously and non-cooperatively upon observation of state variables.

The vector of state variables $s_{jt}$ is composed of two subvectors $x_{t}$ and $\epsilon _{jt}$. The subvector $x_{t}\in X\equiv \{1,\dots ,|X|\}$ represents a state variable observed by all other players and the researcher, whereas the subvector $\epsilon _{jt}\in \mathbb{R}^{|A|}$ represents an action-specific state vector only observed by player $j$. We denote $\epsilon _{t} \equiv \{ \epsilon _{jt}:j \in J\} \in \mathbb{R}^{|A| \times |J|}$ and $\vec{a}_{t} \equiv \{ a_{jt}:j \in J\} \in A^{|J|}$.

We assume that $\epsilon _{t}$ is independent of $(\epsilon _{\tau},\vec{a}_{\tau},x_{\tau})$ for $\tau<t$, with density $dF_{\epsilon }(e) = \prod_{j=1}^{J} dF_{\epsilon,j}(e_j)$, where $dF_{\epsilon,j}$ is an absolutely continuous density function. We also assume that $\epsilon _{t}$ has full support and that $ E[\epsilon _{jt}|\epsilon _{jt}\geq e]$ is finite for all $e\in \mathbb{R}^{|A|} $. Conditional on $(\vec{a}_{t},x_{t})$, we assume that $x_{t+1}$ is independent of $(\epsilon _{\tau},\vec{a}_{\tau-1},x_{\tau-1})$ for $\tau\leq t$, and with probability $dF_{x}(x_{t+1}|\vec{a}_{t},x_{t})$. It then follows that $s_{t+1}=(x_{t+1},\epsilon _{t+1})$ is a Markov process with a probability density that satisfies
\begin{equation*}
d\Pr (x_{t+1},\epsilon _{t+1}|x_{t},\epsilon _{t},\vec{a}_{t})~=~ \prod_{j=1}^{J} dF_{\epsilon,j}(\epsilon _{t+1,j})~\times~dF_{x}(x_{t+1}|\vec{a}_{t},x_{t}).
\end{equation*}

Every player $j$ has a time-separable utility and discounts future payoffs by $\beta^* _{j}\in (0,1)$. The period $t$ payoff is received after every player made their choices and is given by:
\begin{equation*}
\pi _{j}(\vec{a}_{t},x_{t})~+~\sum_{k\in A}\epsilon _{j{t}}( k) ~1 [ a_{jt}=k] .
\end{equation*}

Following the literature, we assume Markov perfect equilibrium (MPE) as the equilibrium concept for the game. By definition, an MPE is a collection of strategies and beliefs for each player such that each player has: (a) rational beliefs, (b) an optimal strategy given his beliefs and other players' choices, and (c) Markovian strategies. According to \citet[Theorem 1]{pesendorfer/schmidt-dengler:2008}, this model has an MPE and it could even have multiple MPEs (e.g., see \citet[Sections 2 and 7]{pesendorfer/schmidt-dengler:2008}). We follow \cite{aguirregabiria/mira:2007} and assume that data come from one of the MPEs in which every player uses pure strategies.\footnote{As explained in \citet[footnote 3]{aguirregabiria/mira:2007}, this can be rationalized by Harsanyi's Purification Theorem.}

An MPE is a collection of equilibrium strategies and common beliefs. We denote the probability that player $j$ will choose action $a\in A $ given observed state $x$ by $P_{j}^{\ast }(a|x)$, and we denote $P^{\ast }\equiv \{P_{j}^{\ast }(a|x):(j,a,x)\in J\times \tilde{A}\times X\} \in \mathbb{R}^{d_{P}}$ with $d_{P} = |J|\times |\tilde{A}| \times |X|$. Note that beliefs only need to be specified in $\tilde{A}$ for every $(j,x)\in J\times X$, as $P_{j}^{\ast }(0|x)=1-\sum_{a\in \tilde{A}}P_{j}^{\ast }(a|x)$. We denote player $j$'s equilibrium strategy by $\{a_{j}^{\ast }(e,x):(e,x)\in \mathbb{R}^{|A|}\times X\}$, where $a_{j}^{\ast }(e,x)$ denotes player $j$'s optimal choice when the current private shock is $e$ and the observed state is $x$. Given that equilibrium strategies are time-invariant, we can abstract from calendar time for the remainder of the paper, and denote $\vec{a}=\vec{a}_{t}$, $\vec{a}^{\prime }=\vec{a}_{t+1}$, $ x=x_{t}$, and $x^{\prime }=x_{t+1}$.

We use $\theta^{\ast } \in \Theta$ to denote the finite-dimensional parameter vector that collects the model elements $(\{\pi _{j}:j \in J\},\{\beta^* _{j}:j \in J\},dF_{\epsilon },dF_{x})$. Throughout this paper, we split the parameter vector as follows:
\begin{align}
	\theta^{\ast }~\equiv~( \alpha^{\ast },g^{\ast }) ~\in~ \Theta ~\equiv~ \Theta _{\alpha }\times \Theta _{g},
	\label{eq:paramSplit}
\end{align}
where $\alpha^{\ast } \in \Theta_{\alpha}\subseteq \mathbb{R}^{d_\theta}$ denotes a parameter vector of interest that is estimated iteratively and $g^{\ast } \in \Theta_g\subseteq \mathbb{R}^{d_g}$ denotes a nuisance parameter vector that is estimated directly from the data. We assume $d_\alpha \leq d_P$ and note that, in a typical application, $d_\alpha$ is much smaller than $d_P$. In practice, structural parameters that determine the payoff functions $\{\pi _{j}:j \in J\}$ or the distribution $dF_{\varepsilon}$ usually belong to $\alpha^{\ast } $, while the transition probability density function $dF_{x}$ is typically part of $g^*$.\footnote{Note that the distinction between components of $\theta^*$ is without loss of generality, as one can choose to estimate all parameters iteratively by setting $\theta^* = \alpha^*$. The goal of estimating a nuisance parameter $g^*$ directly from the data is to simplify the computation of the iterative procedure by reducing its dimensionality.}

We now describe a fixed point mapping that characterizes equilibrium beliefs in any MPE. Let $ P=\{P_{j}(a|x):(j,a,x)\in J\times \tilde{A}\times X\}$ denote a set of beliefs, which need not be optimal. Given these beliefs, the ex-ante probability that player $j$ chooses equilibrium action $a$ given observed state $x$ is
\begin{equation}
\Psi _{j}(a,x; \alpha^{\ast },g^{\ast },P)\equiv \int_{\epsilon }\prod_{k\in A}1[u_{j}(a,x, \alpha^{\ast },g^{\ast },P)+\epsilon _{j}(a)\geq u_{j}(k,x, \alpha^{\ast },g^{\ast },P)+\epsilon _{j}(k)]dF_{\epsilon }(\epsilon |x), \label{eq:Phi}
\end{equation}
where $u_{j}(a,x, \alpha^{\ast },g^{\ast },P)$ denotes player $j$'s conditional choice value function under action $a$, state variable $x$, and with beliefs $P$. In turn,
\begin{equation*}
u_{j}(a,x, \alpha^{\ast },g^{\ast },P)\equiv \sum_{\tilde{a}\in A^{|J|-1}}1[\vec{a}=(a, \tilde{a})]\prod_{s\in J\backslash \{j\}}P_{s}(\tilde{a}_{s}|x)[\pi _{j}((a, \tilde{a}),x)+\beta^*_{j} \sum_{x^{\prime }\in X}dF_{x}(x^{\prime }|(a,\tilde{a} ),x)V_{j}(x^{\prime };P)],
\end{equation*}
where $\prod_{s\in J\backslash \{j\}}P_{s}(\tilde{a}_{s}|x)$ denotes the beliefs that the remaining players choose $\tilde{a}\equiv \{\tilde{a} _{s}:s\in J\backslash \{j\}\}$ conditional on $x$, and $V_{j}(x;P)$ is player $j$'s ex-ante value function conditional on $x$.\footnote{
The ex-ante value function is the discounted sum of future payoffs in the MPE given $x$ and before players observe shocks and choose actions. It can be computed with the mapping valuation operator defined in \citet[Eqs.\ 10 and 14]{aguirregabiria/mira:2007} or \citet[Eqs.\ 5 and 6]{pesendorfer/schmidt-dengler:2008}.} By stacking up this mapping for all decisions and states $(a,x)\in \tilde{A}\times X$ and all players $j \in J$, we define the probability mapping $\Psi ( \alpha^{\ast },g^{\ast },P)\equiv \{\Psi _{j}(a,x; \alpha^{\ast },g^{\ast },P):(j,a,x)\in J\times \tilde{A} \times X\}$. Given any set of beliefs $P$ (optimal or not), $\Psi ( \alpha^{\ast },g^{\ast },P)$ indicates the corresponding optimal CCPs. 

\citet[Representation Lemma]{aguirregabiria/mira:2007} and \citet[Proposition 1]{pesendorfer/schmidt-dengler:2008} show that the mapping $\Psi $ fully characterizes equilibrium beliefs $P^{\ast }$ in an MPE. That is, $P^{\ast }$ is an equilibrium belief if and only if
\begin{equation}
P^{\ast }~=~\Psi ( \alpha^{\ast },g^{\ast },P^{\ast }). \label{eq:FP}
\end{equation}
The goal of the paper is to study the problem of inference of $ \alpha^* \in \Theta _{\alpha }$ based on the fixed point equilibrium condition in Eq.\ \eqref{eq:FP}.

%%%%%%%% DIVIDER %%%%%%%%%%%%
\subsection{Estimation procedure}\label{seq:Procedure}
%%%%%%%% DIVIDER %%%%%%%%%%%%

The researcher estimates $\theta^{\ast }=(\alpha^{\ast },g^{\ast })\in \Theta \equiv \Theta _{\alpha }\times \Theta _{g}$ using a two-step and $K$-stage PI estimator. For any $K\in \mathbb{N}$, this estimator is defined as follows:

\begin{itemize}
\item {\bf Step 1:} Estimate $( g^{\ast },P^{\ast }) $ with preliminary estimators $( \hat{g},\hat{P}_{0})$.

\item {\bf Step 2:} Estimate $\alpha^{\ast }$ with $\hat{\alpha}_{K}$, computed by the following algorithm. Initialize $k=1$ and then:
\begin{itemize}
\item[(a)] Compute
\begin{equation}
\hat{\alpha}_{k}~~\equiv ~~\underset{{\alpha \in \Theta }_{\alpha }}{\arg \max }~{\hat{Q}}_{k}(\alpha ,\hat{g},\hat{P}_{k-1}), \label{eq:k-StepDefn}
\end{equation}
where ${\hat{Q}}_{k}:\Theta _{\alpha }\times \Theta _{g}\times \Theta _{P}\to \mathbb{R}$ is the $k$-th step sample objective function. If $k=K$, exit the algorithm. If $k<K$, go to (b).
\item[(b)] Estimate $P^{\ast }$ with the $k$-step estimator of the CCPs, given by
\begin{equation}
\hat{P}_{k}~~\equiv ~~\Psi ({\hat{\alpha}}_{k},\hat{g},\hat{P}_{k-1}).
\label{eq:k-PkDefn}
\end{equation}
Then, increase $k$ by one unit and return to (a).
\end{itemize}
\end{itemize}

Throughout this paper, we consider $\alpha^{\ast } $ to be our main parameter of interest, while $g^{\ast }$ is a nuisance parameter. For any $K\in \mathbb{N}$, the two-step and $K$-stage PI estimator of $\alpha^{\ast }$ is given by
\begin{equation}
\hat{\alpha}_{K}~~\equiv ~~\underset{{\alpha \in \Theta }_{\alpha }}{\arg \max }~{\hat{Q}}_{K}(\alpha ,\hat{g},\hat{P}_{K-1}), \label{eq:K-StepDefn}
\end{equation}
and the corresponding estimator of $\theta^{\ast }=(\alpha^{\ast },g^{\ast })$ is $\hat{\theta}_{K} = (\hat{\alpha}_{K}, \hat{g})$.

The algorithm does not specify the first-step estimators $(\hat{g},\hat{P}_{0})$ or the sequence of sample criterion functions $\{{\hat{Q}}_{k}:k\leq K\}$. One possible choice of $ \hat{P}_{0}$ is the sample frequency estimator of the CCPs, although this is not required. Rather than determining these objects now, we restrict them by making assumptions in the next subsection (see Assumptions \ref{ass:Baseline} and \ref{ass:Baseline2}).

To conclude this subsection, we note that \cite{aguirregabiria/mira:2007} consider a version of the algorithm described above in which only preliminary estimator of the CCPs is estimated in Step 1, and the entire parameter vector $(g^*,\alpha^*)$ is estimated (iteratively) in Step 2. This can be considered a special case of our algorithm by using $\alpha$ to denote the entire parameter vector $\theta$.

%%%%%%%% DIVIDER %%%%%%%%%%%%
\subsection{Assumptions}\label{sec:assumptions}
%%%%%%%% DIVIDER %%%%%%%%%%%%

This section introduces the main assumptions used in our analysis. We note that these conditions are similar to those used in several other papers in the literature, especially \cite{aguirregabiria/mira:2007} and \cite{pesendorfer/schmidt-dengler:2008}. In addition, our simulation results suggest that our assumptions are satisfied in the entry game described in Section \ref{sec:example}.\footnote{The MPEs in our entry game are found numerically, and this complicates verifying these conditions in practice, especially those related to multiplicity of equilibria. However, our simulation evidence does not suggest issues with any of our assumptions.}

As explained in Section \ref{sec:Model}, the game has an MPE but this need not be unique. To address this issue, we follow most of the literature, and assume that the researcher observes an i.i.d.\ sample from a single MPE.

\begin{assumptionA}[(I.i.d.\ data from one MPE)]\label{ass:iid} 
The data $\{\{(\{a_{j,i}:j \in J\},x_{i},x_{i}^{\prime })\}:i\leq n\}$ are an i.i.d.\ sample from a single MPE. This MPE determines the data generating process (DGP) denoted by $\Pi ^{\ast }\equiv \{\Pi ^{\ast }(\vec{a},x,x^{\prime }):(\vec{a},x,x^{\prime })\in A^{|J|}\times X\times X\}$, where $\Pi ^{\ast }(\vec{a},x,x^{\prime })$ denotes the probability that players choose action $\vec{a}$ and the current state variable evolves from $x$ to $x'$, i.e.,
\[
\Pi ^{\ast }(\vec{a},x,x')~\equiv ~\Pr [~(\{a_{j,i}:j \in J\},x_{i},x_{i}')=(\vec{a},x,x')~].
\]
\vspace{-0.5cm}
\end{assumptionA}

See \citet[Assumptions 5(A) and 5(D)]{aguirregabiria/mira:2007}) for a similar condition. The observations in the i.i.d.\ sample are indexed by $ i=1,\ldots ,n$, which could denote different markets as in \cite{aguirregabiria/mira:2007}. By Assumption \ref{ass:iid}, the data identify the DGP, i.e., $\Pi ^{\ast }(\vec{a},x,x')$ for every $ (\vec{a},x,x')\in  A^{|J|}\times X\times X$, which determine the equilibrium CCPs, transition probabilities, and marginal state distribution. For all $(j,\vec{a},x,x')\in J \times A^{|J|}\times X\times X$ with $\vec{a} = (a, \vec{a}_{-j})$, these are denoted by
\begin{eqnarray}
{P}_{j}^{\ast }(a|x)~ &\equiv &~\frac{\sum_{(\vec{a} _{-j},x^{\prime })\in A^{|J|-1}\times X}\Pi ^{\ast }(( a,\vec{a}_{-j}) ,x,x^{\prime })}{\sum_{(\vec{a},x^{\prime })\in A^{|J|}\times X}\Pi ^{\ast }(\vec{a},x,x^{\prime })} \notag \\
{\Lambda }^{\ast }(x^{\prime }|x,\vec{a})~ &\equiv &~\frac{\Pi ^{\ast }(\vec{ a},x,x^{\prime })}{\sum_{(\vec{a},x)\in A^{|J|}\times X}\Pi ^{\ast }(\vec{a} ,x,x^{\prime })} \notag \\
{m}^{\ast }(x)~ &\equiv &~\sum_{(\vec{a},x^{\prime })\in A^{|J|}\times X}\Pi ^{\ast }(\vec{a},x,x^{\prime }),\label{eq:DGP_elements}
\end{eqnarray}
where ${P}_{j}^{\ast }(a|x)$ denotes the probability that player $j$ will choose action $a$ given that the observed state is $x$, ${ \Lambda }^{\ast }(x'|x,\vec{a})$ denotes the probability that the future state observed state is $x^{\prime }$ given that the current observed state is $x$ and the action vector is $\vec{a}$, and ${m}^{\ast }(x)$ denotes the (unconditional) probability that the current observed state is $x $. Finally, recall that the equilibrium CCPs are $P^{\ast } = \{ P_{j}^{\ast }(a|x):(a,j,x) \in \tilde{A}\times J \times X\}$.

Identification of the CCPs, however, is not sufficient for identification of the parameters of interest. To this end, it is essential to make the following assumption.

\begin{assumptionA}[(Identification)]\label{ass:Identification}
$\Psi ( \alpha ,g^{\ast },P^{\ast }) =P^{\ast }$ if and only if $\alpha =\alpha^{\ast }$.
\end{assumptionA}

See \citet[Assumption 5(C)]{aguirregabiria/mira:2007} and \citet[Assumption A4]{pesendorfer/schmidt-dengler:2008} for a similar condition. Identification in these models is studied in \citet[Section 5]{pesendorfer/schmidt-dengler:2008}. In particular, \citet[Proposition 2]{pesendorfer/schmidt-dengler:2008} indicate the maximum number of parameters that could be identified from the model and \citet[Proposition 3]{pesendorfer/schmidt-dengler:2008} provides sufficient conditions for identification.

The $K$-stage PI estimator $\hat{\alpha} _{K}$ is an example of an extremum estimator. The following assumption imposes mild conditions that are typically used to show the asymptotic properties of these estimators.

\begin{assumptionA}[(Regularity conditions)]\label{ass:Regularity} 
Assume the following conditions:
\begin{enumerate}[(i)]
\item $\alpha^{\ast }$ belongs to the interior of $\Theta _{\alpha }$.
\item $\sup_{\alpha \in \Theta _{\alpha }}|\Psi (\alpha ,\tilde{g},\tilde{P} )-\Psi (\alpha ,g^{\ast },P^{\ast })|=o_{p}(1)$, provided that $(\tilde{g}, \tilde{P})=(g^{\ast },P^{\ast })+o_{p}(1)$.
\item $\inf_{\alpha \in \Theta _{\alpha }}\Psi _{ajx}(\alpha ,\tilde{g}, \tilde{P})>0$ for all $(a,j,x)\in A\times J\times X$, provided that $(\tilde{ g},\tilde{P})=(g^{\ast },P^{\ast })+o_{p}(1)$.
\item $\Psi (\alpha ,g,P)$ is twice continuously differentiable in a neighborhood of $(\alpha^{\ast },g^{\ast },P^{\ast })$. We use $\Psi _{\lambda }\equiv \partial \Psi (\alpha^{\ast },g^{\ast },P^{\ast })/\partial \lambda $ for $\lambda \in \{\alpha ,g,P\}$.
\item $(\mathbf{I}_{d_{P}}-\Psi _{P},-\Psi _{g})\in \mathbb{R}^{d_{P}\times (d_{P}+d_{g})}$ and $\Psi _{\alpha }\in \mathbb{R}^{d_{P}\times d_{\alpha }}$ are full rank matrices.
\end{enumerate}
\end{assumptionA}

We now comment on these conditions. First, the asymptotic analysis of $\hat{\alpha} _{K}$ follows from the first order condition that results from Eq.\ \eqref{eq:K-StepDefn}. Assumption \ref{ass:Regularity}(i) justifies the use of the first order condition for interior parameter values, and is also required by \cite{aguirregabiria/mira:2007} and \citet[Assumption A2]{pesendorfer/schmidt-dengler:2008}. Second, standard arguments to establish the consistency of $\hat{\alpha} _{K}$ require that the corresponding sample criterion function converges uniformly to its limit, which can be related to Assumptions \ref{ass:Regularity}(ii)-(iii). Third, standard argument to prove the asymptotic normality of $\hat{\alpha} _{K}$ are based on a mean value expansion based on the first order condition, which requires second-degree differentiability in Assumption \ref{ass:Regularity}(iv). We note that this assumption coincides with \citet[Assumption A5]{pesendorfer/schmidt-dengler:2008}. Fourth, Assumption \ref{ass:Regularity}(v) imposes rank conditions on the structure of the dynamic game that also required by \citet[Assumption A7]{pesendorfer/schmidt-dengler:2008}. Finally, we note that the validity of Assumptions \ref{ass:Regularity}(iii)-(v) can be formally tested in empirical applications.

We next introduce assumptions on $(\hat{g},\hat{P}_{0}) $, i.e., the preliminary estimators of $( g^{\ast },P^{\ast }) $. For reference, we first define the sample frequency estimator of the CCPs, given by
\begin{equation}
\hat{P}~\equiv ~\{\hat{P}_{j}(a|x):(a,j,x)\in \tilde{A}\times J\times X\}, \label{eq:SampleFrequency}
\end{equation}
with 
\begin{equation*}
\hat{P}_{j}(a|x)~\equiv~ \frac{\sum_{i=1}^{n} 1[(a_{jt,i},x_{t,i})=(a,x)]/n}{ \sum_{i=1}^{n} 1[x_{t,i} = x]/n }.
\end{equation*}
It is not hard to show that
\begin{equation*}
\sqrt{n}(\hat{P}-P^{\ast })~\overset{d}{\to}~N(\mathbf{0}_{d_{P}\times 1},\Omega _{PP}),
\end{equation*}
where $\Omega _{PP}$ is the block diagonal matrix defined by $\Omega _{PP}\equiv diag\{\Sigma _{jx}:(j,x)\in J\times X\}$ with $ \Sigma _{jx}\equiv (diag\{P_{jx}^{\ast }\}-P_{jx}^{\ast }P_{jx}^{\ast \prime })/m^{\ast }(x)$ and $P_{jx}^{\ast }\equiv \{P_{j}^{\ast }(a|x):a\in \tilde{A} \} $ for all $(j,x)\in J\times X$.

Rather than imposing specific preliminary estimators $( \hat{g},\hat{P} _{0}) $, we entertain two high-level assumptions that restrict the relationship between these and $\hat{P}$.

\begin{assumptionA}[(Baseline convergence)]\label{ass:Baseline}
$(\hat{P},\hat{P}_{0},\hat{g})$ satisfies the following
condition: 
\begin{equation*}
\sqrt{{n}}\left( 
\begin{array}{c}
\hat{P}-P^{\ast } \\ 
\hat{P}_{0}-P^{\ast } \\ 
\hat{g}-g^{\ast }
\end{array}
\right) ~~\overset{d}{\to }~~N\left( \left( 
\begin{array}{c}
\mathbf{0}_{d_{P}\times 1} \\ 
\mathbf{0}_{d_{P}\times 1} \\ 
\mathbf{0}_{d_{g}\times 1}
\end{array}
\right) ,\left( 
\begin{array}{ccc}
\Omega _{PP} & \Omega _{P0} & \Omega _{Pg} \\ 
\Omega _{P0}^{\prime } & \Omega _{00} & \Omega _{0g} \\ 
\Omega _{Pg}^{\prime } & \Omega _{0g}^{\prime } & \Omega _{gg}
\end{array}
\right) \right).
\end{equation*}
\end{assumptionA}

\begin{assumptionA}[(Baseline convergence II)]\label{ass:Baseline2}
$(\hat{P},\hat{P}_{0},\hat{g})$ satisfies the following conditions.
\begin{enumerate}[(i)]
\item The asymptotic variance of $(\hat{P},\hat{g})$ is nonsingular.
\item $( \hat{P},\hat{g})$ is an estimator of $(P^*,g^*)$ that at least as efficient as $((\mathbf{I}_{d_{P}}-M) \hat{P}+M\hat{P}_{0},\hat{g})$ for any $M\in \mathbb{R}^{d_{P}\times d_{P}}$.
\end{enumerate}
\end{assumptionA}

Assumption \ref{ass:Baseline} imposes the consistency and joint asymptotic normality of $(\hat{P},\hat{P}_{0},\hat{g})$, which is satisfied by all standard choices for these estimators when the observed states and actions have a finite support. Assumption \ref{ass:Baseline2}(i) is a natural condition that is implicitly imposed by the definition of the optimal weight matrix in \citet[Proposition 5]{pesendorfer/schmidt-dengler:2008}. The interpretation of Assumption \ref{ass:Baseline2}(ii) requires additional discussion. The point of a preliminary estimator $(\hat{P}_{0},\hat{g})$ is to approximate $(P^{\ast },g^{\ast })$ without the need for imposing the restrictions from the structural model, since this may be computationally burdensome. If we ignore these restrictions, $\hat{P}$ is the maximum likelihood estimator (MLE) of $ P^{\ast }$ and it is thus an efficient estimator of $P^{\ast }$.  As a corollary, $\hat{P}$ is an estimator of $P^*$ that is at least as efficient as $(\mathbf{I}_{d_{P}}-M)\hat{P}+M\hat{P}_{0}$ for any $M\in \mathbb{R}^{d_{P}\times d_{P}}$. Assumption \ref{ass:Baseline2}(ii) essentially requires that this conclusion also applies when these estimators are coupled with $\hat{g}$ as an estimator of $g^*$.

To illustrate Assumptions \ref{ass:Baseline} and \ref{ass:Baseline2}, it is necessary to specify the parameter vector $g^{\ast }$. A typical specification in the literature is the one in \cite{pesendorfer/schmidt-dengler:2008}, where $g^{\ast }$ is the vector of state transition probabilities, i.e.,
\begin{equation}
g^{\ast }~=~\{\Lambda^* (x^{\prime }|\vec{a},x):(x^{\prime },\vec{a},x)\in \tilde{X} \times A^{|J|}\times X\},
\label{eq:g_star_example}
\end{equation}
and $\tilde{X}\equiv \{ 2,\ldots ,|X|\} $ is the state space with the first action removed (to avoid redundancy). In this setting, a reasonable specification of $( \hat{P}_{0},\hat{g})$ is to set them equal to their corresponding sample frequency estimators, i.e., $\hat{P}_{0}=\hat{P}$ and $\hat{g}\equiv \{\hat{ g}(\vec{a},x,x^{\prime }):(x^{\prime },\vec{a},x)\in \tilde{X} \times A^{|J|}\times X\}$ with
\begin{equation}
\hat{g}(\vec{a},x,x^{\prime })~\equiv ~\frac{\sum_{i=1}^{n}1[(\vec{a} _{t,i},x_{t,i},x_{t,i}^{\prime })=(\vec{a},x,{x}^{\prime })]/n}{ \sum_{i=1}^{n}1[(\vec{a} _{t,i},x_{t,i})=(\vec{a},x)]/n}.
\label{eq:g_hat_example}
\end{equation}
If we ignore the restrictions of the econometric model, $(\hat{P}_{0},\hat{g})$ is the MLE of $( P^{\ast },g^{\ast }) $, and standard arguments imply Assumptions \ref{ass:Baseline} and \ref{ass:Baseline2}. However, note that one could obtain the same results if were replaced $(\hat{P}_{0},\hat{g})$ with any asymptotically equivalent estimator, i.e., any estimator $(\tilde{P}_{0}, \tilde{g})$ such that $(\tilde{P}_{0},\tilde{g})=(\hat{P}_{0},\hat{g} )+o( n^{-1/2}) $. Examples of asymptotically equivalent estimators would be the ones resulting from the flexible logit estimator (e.g., see \citet[page 382]{arcidiacono/ellickson:2011}), the relaxation method in \cite{kasahara/shimotsu:2012}, or the undersmoothed kernel estimator in \citet[Theorem 5.3]{grund:1993}.

%%%%%%%% DIVIDER %%%%%%%%%%%%
\subsection{An illustrative example}\label{sec:example}
%%%%%%%% DIVIDER %%%%%%%%%%%%

We illustrate the framework with the two-player version dynamic entry game in \citet[Example 5]{aguirregabiria/mira:2007}. In each period $t=1,\ldots ,T\equiv\infty$, two firms indexed by $j \in J=\{1,2\}$ simultaneously decide whether to enter or not into the market, upon observation of the state variables. Firm $j$'s decision at time $t$ is $a_{jt}\in A=\{0,1\}$, which takes value one if firm $j$ enters the market at time $t$, and zero otherwise. In each period $t$, the vector of state variables observed by firm $j$ is $s_{jt}=(x_{jt},\epsilon _{jt})$, where $\epsilon _{jt}=(\epsilon _{0jt},\epsilon _{1jt})\in \mathbb{R}^{2}$ represents the privately-observed vector of action-specific state variables and $ x_{jt}=x_{t}\in X\equiv \{ 1,2,3,4\} $ is a publicly-observed state variable that indicates the entry decisions in the previous period, i.e.,
\begin{equation*}
x_{jt}~=~\left[ 
\begin{array}{c}
1[ (a_{1,t-1},a_{2,t-1})=(0,0)] +2\times 1[ (a_{1,t-1},a_{2,t-1})=(0,1)] + \\
3\times 1[ (a_{1,t-1},a_{2,t-1})=(1,0)] +4\times 1[ (a_{1,t-1},a_{2,t-1})=(1,1)]
\end{array}
\right].
\end{equation*}

We specify the profit function as in \citet[Eq.\ (48)]{aguirregabiria/mira:2007}. If firm $j$ enters the market in period $t$, its period profits are
\begin{equation} 
	\pi _{j}((1,a_{-j,t}),x_{t})~=~\lambda _{RS}^{\ast }-\lambda _{RN}^{\ast }\ln (1+a_{-j,t})-\lambda _{FC,j}^{\ast }-\lambda _{EC}^{\ast }(1-a_{j,t-1})+\epsilon _{1jt}, \label{eq:ProfitsEnter} 
\end{equation} 
where $\lambda _{RS}^{\ast }$ represents fixed entry profits, $\lambda _{RN}^{\ast }$ represents the effect of a competitor's entry, $\lambda  _{FC,j}^{\ast }$ represents a firm-specific fixed cost, and $\lambda _{EC}^{\ast }$ represents the entry cost. On the other hand, if firm $j$ does not enter the market in period $t$, its period profits are 
\begin{equation*} 
	\pi _{j}((0,a_{-j,t}),x_{t})~=~\epsilon _{0jt}. 
\end{equation*} 
Firms discount future profits at a common discount factor $\beta^{\ast }\in (0,1)$. We assume that $\{\epsilon _{ajt}:(a,j,t) \in A \times J \times \mathbb{N}\}$ are i.i.d.\ with standard Gumbel distribution, and so
\begin{equation*}
dF_{\epsilon,j }(e_{jt})~=~\prod_{a=0}^{1}\exp (-\exp (-e_{ajt})).
\end{equation*}
Finally, $x_{t+1}$ is uniquely determined by $ \vec{a}_{t}$ and so,
\begin{equation*}
dF_{x}(x_{t+1}|\vec{a}_{t},x_{t})~=~1\left[ x_{t+1}=\left[ 
\begin{array}{c}
1[(a_{1t},a_{2t})=(0,0)]+2\times 1[(a_{1t},a_{2t})=(0,1)]+ \\ 3\times 1[(a_{1t},a_{2t})=(1,0)]+4\times 1[(a_{1t},a_{2t})=(1,1)]
\end{array}
\right] \right].
\end{equation*}
This completes the specification of the econometric model up to $(\lambda _{RN}^{\ast },\lambda _{EC}^{\ast },\lambda _{RS}^{\ast },\lambda _{FC,1}^{\ast },\lambda _{FC,2}^{\ast },\beta^{\ast })$. These parameters are known to the players but not necessarily known to the researcher.

We will use this econometric model to illustrate our theoretical results and for our Monte Carlo simulations. For simplicity, we presume that the researcher knows $(\lambda _{RS}^{\ast },\lambda _{FC,1}^{\ast },\lambda _{FC,2}^{\ast },\beta^{\ast })$, and is interested in estimating $(\lambda _{RN}^{\ast },\lambda _{EC}^{\ast })$. In addition, we assume that these parameters are all estimated iteratively in Step 2, i.e., $ \theta^{\ast }=\alpha^{\ast }\equiv (\lambda _{RN}^{\ast },\lambda _{EC}^{\ast })$. The only task in Step 1 is to estimate $P^{\ast }$ with a preliminary step estimator $\hat{P}_{0}$, i.e., this model has no parameter $g^{\ast }$. 

%%%%%%%% DIVIDER %%%%%%%%%%%%
\section{Results for $K$-PML estimation}\label{sec:ML}

This section provides formal results for the $K$-PML estimator introduced in \cite{aguirregabiria/mira:2002,aguirregabiria/mira:2007} given an arbitrary number of iteration steps $K\in \mathbb{N}$. The $K$-PML estimator is defined by Eq.\ \eqref{eq:K-StepDefn} with the pseudo log-likelihood criterion function, i.e., $\hat{Q}_{K}=\hat{Q}_{PML}$. That is,
\begin{itemize}
\item {\bf Step 1:} Estimate $( g^{\ast },P^{\ast }) $ with preliminary step estimators $( \hat{g},\hat{P}_{0}) $.
\item {\bf Step 2:} Estimate $\alpha^{\ast }$ with $\hat{\alpha}_{K-PML}$, computed by the following algorithm. Initialize $k=1$ and then:
\begin{itemize}
\item[(a)] Compute
\begin{equation*}
\hat{\alpha}_{k-PML}~\equiv~ \underset{\alpha \in \Theta _{\alpha }}{\arg \min}~\frac{1}{n} \sum_{i=1}^{n}\ln \Psi ( \alpha ,\hat{g},\hat{P}_{k-1}) ( a_{i}|x_{i}) .
\end{equation*}
If $k=K$, exit the algorithm. If $k<K$, go to (b).
\item[(b)] Estimate $P^{\ast }$ with the $k$-step estimator of the CCPs, given by
\begin{equation*}
\hat{P}_{k}~\equiv~ \Psi ( \hat{\alpha}_{k-PML},\hat{g},\hat{P} _{k-1}) .
\end{equation*}
Then, increase $k$ by one unit and return to (a).
\end{itemize}
\end{itemize}

As explained in Section \ref{sec:Introduction}, The $K$-PML estimator is the $K$-stage PI estimator introduced by \cite{aguirregabiria/mira:2002} for dynamic single-agent problems and \cite{aguirregabiria/mira:2007} for dynamic games. \cite{aguirregabiria/mira:2007} study the asymptotic behavior of $\hat{\alpha}_{K-PML}$ for two extreme values of $K$: $K=1$ and $K$ large enough to induce the convergence of the estimator. Under some conditions, they show that iterating the $K$-PML estimator until convergence produces efficiency gains.

The results in \cite{aguirregabiria/mira:2007} are restricted in two ways. First, they focus on these two extreme values of $K$, without considering other possible values. Second, they restrict attention to the case in which the parameters of the model are estimated in the iterative part of the algorithm (i.e., $\theta = \alpha$). In Theorem \ref{thm:ML}, we complement the analysis in \cite{aguirregabiria/mira:2007} along these two dimensions. In particular, we derive the asymptotic distribution of the two-step $K$-PML estimator for any finite $K\geq 1$.

\begin{theorem}[Two-step $K$-PML]\label{thm:ML}
Fix $K\in \mathbb{N}$ arbitrarily and assume Assumptions \ref{ass:iid}-\ref{ass:Baseline}. Then,
\begin{equation*}
\sqrt{n}( \hat{\alpha}_{K-PML}-\alpha^{\ast }) ~~\overset{d}{ \to }~~N( \mathbf{0}_{d_{\alpha }\times 1},\Sigma _{K-PML}( \hat{P}_{0},\hat{g}) ) ,
\end{equation*}
where
\begin{equation*}
\Sigma _{K-PML}( \hat{P}_{0}) \equiv \left\{ 
\begin{array}{c}
( \Psi _{\alpha }^{\prime }\Omega _{PP}^{-1}\Psi _{\alpha })^{-1}\Psi _{\alpha }^{\prime }\Omega _{PP}^{-1}\times \\
\left[ 
\begin{array}{c}
( \mathbf{I}_{d_{P}}-\Psi _{P}\Phi _{K,P})^{\prime } \\ 
-( \Psi _{P}\Phi _{K,0})^{\prime } \\ 
-\Psi _{g}^{\prime }( \mathbf{I}_{d_{P}}+\Psi _{P}\Phi _{K,g})
^{\prime }
\end{array}
\right]^{\prime }\left(
\begin{array}{ccc}
\Omega _{PP} & \Omega _{P0} & \Omega _{Pg} \\ 
\Omega _{P0}^{\prime } & \Omega _{00} & \Omega _{0g} \\ 
\Omega _{Pg}^{\prime } & \Omega _{0g}^{\prime } & \Omega _{gg}
\end{array}
\right) \left[ 
\begin{array}{c}
( \mathbf{I}_{d_{P}}-\Psi _{P}\Phi _{K,P})^{\prime } \\ 
-( \Psi _{P}\Phi _{K,0})^{\prime } \\ 
-\Psi _{g}^{\prime }( \mathbf{I}_{d_{P}}+\Psi _{P}\Phi _{K,g})^{\prime }
\end{array}
\right]  \\ 
\times \Omega _{PP}^{-1}\Psi _{\alpha }( \Psi _{\alpha }^{\prime }\Omega _{PP}^{-1}\Psi _{\alpha })^{-1}
\end{array}
\right\} ,
\end{equation*}
and $\{ \Phi _{k,P}:k\leq K\} $, $\{ \Phi _{k,0}:k\leq K\} $, and $\{ \Phi _{k,g}:k\leq K\} $ are defined as follows. Set $\Phi _{1,P}\equiv \mathbf{0}_{d_{P}\times d_{P}}$, $\Phi _{1,0}\equiv \mathbf{I}_{d_{P}}$, $\Phi _{1,g}\equiv \mathbf{0}_{d_{P}\times d_{P}}$ and, for any $k\leq K-1$,
\begin{align}
\Phi _{k+1,P}& ~\equiv~ ( \mathbf{I}_{d_{P}}-\Psi _{\alpha }( \Psi _{\alpha }^{\prime }\Omega _{PP}^{-1}\Psi _{\alpha })^{-1}\Psi _{\alpha }^{\prime }\Omega _{PP}^{-1}) \Psi _{P}\Phi _{k,P}+\Psi _{\alpha }( \Psi _{\alpha }^{\prime }\Omega _{PP}^{-1}\Psi _{\alpha })^{-1}\Psi _{\alpha }^{\prime }\Omega _{PP}^{-1}, \notag \\
\Phi _{k+1,0}& ~\equiv~( \mathbf{I}_{d_{P}}-\Psi _{\alpha }( \Psi _{\alpha }^{\prime }\Omega _{PP}^{-1}\Psi _{\alpha })^{-1}\Psi _{\alpha }^{\prime }\Omega _{PP}^{-1}) \Psi _{P}\Phi _{k,0}, \notag \\
\Phi _{k+1,g}& ~\equiv~( \mathbf{I}_{d_{P}}-\Psi _{\alpha }( \Psi _{\alpha }^{\prime }\Omega _{PP}^{-1}\Psi _{\alpha })^{-1}\Psi _{\alpha }^{\prime }\Omega _{PP}^{-1}) \Psi _{P}\Phi _{k,g}+( \mathbf{I}_{d_{P}}-\Psi _{\alpha }( \Psi _{\alpha }^{\prime }\Omega _{PP}^{-1}\Psi _{\alpha })^{-1}\Psi _{\alpha }^{\prime }\Omega _{PP}^{-1}) . \label{eq:coefficients_ML}
\end{align}
\end{theorem}

We make two comments regarding this result. First, Theorem \ref{thm:ML} considers $K\in \mathbb{N}$ and \textit{fixed} as $n\to \infty$. Because of this, our asymptotic framework is not subject to the criticism raised by \cite{pesendorfer/schmidt-dengler:2010}. Second, we note that $\Sigma _{K-PML}(\hat{P}_{0},\hat{g})$ can be consistently estimated using consistent estimators of $(\alpha^*,g^*)$ (e.g., $(\hat{\alpha} _{1-PML},\hat{g})$) and the asymptotic variance in Assumption \ref{ass:Baseline}.

Theorem \ref{thm:ML} reveals that the $K$-PML estimator of $\alpha^{\ast }$ is consistent and asymptotically normally distributed for all $K\geq 1$. Thus, the asymptotic mean squared error of the $K$-PML estimator is equal to its asymptotic variance, $\Sigma _{K-PML}( \hat{P}_{0}) $. The goal for the rest of the section is to investigate how this asymptotic variance changes with the number of iterations $K$.

In single-agent dynamic problems, \cite{aguirregabiria/mira:2002} show that the so-called zero Jacobian property holds, i.e., $\Psi _{P}=\mathbf{0} _{dP\times dP}$. If we plug in this condition into Theorem \ref{thm:ML}, we conclude that the asymptotic variance of the $K$-PML estimator is given by
\begin{align*}
	\Sigma _{K-PML}( \hat{P}_0) = (\Psi _{\alpha }^{\prime }\Omega _{PP}^{-1}\Psi _{\alpha })^{-1}\Psi _{\alpha }^{\prime }\Omega _{PP}^{-1} ( \Omega _{PP}+\Psi _{g}\Omega _{gg}\Psi _{g}^{\prime } -\Psi _{g}\Omega _{Pg}^{\prime }-\Omega _{Pg}\Psi _{g}^{\prime } ) \Omega _{PP}^{-1}\Psi _{\alpha }(\Psi _{\alpha }^{\prime }\Omega _{PP}^{-1}\Psi _{\alpha })^{-1}.
\end{align*}
This expression is invariant to $K$, corresponding to the main finding in \cite{aguirregabiria/mira:2002}.

In multiple-agent dynamic problems, however, the zero Jacobian property no longer holds. In the special case with $ d_\alpha = d_P$, Theorem \ref{thm:ML} shows that the asymptotic variance of the $K$-PML does not depend on $K$, and is given by
\begin{align*}
    \Sigma _{K-PML}(\hat{P}_{0})=\Psi _{\alpha }^{-1}[(\mathbf{I}_{d_{P}}-\Psi
_{P})\Omega _{PP}(\mathbf{I}_{d_{P}}-\Psi _{P})^{\prime }-\Psi _{g}\Omega
_{Pg}^{\prime }(\mathbf{I}_{d_{P}}-\Psi _{P})^{\prime }-(\mathbf{I}_{d_{P}}-\Psi _{P})\Omega _{Pg}\Psi _{g}^{\prime }-\Psi _{g}\Omega _{gg}\Psi
_{g}^{\prime }](\Psi _{\alpha }^{\prime })^{-1}.
\end{align*}
However, most applications will have $ d_\alpha<d_P$. For this more common case, Theorem \ref{thm:ML} reveals that the asymptotic variance of the $K$-PML estimator can be a complicated function of the number of iteration steps $K$. We illustrate this complexity using the example of Section \ref{sec:example}. In this example, the researcher is interested in estimating $(\lambda _{RN}^{\ast },\lambda _{EC}^{\ast })$. For simplicity, we set $\hat{P}_{0}=\hat{P}$. In this context, the asymptotic variance of $\hat{\alpha}_{K-PML}=(\hat{\lambda}_{RN,K-PML},\hat{ \lambda}_{EC,K-PML})$ is given by
\begin{equation}
\Sigma _{K-PML}( \hat{P}_0) = (\Psi _{\alpha }^{\prime }\Omega _{PP}^{-1}\Psi _{\alpha })^{-1}\Psi _{\alpha }^{\prime }\Omega _{PP}^{-1}( \mathbf{I}_{d_{P}}-\Psi _{P}\Phi _{K,P0})\Omega _{PP}(\mathbf{I} _{d_{P}}-\Psi _{P}\Phi _{K,P})^{\prime }\Omega _{PP}^{-1}\Psi _{\alpha }(\Psi _{\alpha }^{\prime }\Omega _{PP}^{-1}\Psi _{\alpha })^{-1}, \label{eq:Sigma_ML_example}
\end{equation}
where $\{\Phi _{k,P0}:k\leq K\}$ is defined by $\Phi _{k,P0}\equiv \Phi _{k,P}+\Phi _{k,0}$, with $\{\Phi _{k,P}:k\leq K\}$ and $\{\Phi _{k,0}:k\leq K\}$ as in Eq.\ \eqref{eq:coefficients_ML}. For any true parameter vector and any $K\in \mathbb{N} $, we can numerically compute Eq.\ \eqref{eq:Sigma_ML_example}. For the exposition, we focus on the asymptotic variance of $\hat{\lambda} _{RN,K-PML}$, which corresponds to the [1,1]-element of $\Sigma _{K-PML}( \hat{P}_0 )$. Figures \ref{fig:example1}, \ref{fig:example2}, and \ref{fig:example3} show the asymptotic variance of $\hat{ \lambda}_{RN,K-PML}$ as a function of $K$ for $(\lambda _{RN}^{\ast },\lambda _{EC}^{\ast },\lambda _{RS}^{\ast },\lambda _{FC,1}^{\ast },\lambda _{FC,2}^{\ast },\beta^{\ast }) = (2.8,0.8,0.7,0.6,0.4,0.95)$, $(2,1.8,0.2,0.01,0.03,0.95)$, and $ (2.2,1.45,0.45,0.22,0.29,0.95)$, respectively. These figures confirm that, in general, the asymptotic variance of the $K$-PML estimator can decrease, increase, or even fluctuate with the number of iterations $K$. Note that these widely different patterns occur within the same econometric model. Finally, we point out that qualitatively similar results can be obtained for the asymptotic variance of $\hat{\lambda} _{EC,K-PML}$, which corresponds to the [2,2]-element of $\Sigma _{K-PML}( \hat{P}_0 )$.

We view the fact that $\Sigma _{K-PML}( \hat{P}_{0}) $ can change so much with the number of iterations $K$ as a negative feature of the $K$-PML estimator. A researcher who uses the $K$-PML estimator and is interested in efficiency faces difficulties when choosing $K$. Prior to estimation, the researcher cannot be certain regarding the effect of $K$ on the efficiency of the $K$-PML estimator. Additional iterations could help efficiency (as in Figure \ref{fig:example1}) or hurt efficiency (as in Figure \ref{fig:example2}). In principle, the researcher could consistently estimate the asymptotic variance of $K$-PML for each $K$ by plugging in any consistent estimator of the structural parameters (e.g., $\hat{\alpha} _{1-PML}$ and $\hat{g}$). 

\begin{figure}[H]
	\centering
		\includegraphics[scale=0.57]{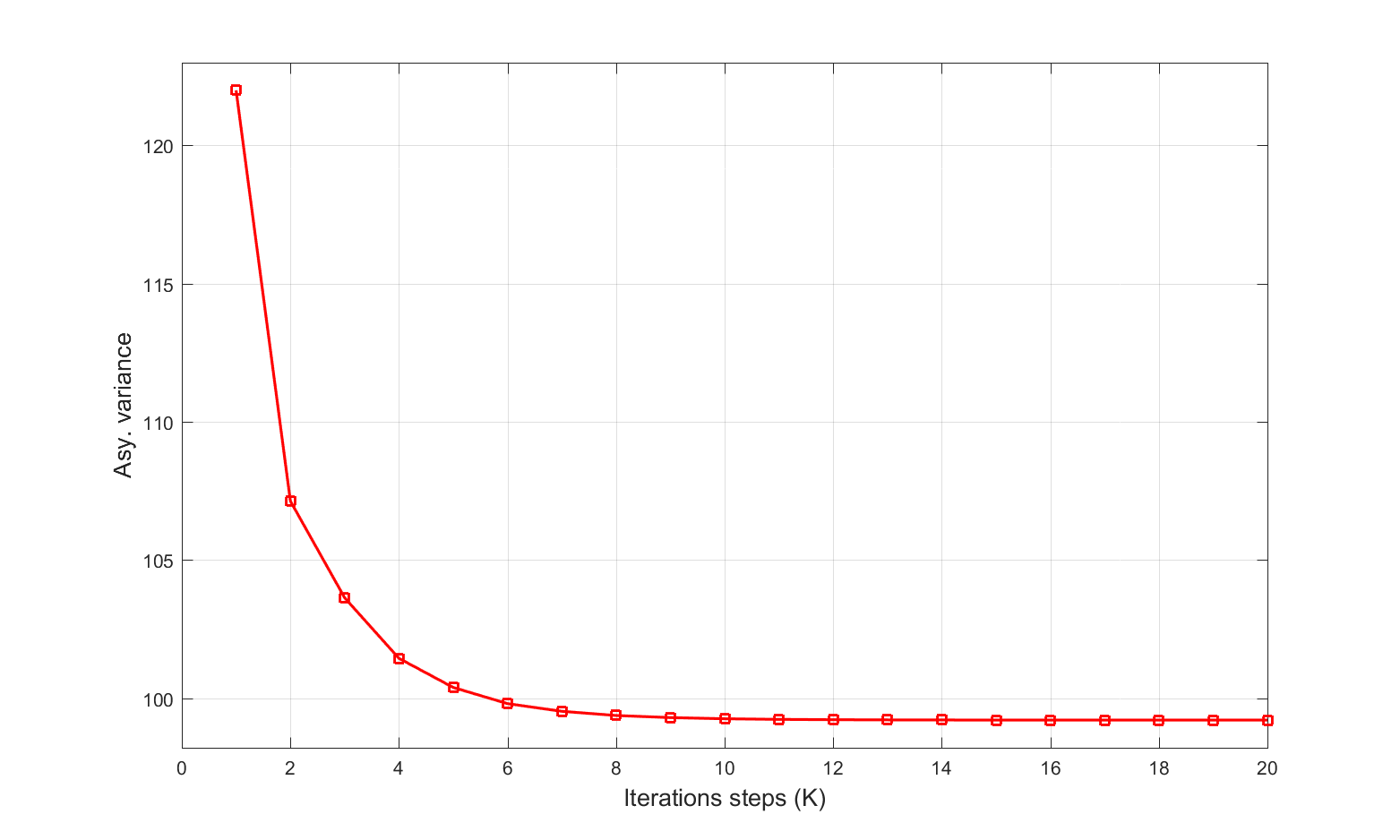}
	\caption{\small Asymptotic variance of the $K$-PML estimator of $\lambda _{RN}^{\ast }$ as a function of the number of iterations $K$ when $(\lambda _{RN}^{\ast },\lambda _{EC}^{\ast },\lambda _{RS}^{\ast },\lambda _{FC,1}^{\ast },\lambda _{FC,2}^{\ast },\beta^{\ast }) = (2.8,0.8,0.7,0.6,0.4,0.95)$.}
	\label{fig:example1}
\end{figure}

\begin{figure}[H]
	\centering
		\includegraphics[scale=0.57]{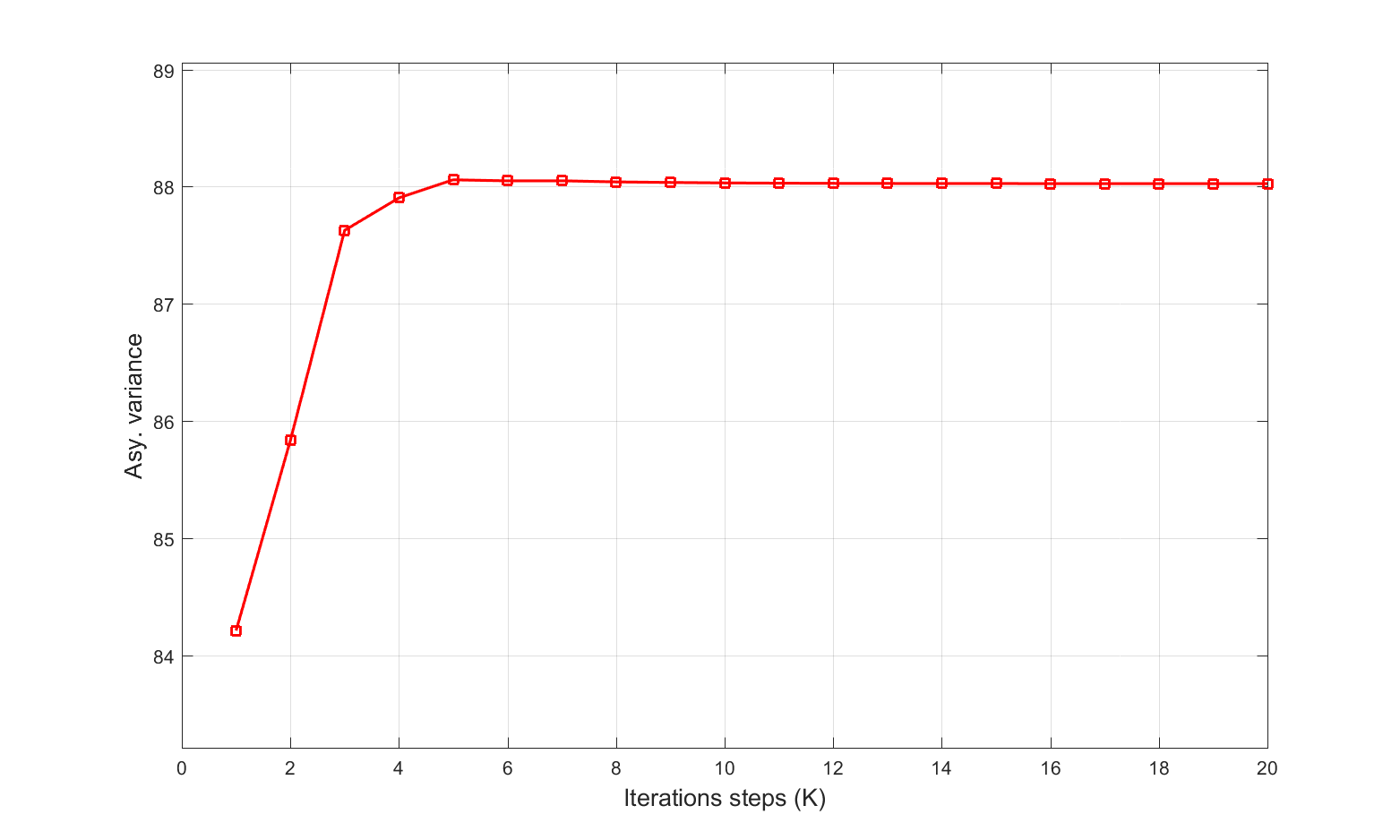}
	\caption{\small Asymptotic variance of the $K$-PML estimator of $\lambda _{RN}^{\ast }$ as a function of the number of iterations $K$ when $(\lambda _{RN}^{\ast },\lambda _{EC}^{\ast },\lambda _{RS}^{\ast },\lambda _{FC,1}^{\ast },\lambda _{FC,2}^{\ast },\beta^{\ast }) = (2,1.8,0.2,0.01,0.03,0.95)$.}
	\label{fig:example2}
\end{figure}

\begin{figure}[H]
	\centering
		\includegraphics[scale=0.57]{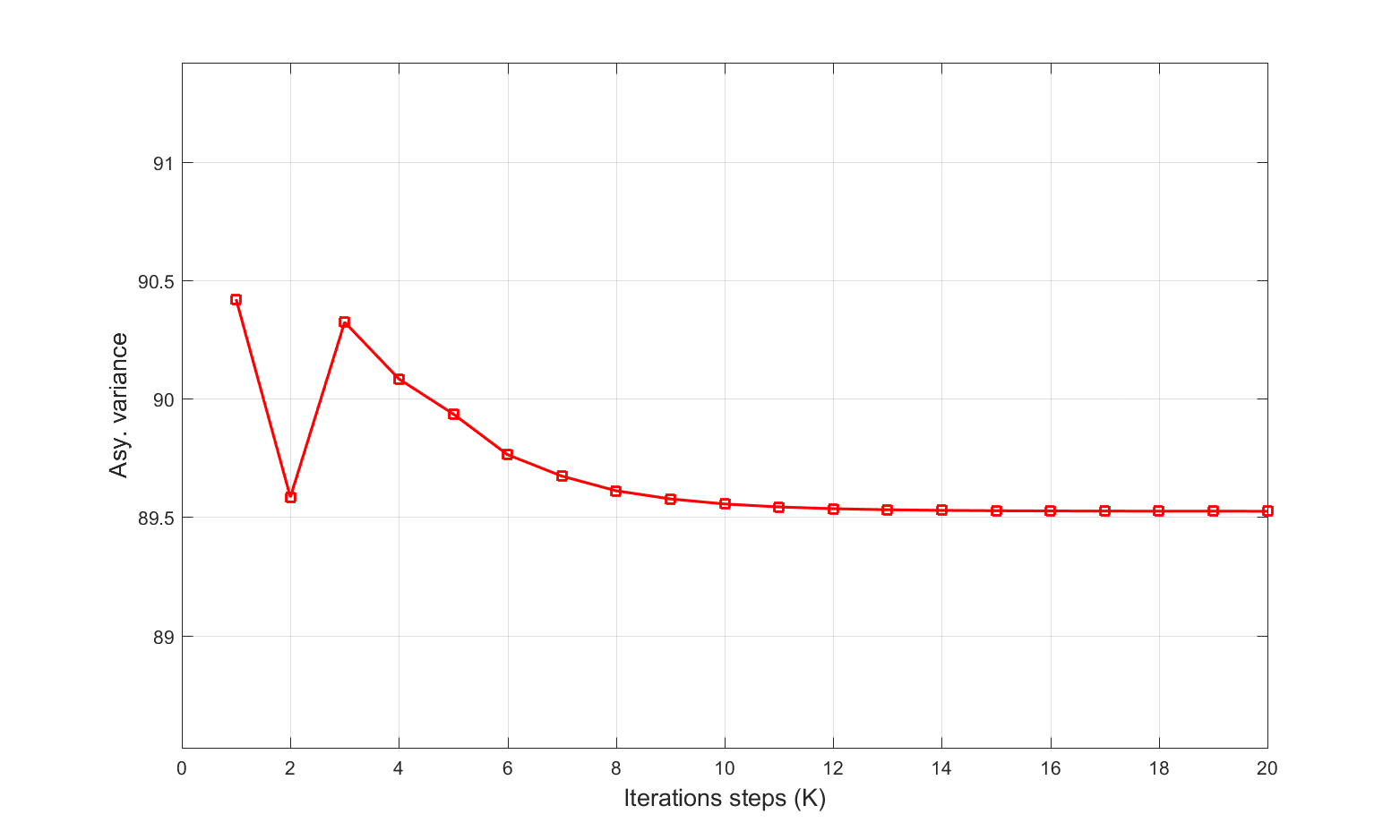}
	\caption{\small  Asymptotic variance of the $K$-PML estimator of $\lambda _{RN}^{\ast }$ as a function of the number of iterations $K$ when $(\lambda _{RN}^{\ast },\lambda _{EC}^{\ast },\lambda _{RS}^{\ast },\lambda _{FC,1}^{\ast },\lambda _{FC,2}^{\ast },\beta^{\ast }) = (2.2,1.45,0.45,0.22,0.29,0.95)$.}
	\label{fig:example3}
\end{figure}

\section{Results for $K$-MD estimation}\label{sec:MD}

In this section, we introduce a new class of $K$-stage PI estimators, referred to as the $K$-MD estimator. We demonstrate that this has several advantages over the $K$-PML estimator. In particular, we show that an optimal $K$-MD estimator dominates the $K$-PML estimator in terms of efficiency, and its asymptotic variance does not change with $K$. In addition, we show that an optimal $K$-MD estimator can be achieved with $K=1$.

For any $K\in \mathbb{N}$, the $K$-MD estimator is defined by Eq.\ \eqref{eq:K-StepDefn} with (negative) minimum distance criterion function:
\begin{equation*}
{\hat{Q}}_{K-MD}(\alpha ,g,P)~\equiv ~-( \hat{P}-\Psi ( \alpha ,g,P) )^{\prime }\hat{W}_{K}( \hat{P}-\Psi ( \alpha ,g,P) ),
\end{equation*}
and where $\{ \hat{W}_{k}:k\leq K\} $ is a sequence of positive semidefinite weight matrices. That is,
\begin{itemize}
\item \textbf{Step 1:} Estimate $( g^{\ast },P^{\ast }) $ with preliminary step estimators $( \hat{g},\hat{P}_{0}) $.
\item \textbf{Step 2:} Estimate $\alpha^{\ast }$ with $\hat{\alpha}_{K-MD}$, computed by the following algorithm. Initialize $k=1$ and then:
\begin{itemize}
\item[(a)] Compute 
\begin{equation*}
\hat{\alpha}_{k-MD}~\equiv ~\underset{\alpha \in \Theta _{\alpha }}{\arg \min }~( \hat{P}-\Psi ( \alpha ,\hat{g},\hat{P}_{k-1}) )^{\prime }\hat{W}_{k}( \hat{P}-\Psi ( \alpha ,\hat{g},\hat{ P}_{k-1}) ) .
\end{equation*}
If $k=K$, exit the algorithm. If $k<K$, go to (b).
\item[(b)] Estimate $P^{\ast }$ with the $k$-step estimator of the CCPs, given by
\begin{equation*}
\hat{P}_{k}~\equiv ~\Psi (\hat{\alpha}_{k-MD},\hat{g},\hat{P}_{k-1}).
\end{equation*}
Then, increase $k$ by one unit and return to (a).
\end{itemize}
\end{itemize}

The implementation of the $K$-MD estimator requires several choices: the number of iteration steps $K$ and the associated weight matrices $\{ \hat{W}_{k}:k\leq K\}$. We note that the sequence of weight matrices does not affect the $K$-MD estimator in the special case with $ d_\alpha = d_P$, although, as already mentioned, it is far more common for applications to have $ d_\alpha<d_P$. Also, note that the least squares estimator in \cite{pesendorfer/schmidt-dengler:2008} is a particular case of our $1$-MD estimator with $\hat{P}_{0}=\hat{P}$. In this sense, our $K$-MD estimator can be considered as an iterative version of their least squares estimator. The primary goal of this section is to study how to make optimal choices of $K$ and $\{ \hat{W}_{k}:k\leq K\}$.

To establish the asymptotic properties of the $K$-MD estimator, we add the following assumption.

\begin{assumptionA}[(Weight matrices)]\label{ass:Weight}
	For every $k\leq K$, $\hat{W}_{k}\overset{p}{\to}W_{k}$ and $W_{k} \in \mathbb{R}^{d_P \times d_P}$ is positive definite and symmetric.
\end{assumptionA}

The next result derives the asymptotic distribution of the two-step $K$-MD estimator for any $K\in \mathbb{N}$.

\begin{theorem}[Two-step $K$-MD]\label{thm:MD}
Fix $K\in \mathbb{N}$ arbitrarily and assume Assumptions \ref{ass:iid}-\ref{ass:Baseline} and \ref{ass:Weight}. Then,
\begin{equation*}
\sqrt{n}( \hat{\alpha}_{K-MD}-\alpha^{\ast }) ~\overset{d}{\to}~N( \mathbf{0}_{d_{\alpha }\times 1},\Sigma_{K-MD}( \hat{P} _{0},\{ W_{k}:k\leq K\} ) ) ,
\end{equation*}
where
\begin{equation*}
\Sigma _{K-MD}( \hat{P}_{0},\{ W_{k}:k\leq K\} ) ~\equiv~ 
\left\{
\begin{array}{c}
( \Psi _{\alpha }^{\prime }W_{K}\Psi _{\alpha })^{-1}\Psi _{\alpha }^{\prime }W_{K}\times \\
\left[ 
\begin{array}{c}
( \mathbf{I}_{d_{P}}-\Psi _{P}\Phi _{K,P})^{\prime } \\ 
-( \Psi _{P}\Phi _{K,0})^{\prime } \\ 
-\Psi _{g}^{\prime }( \mathbf{I}_{d_{P}}+\Psi _{P}\Phi _{K,g})^{\prime }
\end{array}
\right]^{\prime }\left( 
\begin{array}{ccc}
\Omega _{PP} & \Omega _{P0} & \Omega _{Pg} \\ 
\Omega _{P0}^{\prime } & \Omega _{00} & \Omega _{0g} \\ 
\Omega _{Pg}^{\prime } & \Omega _{0g}^{\prime } & \Omega _{gg}
\end{array}
\right) \left[ 
\begin{array}{c}
( \mathbf{I}_{d_{P}}-\Psi _{P}\Phi _{K,P})^{\prime } \\ 
-( \Psi _{P}\Phi _{K,0})^{\prime } \\ 
-\Psi _{g}^{\prime }( \mathbf{I}_{d_{P}}+\Psi _{P}\Phi _{K,g})^{\prime }
\end{array}
\right]  \\ 
\times W_{K}^{\prime }\Psi _{\alpha }( \Psi _{\alpha }^{\prime }W_{K}^{\prime }\Psi _{\alpha })^{-1}
\end{array}
\right\} ,
\end{equation*}
and $\{ \Phi _{k,0}:k\leq K\} $, $\{ \Phi _{k,P}:k\leq K\} $, and $\{ \Phi _{k,g}:k\leq K\} $ defined as follows. Set $\Phi _{1,P}\equiv \mathbf{0} _{d_{P}\times d_{P}}$, $\Phi _{1,0}\equiv \mathbf{I}_{d_{P}}$, $\Phi _{1,g}\equiv \mathbf{0}_{d_{P}\times d_{P}}$ and, for any $k\leq K-1$,
\begin{align}
\Phi _{k+1,P}& ~\equiv~ ( \mathbf{I}_{d_{P}}-\Psi _{\alpha }( \Psi _{\alpha }^{\prime }W_{k}\Psi _{\alpha })^{-1}\Psi _{\alpha }^{\prime }W_{k}) \Psi _{P}\Phi _{k,P}+\Psi _{\alpha }( \Psi _{\alpha }^{\prime }W_{k}\Psi _{\alpha })^{-1}\Psi _{\alpha }^{\prime }W_{k}, \notag \\
\Phi _{k+1,0}& ~\equiv~( \mathbf{I}_{d_{P}}-\Psi _{\alpha }( \Psi _{\alpha }^{\prime }W_{k}\Psi _{\alpha })^{-1}\Psi _{\alpha }^{\prime }W_{k}) \Psi _{P}\Phi _{k,0}, \notag \\
\Phi _{k+1,g}& ~\equiv~( \mathbf{I}_{d_{P}}-\Psi _{\alpha }( \Psi _{\alpha }^{\prime }W_{k}\Psi _{\alpha })^{-1}\Psi _{\alpha }^{\prime }W_{k}) ( \mathbf{I}_{d_{P}}+\Psi _{P}\Phi _{k,g}) . \label{eq:coefficients_MD}
\end{align}
\end{theorem}

We make several comments about this result. First, as in Theorem \ref{thm:ML}, Theorem \ref{thm:MD} considers $K\in \mathbb{N}$ and \textit{fixed} as $n\to \infty$, and it thus is free from the criticism raised by \cite{pesendorfer/schmidt-dengler:2010}. Second, we note that $\Sigma_{K-MD}(\hat{P}_{0},\hat{g})$ can be consistently estimated using consistent estimators of $(\alpha^*,g^*)$ (e.g., $(\hat{\alpha} _{1-MD},\hat{g})$) and the asymptotic variance in Assumption \ref{ass:Baseline}. Third, as expected, note that the sequence of weight matrices does not affect the asymptotic distribution of the $K$-MD estimator if $d_\alpha = d_P$. Finally, note that the asymptotic distribution of the $K$-PML estimator coincides with that of the $K$-MD estimator when $W_{k}\equiv \Omega _{PP}^{-1}$ for all $k \leq K$. We record this in the following corollary.

\begin{corollary}[$K$-PML is a special case of $K$-MD]\label{cor:MLandMD}
Fix $K\in \mathbb{N}$ arbitrarily and assume Assumptions \ref{ass:iid}-\ref{ass:Baseline}. The asymptotic distribution of the $K$-PML estimator is a special case of that of the $K$-MD estimator with $ W_{k}\equiv \Omega _{PP}^{-1}$ for all $k\leq K$.
\end{corollary}

Theorem \ref{thm:MD} reveals that, unless $d_\alpha = d_P$, the asymptotic variance of the $K$-MD estimator is a complicated function of the number of iteration steps $K$ and sequence of limiting weighting matrices $\{ W_{k}:k\leq K\} $. For the remainder of the section, we focus on the typical case in which $d_\alpha <d_P$. A natural question to ask is the following: Is there an optimal way of choosing these parameters? In particular, what is the optimal choice of $K$ and $ \{ W_{k}:k\leq K\} $ that minimizes the asymptotic variance of the $K$-MD estimator? We devote the rest of this section to this question.

As a first approach to this problem, we consider the non-iterative $1$-MD estimator. As shown in \cite{pesendorfer/schmidt-dengler:2008}, the asymptotic distribution of this estimator is analogous to that of a GMM estimator so we can leverage well-known optimality results. The next result provides a concrete answer regarding the optimal choices of $\hat{P}_{0}$ and $W_{1}$.

\begin{theorem}[Optimality with $K=1$] \label{thm:MD_K1}
Assume Assumptions \ref{ass:iid}-\ref{ass:Weight}. Let $\hat{\alpha}_{1-MD}^{\ast }$ denote the $1$-MD estimator with $\hat{P}_{0} = \tilde{P}$ that is asymptotically equivalent to $\hat{P}$ in the sense that
\begin{equation}
	\sqrt{n}(\tilde{P}-P^{*})~=~\sqrt{n}(\hat{P}-P^{*}) ~+~ o_{p}(1),
	\label{eq:asyequiv_P0}
\end{equation}
and $W_{1} = W_{1}^{\ast }$ with
\begin{align}
	W_{1}^{\ast }~\equiv~ 
[ ( \mathbf{I}_{d_{P}}-\Psi _{P}) \Omega _{PP}( \mathbf{I} _{d_{P}}-\Psi _{P}^{\prime }) +\Psi _{g}\Omega _{gg}\Psi _{g}^{\prime } -\Psi _{g}\Omega _{Pg}^{\prime }( \mathbf{I}_{d_{P}}-\Psi _{P}^{\prime }) -( \mathbf{I}_{d_{P}}-\Psi _{P}) \Omega _{Pg}\Psi _{g}^{\prime }]^{-1}. \label{eq:W1_optimal}
\end{align}
Then, 
\begin{equation*}
\sqrt{n}( \hat{\alpha}_{1-MD}^{\ast }-\alpha^{\ast }) ~\overset{d}{\to}~ N( \mathbf{0}_{d_{\alpha }\times 1},\Sigma^{\ast }) ,
\end{equation*}
with
\begin{align}
\Sigma^{\ast }~\equiv~ ( \Psi _{\alpha }^{\prime }[ 
( \mathbf{I}_{d_{P}}-\Psi _{P}) \Omega _{PP}( \mathbf{I} _{d_{P}}-\Psi _{P}^{\prime }) + \Psi _{g}\Omega _{gg}\Psi _{g}^{\prime } 
	-\Psi _{g}\Omega _{Pg}^{\prime }( \mathbf{I}_{d_{P}}-\Psi _{P}^{\prime }) -( \mathbf{I}_{d_{P}}-\Psi _{P}) \Omega _{Pg}\Psi _{g}^{\prime }
]^{-1}\Psi _{\alpha })^{-1}.
 \label{eq:Optimal_AVar_MD}
\end{align}
Furthermore, $\Sigma _{1-MD}( \hat{P}_{0},W_{1}) -\Sigma^{\ast }$ is positive semidefinite for all $( \hat{P}_{0},W_{1})$, i.e., $\hat{\alpha}_{1-MD}^{\ast }$ is optimal among all $1$-MD estimators that satisfy our assumptions.
\end{theorem}

Theorem \ref{thm:MD_K1} indicates that using $\hat{P}_{0}$ asymptotically equivalent to $\hat{P}$ and $W_{1}=W_{1}^{\ast }$ produces an optimal $1$-MD estimator. On the one hand, the choice of $\hat{P}_{0}$ is natural since $\hat{P}$ is an optimal preliminary estimator of the CCPs. Given this choice, the asymptotic distribution of the $1$-MD estimator is analogous to that of a standard GMM problem, and the corresponding optimal weight is $W_{1}=W_{1}^{\ast }$. Several comments are in order. First, as one would expect, $W_{1}^{\ast }$ coincides with the optimal weight matrix in the non-iterative analysis in \citet[Proposition 5]{pesendorfer/schmidt-dengler:2008}. Second, while Eq.\ \eqref{eq:asyequiv_P0} is satisfied by $\hat{P}_{0} = \hat{P}$, this condition can also be achieved by the other estimators of the CCPs mentioned in Section \ref{sec:assumptions}, provided that they are sufficiently flexible. Third, we note that $W_{1}^{\ast } \not=\Omega _{PP}^{-1}$ in general, i.e., the optimal weight matrix does not coincide the one that produces the $1$-PML estimator. In fact, the $1$-PML estimator need not be an optimal $1$-MD estimator, i.e., $\Sigma _{1-MD}( \hat{P}_{0},\Omega _{PP}^{-1}) -\Sigma^{\ast }$ can be positive definite. Finally, we note that $W_{1}^{\ast }$ can be consistently estimated by replacing each component of Eq.\ \eqref{eq:W1_optimal} by sample analogs. By standard asymptotic arguments, the optimality result in Theorem \ref{thm:MD_K1} extends to the feasible optimal $1$-MD estimator that replaces $W_{1}^{\ast }$ by any consistent estimator. In practice, however, estimating $W_{1}^{\ast }$ can introduce substantial finite sample biases, which can deteriorate the performance of the estimator. For a discussion of these issues, see \cite{altonji/segal:1996} and \cite{horowitz:1998}, among others.

We now move on to the general case with $K\geq 1$. According to Theorem \ref{thm:MD}, the asymptotic variance of the $K$-MD estimator depends on the number of iteration steps $K \in \mathbb{N}$, the asymptotic distribution of $\hat{P}_{0}$, and the entire sequence of limiting weight matrices $\{ W_{k}:k\leq K\}$. In this sense, determining an optimal $K$-MD estimator appears to be a complicated task. The next result provides a concrete answer to this problem.

\begin{theorem}[Invariance and optimality]\label{thm:MD_main} 
Fix $K\in \mathbb{N}$ arbitrarily and assume Assumptions \ref{ass:iid}-\ref{ass:Weight}. Then,
\begin{enumerate}
\item \underline{Invariance.} Let $\hat{\alpha}_{K-MD}^{\ast }$ denote the $K$-MD estimator with $\hat{P}_{0} = \tilde{P}$ that is asymptotically equivalent to $\hat{P}$ in the sense of Eq.\ \eqref{eq:asyequiv_P0}, weight matrices $\{ W_{k}:k\leq K-1\} $ for steps $1,\ldots ,K-1$ (if $K>1$), and the corresponding optimal weight matrix in step $K$, which we assume to be well-defined. Then,
\begin{equation*}
\sqrt{n}( \hat{\alpha}_{K-MD}^{\ast }-\alpha^{\ast }) ~\overset{d}{\to}~N( \mathbf{0}_{d_{\alpha }\times 1},\Sigma^{\ast }) ,
\end{equation*}
where $\Sigma^{\ast }$ is as in Eq.\ \eqref{eq:Optimal_AVar_MD}.
\item \underline{Optimality.} Let $\hat{\alpha}_{K-MD}$ denote the $K$-MD estimator with $\hat{P}_{0}$ and weight matrices $\{ W_{k}:k\leq K\} $. Then,
\begin{equation*}
\sqrt{n}( \hat{\alpha}_{K-MD}-\alpha^{\ast }) ~\overset{d}{\to}~N( \mathbf{0}_{d_{\alpha }\times 1},\Sigma_{K-MD}( \hat{P} _{0},\{ W_{k}:k\leq K\} ) ).
\end{equation*}
Furthermore, $\Sigma _{K-MD}(\hat{P}_{0},\{ W_{k}:k\leq K\})-\Sigma^{\ast }$ is positive semidefinite, i.e., $\hat{\alpha}_{K-MD}^{\ast }$ is optimal among all $K$-MD estimators that satisfy our assumptions.
\end{enumerate}
\end{theorem}

Theorem \ref{thm:MD_main} is the main finding of this paper, and it establishes two central results regarding the optimality of the $K$-MD estimator. We begin by discussing the first one, referred to as ``invariance''. The result focuses on a preliminary estimator of the CCPs that is asymptotically equivalent to $\hat{P}$. As explained earlier, this is a natural choice to consider since $\hat{P}$ is an optimal preliminary estimator of the CCPs. Given this choice, the asymptotic variance of the $K$-MD estimator depends on the entire sequence of weight matrices $ \{W_{k}:k\leq K\}$. While the dependence on the first $K-1$ weight matrices is fairly complicated, the dependence on the last weight matrix (i.e., $W_{K}$) resembles that of the weight matrix in a standard GMM problem. Standard GMM results characterize the optimal choice for $W_{K}$, {\it given the sequence of first $K-1$ weight matrices}.\footnote{The result requires this optimal choice for $W_{K}$ to be well defined.  For typical choices of the first $K-1$ weight matrices, this additional requirement was not restrictive in our Monte Carlo simulations.} In principle, one might expect that the resulting asymptotic variance depends on the first $K-1$ weight matrices. The ``invariance'' result reveals that this is not the case. In other words, for $\hat{P}_{0}=\hat{P}$ and an optimal choice of $W_{K}$, the asymptotic distribution of the $K$-MD estimator is invariant to the first $K-1$ weight matrices, or even $K$. Furthermore, the resulting asymptotic distribution coincides with that of the optimal $1$-MD estimator obtained in Theorem \ref{thm:MD_K1}.

The ``invariance'' result is the key to the second result in Theorem \ref{thm:MD_main}, referred to as ``optimality''. This second result characterizes the optimal choice of $\hat{P}_{0}$ and $ \{W_{k}:k\leq K\}$ for $K$-MD estimators. The intuition of the result is as follows. First, since $\hat{P}$ is the optimal estimator of CCPs, it is intuitive that optimality requires this choice or possibly something asymptotically equivalent. Second, it is also intuitive that optimality requires an optimal choice of $W_K$, {\it given the sequence of first $K-1$ weight matrices}. At this point, our ``invariance'' result indicates that the asymptotic distribution does not depend on $K$ or the first $K-1$ weight matrices. From this, we can then conclude that the $K$-MD estimator with $ \hat{P}_{0}$ asymptotically equivalent to $\hat{P}$ and an optimal last weight matrix $W_{K}$ (given any first $K-1$ weight matrices) is efficient among all $K$-MD estimators. 

Theorem \ref{thm:MD_main} implies two important corollaries regarding the optimal $1$-MD estimator discussed in Theorem \ref{thm:MD_K1}. The first corollary is that the optimal $1$-MD estimator is efficient in the class of all $K$-MD estimators that satisfy the assumptions of Theorem \ref{thm:MD_main}. In other words, additional policy iterations do not provide efficiency gains relative to the optimal $1$-MD estimator. The intuition behind this result is that the multiple iteration steps of the $K$-MD estimator are merely reprocessing the sample information, i.e., no new information is added in each iteration step. Provided that the criterion function is optimally weighted, the $1$-MD estimator is capable of processing the sample information in an efficient manner.

The second corollary of Theorem \ref{thm:MD_main} is that the $K$-PML estimator is \textit{usually} not efficient, and can be feasibly improved upon. In particular, under the assumptions of Theorem \ref{thm:MD_main}, the optimal $1$-MD estimator (i.e.\ with $\hat{P}_{0}$ asymptotically equivalent to $\hat{P}$ and $ W_{1}=W_{1}^{\ast }$) is more or equally efficient than the $K$-PML estimator.\footnote{One special case in which the $K$-PML estimator is as efficient as the optimal $1$-MD estimator is when both: (a) the zero Jacobian property holds (i.e.\ $\Psi_P = {\bf 0}_{dP \times dP}$) and (b) Step 1 of the estimation procedure does not include a preliminary estimator of $g^*$ (i.e.\ $\theta^* = \alpha^*$). It is worthwhile to point out that (a) or (b) taken in isolation would not be sufficient to imply that the $K$-PML estimator is as efficient as the optimal $1$-MD estimator.}

The results up to this point show that the optimal $1$-MD estimator is efficient among all $K$-MD estimators. Given these findings, one might wonder whether the optimal $1$-MD estimator is efficient. Section \ref{sec:MLE} in the appendix compares the asymptotic distribution of the optimal $1$-MD estimator with that of the MLE of $\alpha^*$. Under appropriate conditions, we show that these two asymptotic distributions coincide, and so the optimal $1$-MD estimator is indeed efficient among all regular estimators. 

We illustrate the results of this section by revising the example of Section \ref{sec:example} with the parameter values considered in Section \ref{sec:ML}. The asymptotic variance of $\hat{\alpha}_{K-MD}=(\hat{\lambda}_{RN,K-MD},\hat{ \lambda}_{EC,K-MD})$ is given by
\begin{equation}
\Sigma _{K-MD}(\hat{P}_{0},\{W_{k}:k\leq K\}) = (\Psi _{\alpha }^{\prime }W_{K}\Psi _{\alpha })^{-1}\Psi _{\alpha }^{\prime }W_{K}( \mathbf{I}_{d_{P}}-\Psi _{P}\Phi _{K,P0})\Omega _{PP}(\mathbf{I} _{d_{P}}-\Psi _{P}\Phi _{K,P})^{\prime }W_{K}\Psi _{\alpha }(\Psi _{\alpha }^{\prime }W_{K}\Psi _{\alpha })^{-1},
 \label{eq:Sigma_MD_example}
\end{equation}
where $\{\Phi _{k,P0}:k\leq K\}$ is defined by $\Phi _{k,P0}\equiv \Phi _{k,P}+\Phi _{k,0}$, with $\{\Phi _{k,P}:k\leq K\}$ and $\{\Phi _{k,0}:k\leq K\}$ as in Eq.\ \eqref{eq:coefficients_MD}. For any true parameter vector and any $K\in \mathbb{N} $, we can numerically compute Eq.\ \eqref{eq:Sigma_MD_example}.

In Section \ref{sec:ML}, we considered three specific parameter values of $(\lambda _{RN}^{\ast },\lambda _{EC}^{\ast })$ that produced an asymptotic variance of the $K$-PML estimator of $\lambda_{RN}^{\ast }$ decreased, increased, and fluctuated with $K$. We now compute the optimal $K$-MD estimator of $\lambda_{RN}^*$ for the same parameter values. The results are presented in Figures \ref{fig:example1B}, \ref{fig:example2B}, and \ref{fig:example3B}, respectively.\footnote{According to the ``invariance'' result in Theorem \ref{thm:MD_main}, there are multiple asymptotically equivalent ways of implementing the optimal $K$-MD estimator. For concreteness, we set the weight matrix optimally in each iteration step.} These graphs illustrate the findings in Theorem \ref{thm:MD_main}. In accordance to the ``invariance'' result, the asymptotic variance of the optimal $K$-MD estimator does not vary with the number of iterations $K$. Also in accordance to the ``optimality'' result, the asymptotic variance of the optimal $K$-MD estimator is lower than that of any other $K$-MD estimator. In turn, since the asymptotic distribution of the $K$-PML estimator is a special case of the asymptotic distribution of the $K$-MD estimator, the asymptotic variance of the optimal $K$-MD estimator is lower than that of the $K$-PML estimator for all $K \in \mathbb{N}$. Combining both results, the optimal non-iterative $1$-MD estimator is both computationally convenient and efficient among the estimators under consideration.

\begin{figure}[H]
	\centering
		\includegraphics[scale=0.57]{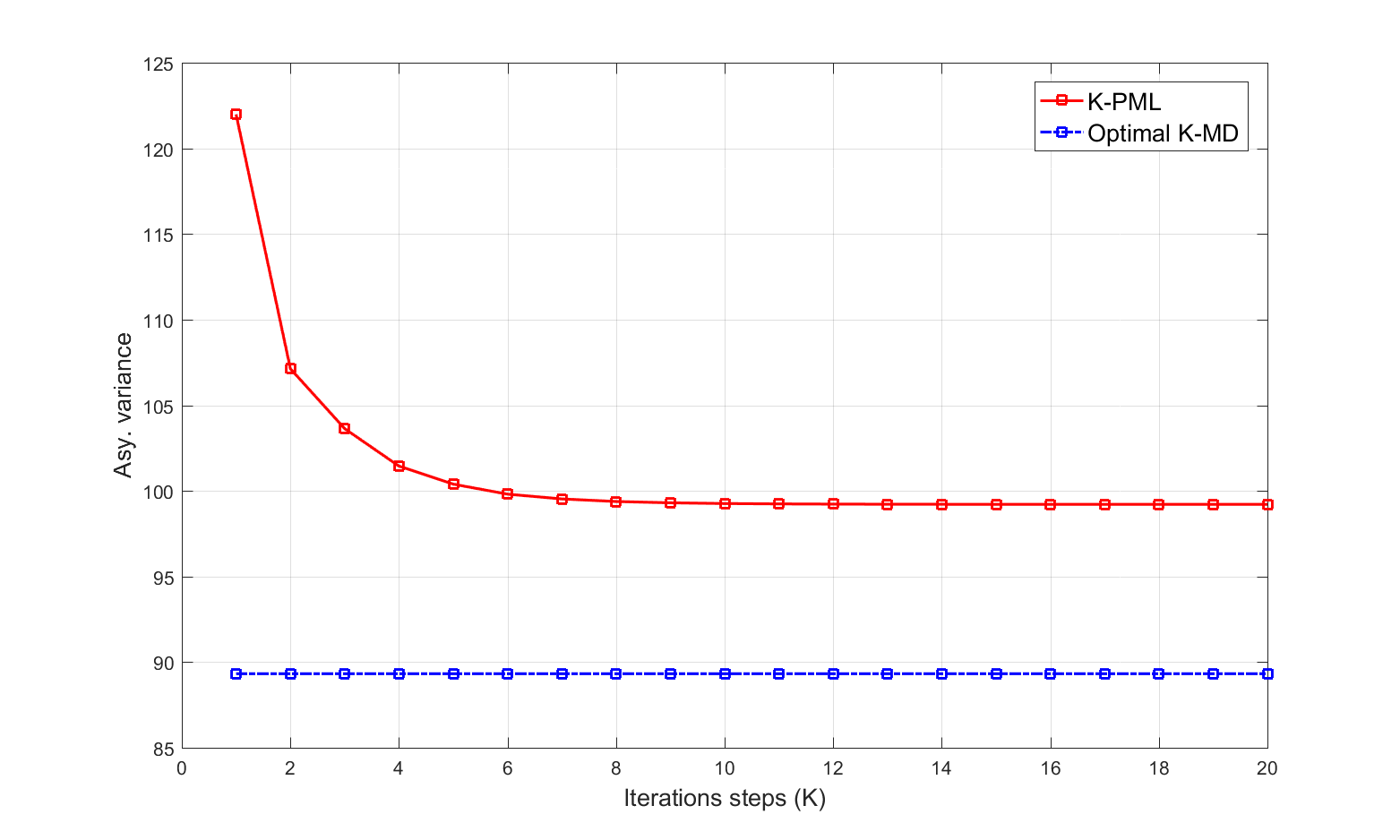}
	\caption{\small  Asymptotic variance of the $K$-PML estimator and optimal $K$-MD estimator of $\lambda _{RN}^{\ast }$ as a function of the number of iterations $K$ when $(\lambda _{RN}^{\ast },\lambda _{EC}^{\ast },\lambda _{RS}^{\ast },\lambda _{FC,1}^{\ast },\lambda _{FC,2}^{\ast },\beta^{\ast }) = (2.8,0.8,0.7,0.6,0.4,0.95)$. The optimal $K$-MD estimator is computed using the optimal weighting matrix in every iteration step.}
	\label{fig:example1B}
\end{figure}

\begin{figure}[H]
	\centering
		\includegraphics[scale=0.57]{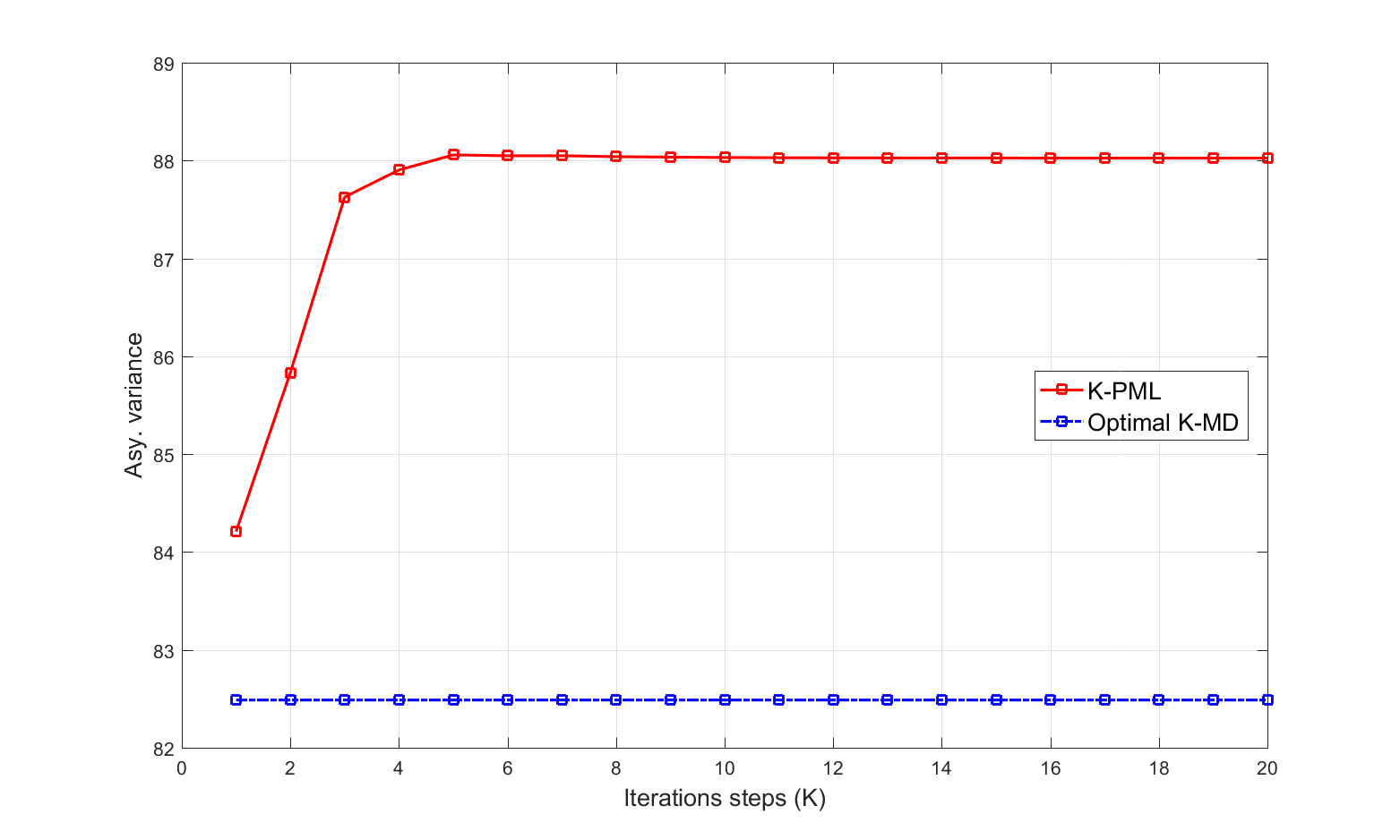}
	\caption{\small Asymptotic variance of the $K$-PML estimator and optimal $K$-MD estimator of $\lambda _{RN}^{\ast }$ as a function of the number of iterations $K$ when $(\lambda _{RN}^{\ast },\lambda _{EC}^{\ast },\lambda _{RS}^{\ast },\lambda _{FC,1}^{\ast },\lambda _{FC,2}^{\ast },\beta^{\ast }) = (2,1.8,0.2,0.01,0.03,0.95)$. The optimal $K$-MD estimator is computed using the optimal weighting matrix in every iteration step.}
	\label{fig:example2B}
\end{figure}

\begin{figure}[H]
	\centering
		\includegraphics[scale=0.57]{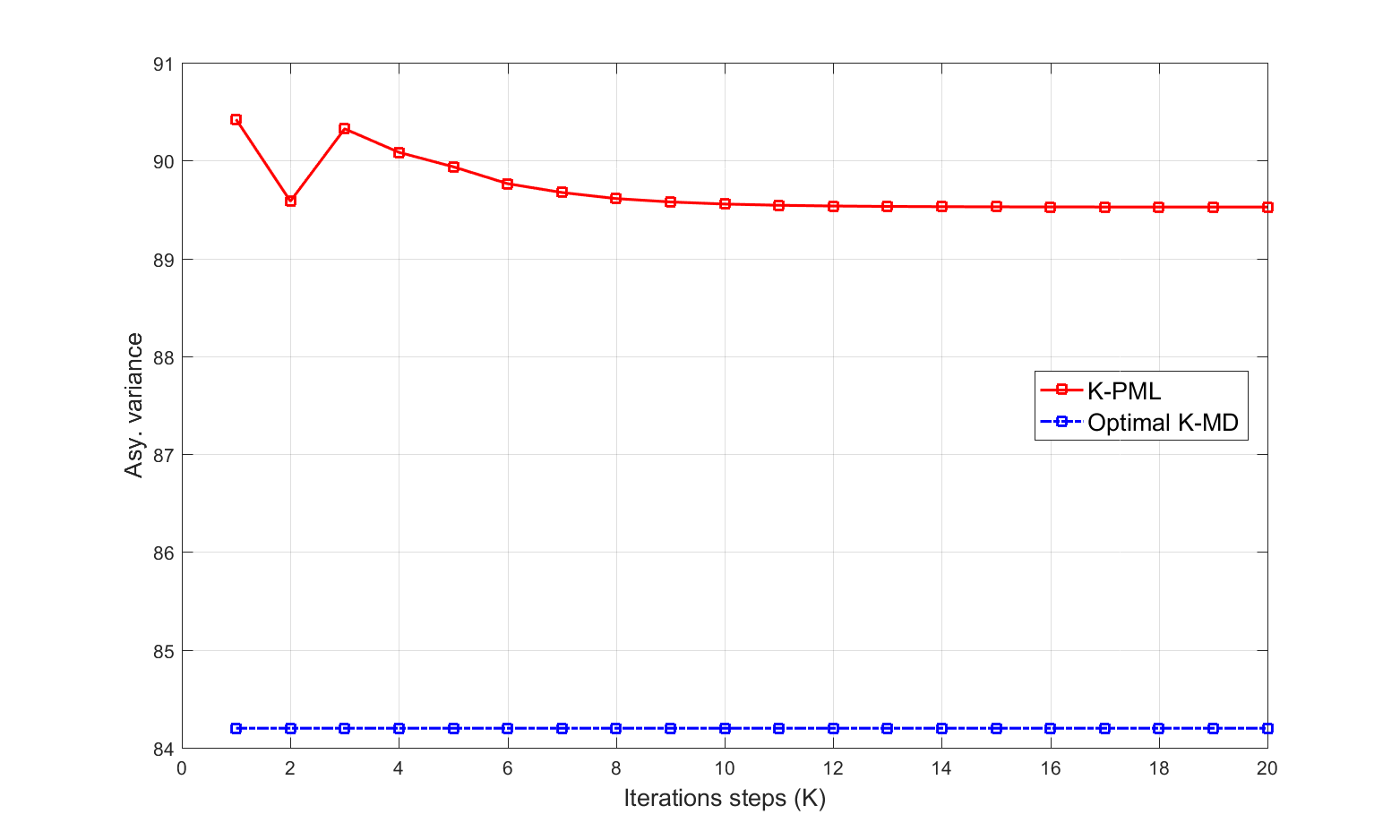}
	\caption{\small  Asymptotic variance of the $K$-PML estimator and optimal $K$-MD estimator of $\lambda _{RN}^{\ast }$ as a function of the number of iterations $K$ when $(\lambda _{RN}^{\ast },\lambda _{EC}^{\ast },\lambda _{RS}^{\ast },\lambda _{FC,1}^{\ast },\lambda _{FC,2}^{\ast },\beta^{\ast }) = (2.2,1.45,0.45,0.22,0.29,0.95)$. The optimal $K$-MD estimator is computed using the optimal weighting matrix in every iteration step.}
	\label{fig:example3B}
\end{figure}

\section{Monte Carlo simulations}\label{sec:MCsimulations}

This section investigates the performance of $K$-PML and $K$-MD estimators in a Monte Carlo simulation. Our main goal in this section is to confirm that our asymptotic analysis can provide a reasonable approximation provided that the size of the sample is sufficiently large. We simulate data using the two-player dynamic entry game described in Section \ref{sec:example}. Recall that this model is specified up to the parameters $(\lambda _{RS}^{\ast },\lambda _{RN}^{\ast },\lambda _{FC,1}^{\ast },\lambda _{FC,2}^{\ast },\lambda _{EC}^{\ast },\beta^{\ast })$. For simplicity, we assume that the researcher knows $(\lambda _{RS}^{\ast },\lambda _{FC,1}^{\ast },\lambda _{FC,2}^{\ast },\beta^{\ast })$ and wants to estimate $\alpha^{\ast }\equiv (\lambda _{RN}^{\ast },\lambda _{EC}^{\ast })$. We consider the three specific parameter values used to illustrate the theoretical results in Sections \ref{sec:ML} and \ref{sec:MD}. For each parameter value, we compute equilibrium CCPs numerically by solving the fixed point problem in Eq.\ \eqref{eq:FP} up to a small tolerance level.\footnote{We did not encounter evidence of multiple equilibria in our computations.} Derivatives of $\Psi$ were computed numerically, and we note that $\Psi_P$ differs from the zero matrix for all parameter values (i.e., the Zero Jacobian property does not hold).

Our simulation results are the average of $S=10,000$ independent datasets $\{ ( \{ a_{j,i}:j \in J\} ,x_{i},x_{i}^{\prime }) :i\leq n\} $ that are i.i.d.\ distributed according to the econometric model. We show results for sample sizes $n \in \{500,~1,000,~2,000\}$. Compared to typical empirical applications, these sample sizes are admittedly large relative to the state space  (e.g., the empirical application in \citet[Section 5]{aguirregabiria/mira:2007} has $n=189$ and $d_P=800$), but this choice reflects our asymptotic framework in which the state space is fixed as the number of observations diverges. With these sample sizes, we can verify our findings regarding the behavior of the asymptotic variance with the number of iterations. On the other hand, for smaller sample sizes, we often find that our asymptotic approximation does not provide a good representation of the distribution of these estimators, possibly due to the large finite sample bias in the preliminary estimator of the CCPs.

For the sake of brevity, we show simulation results for the estimation of $\lambda _{RN}^{\ast }$, which was the object of illustration in Sections \ref{sec:ML} and \ref{sec:MD}. Qualitatively similar results hold for the estimation of $\lambda _{EC}^{\ast }$, and are available upon request. We use $\hat{P}_0 = \hat{P}$ as the preliminary estimator (recall that there is no $g^*$ in this example). 

Table \ref{tab:MC1} provides results for the first parameter value, i.e., $(\lambda _{RN}^{\ast },\lambda _{EC}^{\ast },\lambda _{RS}^{\ast },\lambda _{FC,1}^{\ast },\lambda _{FC,2}^{\ast },\beta^{\ast }) = (2.8,0.8,0.7,0.6,0.4,0.95)$. Recall from previous sections that this parameter value produces an asymptotic variance of the $K$-PML estimator of $\lambda_{RN}^{\ast }$ that decreases with $K$ (see Figures \ref{fig:example1} and \ref{fig:example1B}, repeated at the bottom of Table \ref{tab:MC1}). 

Let us first focus on the results for the $K$-PML estimator. The simulation results closely resemble the predictions from the asymptotic approximation. As mentioned earlier, this is an expected consequence of using sample sizes that are large relative to the dimensionality of the model. First, the empirical variance and mean squared error are extremely close, indicating that the asymptotic bias is almost negligible. Second, the empirical variance is decreasing with $K$ and is close to the one predicted by our asymptotic analysis. Finally, the computational cost of the $K$-PML estimator is low relative to the optimal $K$-MD estimator, and rises linearly with $K$.

Next, we turn attention to the optimal $K$-MD estimator. Recall that the ``invariance'' result in Theorem \ref{thm:MD_main} indicates that there are multiple asymptotically equivalent ways of implementing the optimal $K$-MD estimator. Throughout this section, the optimal $K$-MD estimator is a {\it feasible} estimator of the optimal $K$-MD estimator derived in Theorem \ref{thm:MD_main} in the last iteration step and the $k$-PML weight matrix in steps $k=1,\dots,K-1$. According to our theoretical results, this feasible optimal $K$-MD estimator is optimal among $K$-MD estimators, has zero asymptotic bias, and has an asymptotic variance that does not change with $K$. For the most part, these predictions are satisfied in our simulations. Once again, this is a consequence of using sample sizes that are large relative to the dimensionality of the model. First, we note that the empirical variance and mean squared error are again extremely close, and so the finite-sample bias is almost negligible. Second, the empirical variance is close to the one predicted by our asymptotic analysis. As predicted by the ``optimality'' result in Theorem \ref{thm:MD_main}, the feasible optimal $K$-MD estimator is more efficient than the $K$-PML estimator. For most values of $K$ under consideration, the empirical variance of the $K$-MD estimator appears to be invariant to $K$, especially for the larger sample sizes. However, we find that the empirical variance decreases slightly between $K=1$ and $K=2$. Our first-order asymptotic analysis cannot explain this last empirical finding. This anomalous behavior for low values of $K$ is analogous to the one found for the $K$-PML estimator by \cite{aguirregabiria/mira:2002} and rationalized by the higher-order analysis in \cite{kasahara/shimotsu:2008}. In Section \ref{sec:highorder}, we show that a high-order analysis can explain these anomalous simulation results for the $K$-MD estimator. Finally, the computational cost of the optimal $K$-MD estimator is considerably higher than the $K$-PML estimator. In particular, computing the optimally weighted $1$-MD, $2$-MD, and $3$-MD estimators takes us roughly $33\%$, $75\%$, and $80\%$ more time than computing the $20$-PML estimator, respectively. The reason behind this difference is that the $K$-PML estimator does not require estimating an optimal weight matrix, while the optimal $K$-MD estimator does. This optimal weight matrix for the $K$-MD estimator requires estimating $\Psi_P$ and $\Psi_\alpha$ (the latter only when $K>1$). As noted by an anonymous referee, the computational cost of estimating $\Psi_P$ in large models can be significant, as its dimension grows quadratically with $d_P$.

\begin{table}[h]
	\begin{center}
	\scalebox{0.79}{\begin{tabular}{cc|rrrrrrrr}
	  \hline\hline
\multicolumn{1}{c}{Estimator} & \multicolumn{1}{c|}{Statistic} & \multicolumn{1}{c}{$K=1$} & \multicolumn{1}{c}{$K=2$} & \multicolumn{1}{c}{$K=3$} &
\multicolumn{1}{c}{$K=4$} & \multicolumn{1}{c}{$K=5$} & \multicolumn{1}{c}{$K=10$} & \multicolumn{1}{c}{$K=15$} & \multicolumn{1}{c}{$K=20$} \\
\hline\hline \multicolumn{2}{c}{}  & \multicolumn{8}{c}{$n = 500$}\\   \hline
\multirow{3}[1]{*}{$K$-PML}  & \multicolumn{1}{c|}{Var} & 127.30 & 109.82 & 107.50 & 103.37 & 101.68 &  98.18 &  97.91 &  97.88 \\ 
   & \multicolumn{1}{c|}{MSE} & 127.37 & 110.31 & 108.20 & 103.71 & 101.93 &  98.29 &  98.02 &  98.00 \\ 
   & \multicolumn{1}{c|}{Time} &   0.85 &   1.70 &   2.50 &   3.29 &   4.09 &   8.08 &  12.10 &  16.13 \\ 
  \multirow{3}[1]{*}{Opt.\ $K$-MD} & \multicolumn{1}{c|}{Var} & 106.24 &  92.93 &  90.23 &  87.36 &  86.88 &  86.56 &  86.46 &  86.44 \\ 
   & \multicolumn{1}{c|}{MSE} & 107.16 &  93.56 &  90.23 &  87.37 &  86.94 &  86.65 &  86.55 &  86.54 \\ 
   & \multicolumn{1}{c|}{Time} &  21.50 &  27.47 &  28.40 &  29.28 &  30.06 &  34.29 &  38.40 &  42.68 \\ 
   \hline\hline \multicolumn{2}{c}{}  & \multicolumn{8}{c}{$n = 1,000$} \\  \hline
\multirow{3}[1]{*}{$K$-PML}  & \multicolumn{1}{c|}{Var} & 122.35 & 106.41 & 103.53 & 100.45 &  99.08 &  96.96 &  96.84 &  96.83 \\ 
   & \multicolumn{1}{c|}{MSE} & 122.36 & 106.61 & 103.83 & 100.58 &  99.18 &  97.00 &  96.88 &  96.87 \\ 
   & \multicolumn{1}{c|}{Time} &   0.75 &   1.51 &   2.25 &   2.98 &   3.70 &   7.41 &  11.16 &  14.92 \\ 
  \multirow{3}[1]{*}{Opt.\ $K$-MD} & \multicolumn{1}{c|}{Var} &  90.19 &  88.04 &  86.77 &  86.93 &  86.59 &  86.43 &  86.40 &  86.40 \\ 
   & \multicolumn{1}{c|}{MSE} &  90.42 &  88.34 &  86.77 &  86.94 &  86.64 &  86.50 &  86.47 &  86.47 \\ 
   & \multicolumn{1}{c|}{Time} &  20.27 &  25.95 &  26.78 &  27.47 &  28.21 &  32.21 &  36.11 &  40.10 \\ 
   \hline\hline \multicolumn{2}{c}{}  & \multicolumn{8}{c}{$n = 2,000$} \\  \hline
\multirow{3}[1]{*}{$K$-PML}  & \multicolumn{1}{c|}{Var} & 123.43 & 108.62 & 105.55 & 102.96 & 101.76 & 100.06 &  99.98 &  99.97 \\ 
   & \multicolumn{1}{c|}{MSE} & 123.53 & 108.93 & 105.96 & 103.22 & 101.98 & 100.22 & 100.12 & 100.12 \\ 
   & \multicolumn{1}{c|}{Time} &   0.78 &   1.54 &   2.28 &   3.03 &   3.78 &   7.50 &  11.29 &  15.05 \\ 
  \multirow{3}[1]{*}{Opt.\ $K$-MD} & \multicolumn{1}{c|}{Var} &  91.43 &  91.04 &  90.34 &  90.49 &  90.31 &  90.24 &  90.23 &  90.23 \\ 
   & \multicolumn{1}{c|}{MSE} &  91.68 &  91.41 &  90.37 &  90.51 &  90.31 &  90.24 &  90.23 &  90.23 \\ 
   & \multicolumn{1}{c|}{Time} &  20.79 &  26.66 &  27.41 &  28.23 &  29.01 &  32.90 &  36.92 &  40.84 \\ 
   \hline\hline \multicolumn{2}{c}{}  & \multicolumn{8}{c}{Asymptotic results} \\  \hline
$K$-PML  & \multicolumn{1}{c|}{Asy.\ Var} & 121.98 & 107.13 & 103.63 & 101.44 & 100.39 &  99.26 &  99.21 &  99.21 \\ 
  Opt.\ $K$-MD & \multicolumn{1}{c|}{Asy.\ Var} &  89.33 &  89.33 &  89.33 &  89.33 &  89.33 &  89.33 &  89.33 &  89.33 \\ 
\hline\hline
\end{tabular}}
	\end{center}
	\caption{\small Simulation results for estimation of $\lambda _{RN}^{\ast }$ when $(\lambda _{RN}^{\ast },\lambda _{EC}^{\ast },\lambda _{RS}^{\ast },\lambda _{FC,1}^{\ast },\lambda _{FC,2}^{\ast },\beta^{\ast }) = (2.8,0.8,0.7,0.6,0.4,0.95)$. ``$K$-PML'' denotes the $K$-PML estimator and ``Opt.\ $K$-MD'' denotes the feasible optimal $K$-MD estimator computed with an estimated optimal weight matrix in the last iteration step and $k$-PML weight matrix in steps $k<K$. ``Var'' denotes the average empirical variance scaled by $n$, ``MSE'' denotes the average mean squared error scaled by $n$, and ``Time'' denotes the average time to compute the estimator in milliseconds, and, for all of these, the average is computed over $S=10,000$ simulations. ``Asy.\ Var'' denotes the asymptotic variance according to Theorems \ref{thm:ML} and \ref{thm:MD}.}
	\label{tab:MC1}
\end{table}

Table \ref{tab:MC2} provides results for the second parameter value, i.e., $(\lambda _{RN}^{\ast },\lambda _{EC}^{\ast },\lambda _{RS}^{\ast },\lambda _{FC,1}^{\ast },\lambda _{FC,2}^{\ast },\beta^{\ast }) = (2,1.8,0.2,0.01,0.03,0.95)$. Recall that this parameter value produced an asymptotic variance of the $K$-PML estimator of $\lambda_{RN}^{\ast }$ that increases with $K$ (see Figures \ref{fig:example2} and \ref{fig:example2B}, repeated at the bottom of Table \ref{tab:MC2}). 
In turn, Table \ref{tab:MC3} provides results for the third parameter value, i.e., $(\lambda _{RN}^{\ast },\lambda _{EC}^{\ast },\lambda _{RS}^{\ast },\lambda _{FC,1}^{\ast },\lambda _{FC,2}^{\ast },\beta^{\ast }) = (2.2,1.45,0.45,0.22,0.29,0.95)$. This parameter value produced an asymptotic variance of the $K$-PML estimator of $\lambda_{RN}^{\ast }$ that wiggles with $K$ (see Figures \ref{fig:example3} and \ref{fig:example3B}, repeated at the bottom of Table \ref{tab:MC3}).   

The simulation results for these two parameter values are qualitatively similar to the ones obtained for the first parameter value and, for the most part, support our theoretical conclusions. First, both estimators have very little empirical bias. Second, all the estimators have an empirical variance that is very close to the one predicted by the asymptotic analysis. In particular, the empirical variance of the $K$-PML estimator is increasing in $K$ for the second parameter value and wiggles for the third parameter value. Third, in most cases, the empirical variance of the optimal $K$-MD estimator is lower than that of the $K$-PML estimator. Fourth, the empirical variance of the optimal $K$-MD estimator is invariant to $K$ except for small values of $K$ for which it is decreasing. One notable difference relative to the previous simulation is that the range of iterations over which the empirical variance decreases now extends between $K=1$ and $K=5$. As explained earlier, we attribute this to the high-order analysis that we develop in Section \ref{sec:highorder}. Finally, the comparison of computational costs is similar to the one described for the first parameter value.

\begin{table}[h]
	\begin{center}
	\scalebox{0.79}{\begin{tabular}{cc|rrrrrrrr}
	  \hline\hline
\multicolumn{1}{c}{Estimator} & \multicolumn{1}{c|}{Statistic} & \multicolumn{1}{c}{$K=1$} & \multicolumn{1}{c}{$K=2$} & \multicolumn{1}{c}{$K=3$} &
\multicolumn{1}{c}{$K=4$} & \multicolumn{1}{c}{$K=5$} & \multicolumn{1}{c}{$K=10$} & \multicolumn{1}{c}{$K=15$} & \multicolumn{1}{c}{$K=20$} \\
\hline\hline \multicolumn{2}{c}{}  & \multicolumn{8}{c}{$n = 500$}\\   \hline
\multirow{3}[1]{*}{$K$-PML}  & \multicolumn{1}{c|}{Var} &  87.34 &  89.22 &  93.80 &  93.78 &  94.53 &  93.88 &  93.29 &  92.94 \\ 
   & \multicolumn{1}{c|}{MSE} &  87.40 &  89.82 &  94.98 &  94.78 &  95.52 &  94.73 &  94.11 &  93.74 \\ 
   & \multicolumn{1}{c|}{Time} &   0.79 &   1.57 &   2.32 &   3.06 &   3.79 &   7.50 &  11.23 &  14.98 \\ 
  \multirow{3}[1]{*}{Opt.\ $K$-MD} & \multicolumn{1}{c|}{Var} & 111.01 & 105.28 & 102.97 &  99.81 &  97.91 &  99.80 &  84.82 &  84.58 \\ 
   & \multicolumn{1}{c|}{MSE} & 116.28 & 109.86 & 104.58 & 100.88 &  98.51 &  99.97 &  84.90 &  84.65 \\ 
   & \multicolumn{1}{c|}{Time} &  20.10 &  25.90 &  26.74 &  27.44 &  28.18 &  32.18 &  35.95 &  39.89 \\ 
   \hline\hline \multicolumn{2}{c}{}  & \multicolumn{8}{c}{$n = 1,000$} \\  \hline
\multirow{3}[1]{*}{$K$-PML}  & \multicolumn{1}{c|}{Var} &  85.11 &  86.92 &  90.26 &  90.45 &  90.86 &  90.49 &  90.28 &  90.21 \\ 
   & \multicolumn{1}{c|}{MSE} &  85.13 &  87.18 &  90.80 &  90.90 &  91.30 &  90.87 &  90.64 &  90.57 \\ 
   & \multicolumn{1}{c|}{Time} &   0.74 &   1.51 &   2.26 &   3.01 &   3.75 &   7.52 &  11.26 &  15.03 \\ 
  \multirow{3}[1]{*}{Opt.\ $K$-MD} & \multicolumn{1}{c|}{Var} &  98.02 &  94.01 &  89.55 &  88.28 &  86.38 &  83.27 &  82.67 &  82.54 \\ 
   & \multicolumn{1}{c|}{MSE} & 101.04 &  96.22 &  90.16 &  88.67 &  86.57 &  83.32 &  82.70 &  82.57 \\ 
   & \multicolumn{1}{c|}{Time} &  20.69 &  26.51 &  27.32 &  28.09 &  28.90 &  32.85 &  36.76 &  40.70 \\ 
   \hline\hline \multicolumn{2}{c}{}  & \multicolumn{8}{c}{$n = 2,000$} \\  \hline
\multirow{3}[1]{*}{$K$-PML}  & \multicolumn{1}{c|}{Var} &  87.32 &  89.53 &  92.37 &  92.61 &  92.93 &  92.71 &  92.62 &  92.60 \\ 
   & \multicolumn{1}{c|}{MSE} &  87.37 &  89.76 &  92.76 &  92.95 &  93.26 &  93.00 &  92.91 &  92.89 \\ 
   & \multicolumn{1}{c|}{Time} &   0.72 &   1.46 &   2.18 &   2.90 &   3.63 &   7.23 &  10.86 &  14.52 \\ 
  \multirow{3}[1]{*}{Opt.\ $K$-MD} & \multicolumn{1}{c|}{Var} &  93.83 &  90.62 &  87.91 &  87.29 &  86.43 &  85.19 &  84.95 &  84.89 \\ 
   & \multicolumn{1}{c|}{MSE} &  95.84 &  91.93 &  88.30 &  87.54 &  86.57 &  85.25 &  85.00 &  84.94 \\ 
   & \multicolumn{1}{c|}{Time} &  20.30 &  26.22 &  26.74 &  27.46 &  28.31 &  32.13 &  35.89 &  39.83 \\ 
   \hline\hline \multicolumn{2}{c}{}  & \multicolumn{8}{c}{Asymptotic results} \\  \hline
$K$-PML  & \multicolumn{1}{c|}{Asy.\ Var} &  84.21 &  85.83 &  87.63 &  87.90 &  88.06 &  88.03 &  88.03 &  88.03 \\ 
  Opt.\ $K$-MD & \multicolumn{1}{c|}{Asy.\ Var} &  82.49 &  82.49 &  82.49 &  82.49 &  82.49 &  82.49 &  82.49 &  82.49 \\ 
\hline\hline
\end{tabular}}
	\end{center}
	\caption{\small Simulation results for estimation of $\lambda _{RN}^{\ast }$ when $(\lambda _{RN}^{\ast },\lambda _{EC}^{\ast },\lambda _{RS}^{\ast },\lambda _{FC,1}^{\ast },\lambda _{FC,2}^{\ast },\beta^{\ast }) = (2,1.8,0.2,0.01,0.03,0.95)$. ``$K$-PML'' denotes the $K$-PML estimator and ``Opt.\ $K$-MD'' denotes the feasible optimal $K$-MD estimator computed with an estimated optimal weight matrix in the last iteration step and $k$-PML weight matrix in steps $k<K$. ``Var'' denotes the average empirical variance scaled by $n$, ``MSE'' denotes the average mean squared error scaled by $n$, and ``Time'' denotes the average time to compute the estimator in milliseconds, and, for all of these, the average is computed over $S=10,000$ simulations. ``Asy.\ Var'' denotes the asymptotic variance according to Theorems \ref{thm:ML} and \ref{thm:MD}.}
	\label{tab:MC2}
\end{table}

\begin{table}[h]
	\begin{center}
	\scalebox{0.79}{
	\begin{tabular}{cc|rrrrrrrr}
  \hline\hline
\multicolumn{1}{c}{Estimator} & \multicolumn{1}{c|}{Statistic} & \multicolumn{1}{c}{$K=1$} & \multicolumn{1}{c}{$K=2$} & \multicolumn{1}{c}{$K=3$} &
\multicolumn{1}{c}{$K=4$} & \multicolumn{1}{c}{$K=5$} & \multicolumn{1}{c}{$K=10$} & \multicolumn{1}{c}{$K=15$} & \multicolumn{1}{c}{$K=20$} \\
\hline\hline \multicolumn{2}{c}{}  & \multicolumn{8}{c}{$n = 500$}\\   \hline
\multirow{3}[1]{*}{$K$-PML}  & \multicolumn{1}{c|}{Var} &  93.36 &  91.73 &  94.38 &  93.21 &  92.98 &  90.57 &  89.96 &  89.83 \\ 
   & \multicolumn{1}{c|}{MSE} &  93.44 &  92.38 &  95.52 &  94.08 &  93.80 &  91.17 &  90.52 &  90.39 \\ 
   & \multicolumn{1}{c|}{Time} &   0.80 &   1.60 &   2.35 &   3.11 &   3.86 &   7.62 &  11.39 &  15.17 \\ 
  \multirow{3}[1]{*}{Opt.\ $K$-MD} & \multicolumn{1}{c|}{Var} & 108.23 &  94.76 &  86.76 &  85.07 &  83.05 &  82.17 &  81.93 &  81.88 \\ 
   & \multicolumn{1}{c|}{MSE} & 112.64 &  97.53 &  87.20 &  85.32 &  83.14 &  82.20 &  81.95 &  81.90 \\ 
   & \multicolumn{1}{c|}{Time} &  20.43 &  26.19 &  27.03 &  27.88 &  28.59 &  32.56 &  36.46 &  40.56 \\ 
   \hline\hline \multicolumn{2}{c}{}  & \multicolumn{8}{c}{$n = 1,000$} \\  \hline
\multirow{3}[1]{*}{$K$-PML}  & \multicolumn{1}{c|}{Var} &  91.29 &  90.52 &  92.23 &  91.71 &  91.51 &  90.16 &  89.87 &  89.82 \\ 
   & \multicolumn{1}{c|}{MSE} &  91.38 &  90.99 &  92.97 &  92.31 &  92.07 &  90.61 &  90.31 &  90.26 \\ 
   & \multicolumn{1}{c|}{Time} &   0.73 &   1.46 &   2.20 &   2.91 &   3.65 &   7.33 &  10.99 &  14.65 \\ 
  \multirow{3}[1]{*}{Opt.\ $K$-MD} & \multicolumn{1}{c|}{Var} &  95.17 &  88.14 &  84.62 &  84.39 &  83.62 &  83.08 &  82.96 &  82.93 \\ 
   & \multicolumn{1}{c|}{MSE} &  97.66 &  89.78 &  84.97 &  84.62 &  83.74 &  83.14 &  83.01 &  82.99 \\ 
   & \multicolumn{1}{c|}{Time} &  20.17 &  25.83 &  26.52 &  27.34 &  28.12 &  31.98 &  35.81 &  39.69 \\ 
   \hline\hline \multicolumn{2}{c}{}  & \multicolumn{8}{c}{$n = 2,000$} \\  \hline
\multirow{3}[1]{*}{$K$-PML}  & \multicolumn{1}{c|}{Var} &  92.59 &  92.02 &  93.27 &  92.85 &  92.65 &  91.72 &  91.57 &  91.56 \\ 
   & \multicolumn{1}{c|}{MSE} &  92.60 &  92.11 &  93.45 &  92.97 &  92.76 &  91.79 &  91.64 &  91.63 \\ 
   & \multicolumn{1}{c|}{Time} &   0.64 &   1.31 &   1.97 &   2.61 &   3.27 &   6.54 &   9.83 &  13.13 \\ 
  \multirow{3}[1]{*}{Opt.\ $K$-MD} & \multicolumn{1}{c|}{Var} &  91.62 &  87.79 &  85.66 &  85.48 &  84.96 &  84.50 &  84.42 &  84.41 \\ 
   & \multicolumn{1}{c|}{MSE} &  92.56 &  88.29 &  85.71 &  85.50 &  84.96 &  84.50 &  84.42 &  84.41 \\ 
   & \multicolumn{1}{c|}{Time} &  18.36 &  23.51 &  24.15 &  24.86 &  25.62 &  28.96 &  32.42 &  35.93 \\ 
   \hline\hline \multicolumn{2}{c}{}  & \multicolumn{8}{c}{Asymptotic results} \\  \hline
$K$-PML  & \multicolumn{1}{c|}{Asy.\ Var} &  90.42 &  89.58 &  90.32 &  90.08 &  89.94 &  89.56 &  89.53 &  89.52 \\ 
  Opt.\ $K$-MD & \multicolumn{1}{c|}{Asy.\ Var} &  84.20 &  84.20 &  84.20 &  84.20 &  84.20 &  84.20 &  84.20 &  84.20 \\ 
\hline\hline
\end{tabular}}
	\end{center}
	\caption{\small Simulation results for estimation of $\lambda _{RN}^{\ast }$ when $(\lambda _{RN}^{\ast },\lambda _{EC}^{\ast },\lambda _{RS}^{\ast },\lambda _{FC,1}^{\ast },\lambda _{FC,2}^{\ast },\beta^{\ast }) =(2.2,1.45,0.45,0.22,0.29,0.95)$. ``$K$-PML'' denotes the $K$-PML estimator and ``Opt.\ $K$-MD'' denotes the feasible optimal $K$-MD estimator computed with an estimated optimal weight matrix in the last iteration step and $k$-PML weight matrix in steps $k<K$. ``Var'' denotes the average empirical variance scaled by $n$, ``MSE'' denotes the average mean squared error scaled by $n$, and ``Time'' denotes the average time to compute the estimator in milliseconds, and, for all of these, the average is computed over $S=10,000$ simulations. ``Asy.\ Var'' denotes the asymptotic variance according to Theorems \ref{thm:ML} and \ref{thm:MD}.}
	\label{tab:MC3}
\end{table}

For the sake of comparison, we also compute the one-step MLE described in \citet[Section 3.6]{aguirregabiria/mira:2007}, which does not belong to the class of $K$-PML or $K$-MD estimators. The one-step MLE is the result of taking a Newton step in the MLE problem based on an initial estimator of $(\alpha^*,P^*)$ that is consistent and is a fixed point in Eq.\ \eqref{eq:FP}. This initial estimator serves as a starting point of the Newton step, and is also used to consistently estimate the efficient score and information matrix. \citet[Section 3.6]{aguirregabiria/mira:2007} explain that the one-step MLE has the advantage of being as efficient as the MLE. They propose implementing the one-step MLE with an initial estimator given by their NPL estimator, i.e., the $\infty$-PML estimator. As an approximation to this, we compute the one-step MLE with the $20$-PML estimator as the initial estimator.\footnote{We do not use $K=\infty$ for reasons explained in Section \ref{sec:Introduction}, and 20 is the largest value of $K$ considered in our simulations.} The simulation results for the one-step MLE are provided in Table \ref{tab:newton}. These show that the one-step MLE is more efficient than the $K$-PML estimator, and is as efficient as optimal $K$-MD estimator, especially when $K\geq 10$. These findings are in line with our analysis in Section \ref{sec:MLE}, which reveals that, under appropriate conditions, the optimal $K$-MD estimator is as efficient as the MLE and, consequently, as efficient as the one-step MLE. Finally, the computational costs of estimating the one-step MLE are similar to that of an optimally weighted $20$-MD estimator.

\begin{table}
	\begin{center}
	\scalebox{0.79}{\begin{tabular}{rr|rrr}\hline\hline
~~~~Estimator~~~~ & ~~~~Statistic~~~~ & \multicolumn{1}{c}{~~$n = 500$~~} & \multicolumn{1}{c}{~~$n = 1,000$~~} & \multicolumn{1}{c}{~~$n = 2,000$~~} \\
                            \hline\hline \multicolumn{5}{c}{Design 1: $(\lambda _{RN}^{\ast },\lambda _{EC}^{\ast },\lambda _{RS}^{\ast },\lambda _{FC,1}^{\ast },\lambda _{FC,2}^{\ast },\beta^{\ast })
=(2.8,0.8,0.7,0.6,0.4,0.95)$} \\   \hline
\multicolumn{1}{c}{\multirow{3}[1]{*}{$1$-step MLE}}  & \multicolumn{1}{c|}{Var} & 86.27 & 86.29 & 90.18 \\ 
   & \multicolumn{1}{c|}{MSE} & 86.33 & 86.34 & 90.18 \\ 
   & \multicolumn{1}{c|}{Time} & 40.86 & 38.41 & 39.12 \\ 
   \hline\hline \multicolumn{5}{c}{Design 2: $(\lambda _{RN}^{\ast },\lambda _{EC}^{\ast },\lambda _{RS}^{\ast },\lambda _{FC,1}^{\ast },\lambda _{FC,2}^{\ast },\beta^{\ast })
                        =(2,1.8,0.2,0.01,0.03,0.95)$} \\  \hline
\multicolumn{1}{c}{\multirow{3}[1]{*}{$1$-step MLE}}  & \multicolumn{1}{c|}{Var} & 84.95 & 82.65 & 85.06 \\ 
   & \multicolumn{1}{c|}{MSE} & 85.12 & 82.73 & 85.15 \\ 
   & \multicolumn{1}{c|}{Time} & 38.38 & 39.08 & 38.05 \\ 
   \hline\hline \multicolumn{5}{c}{Design 3: $(\lambda _{RN}^{\ast },\lambda _{EC}^{\ast },\lambda _{RS}^{\ast },\lambda _{FC,1}^{\ast },\lambda _{FC,2}^{\ast },\beta^{\ast })
                        =(2.2,1.45,0.45,0.22,0.29,0.95)$} \\  \hline
\multicolumn{1}{c}{\multirow{3}[1]{*}{$1$-step MLE}}  & \multicolumn{1}{c|}{Var} & 81.84 & 82.97 & 84.49 \\ 
   & \multicolumn{1}{c|}{MSE} & 81.89 & 83.07 & 84.49 \\ 
   & \multicolumn{1}{c|}{Time} & 38.82 & 38.06 & 34.45 \\ 
   \hline\hline
	\end{tabular}}
	\end{center}
	\caption{\small Simulation results for the one-step MLE of $\lambda _{RN}^{\ast }$ for three values of $(\lambda _{RN}^{\ast },\lambda _{EC}^{\ast },\lambda _{RS}^{\ast },\lambda _{FC,1}^{\ast },\lambda _{FC,2}^{\ast },\beta^{\ast })$. This estimator is computed based on the $20$-PML estimator as the initial estimator. ``Var'' denotes the average empirical variance scaled by $n$, ``MSE'' denotes the average mean squared error scaled by $n$, and ``Time'' denotes the average time to compute the estimator in milliseconds, and, for all of these, the average is computed over $S=10,000$ simulations.}
	\label{tab:newton}
\end{table}

\section{Conclusions}\label{sec:Conclusions}

This paper investigates the asymptotic properties of a class of estimators of the structural parameters in dynamic discrete choice games. We consider $K$-stage policy iteration (PI) estimators, where $K$ denotes the number of policy iterations employed in the estimation. This class nests several estimators proposed in the literature. By considering a ``maximum likelihood'' criterion function, the $K$-stage PI estimator becomes the $K$-PML estimator in \citet{aguirregabiria/mira:2002,aguirregabiria/mira:2007}. By considering a ``minimum distance'' criterion function, $K$-stage PI estimator defines a new $K$-MD estimator, which is an iterative version of the estimators in \cite{pesendorfer/schmidt-dengler:2008} and \cite{pakes/ostrovsky/berry:2007}. Since we consider an asymptotic framework with fixed $K \in \mathbb{N}$ as $n\to\infty$, our analysis is not affected by the problems described in \cite{pesendorfer/schmidt-dengler:2010}.

First, we establish that the $K$-PML estimator is consistent and asymptotically normal for any $K \in \mathbb{N}$. This complements findings in \cite{aguirregabiria/mira:2007}, who focus on $K=1$ and $K$ large enough to induce convergence of the estimator. Furthermore, we show under certain conditions that the asymptotic variance of the $K$-PML estimator can exhibit arbitrary patterns as a function of $K$. In particular, we show that by changing the parameter values in a typical dynamic discrete choice game, the asymptotic variance of the $K$-PML estimator can increase, decrease, or even be non-monotonic with $K$.

Second, we also establish that the $K$-MD estimator is consistent and asymptotically normal for any $K$. Its asymptotic distribution depends on the choice of the weight matrix. For a specific weight matrix, the $K$-MD estimator has the same asymptotic distribution as the $K$-PML estimator. We investigate the optimal choice of the weight matrix for the $K$-MD estimator. Our main result shows that an optimally weighted $K$-MD estimator has an asymptotic distribution that is invariant to $K$. This appears to be a novel result in the literature on PI estimation for games, and it is particularly surprising given the findings in \cite{aguirregabiria/mira:2007} for $K$-PML estimators.

The main result in our paper implies two important corollaries regarding the optimal $1$-MD estimator (derived by \cite{pesendorfer/schmidt-dengler:2008}). First, the optimal $1$-MD estimator is optimal among all $K$-MD estimators. In other words, additional policy iterations do not provide efficiency gains relative to the optimal $1$-MD estimator. Second, the optimal $1$-MD estimator is more or equally efficient than any $K$-PML estimator for all $K \in \mathbb{N}$. Finally, Section \ref{sec:MLE} provides appropriate conditions under which the optimal $1$-MD estimator has the same asymptotic distribution as the MLE, and it is thus efficient among regular estimators.

We explored our theoretical findings in Monte Carlo simulations. Provided that the sample size is large enough, the simulation evidence supports the conclusions of our asymptotic analysis. The $K$-PML and the optimal $K$-MD estimators have negligible empirical bias and have an empirical variance that is very close to the one predicted by the asymptotic analysis. In most cases, the empirical variance of the optimal $K$-MD estimator is lower than that of the $K$-PML estimator. Also, it appears to be invariant to $K$ except for very small values of $K$ for which it is decreasing in $K$. The behavior for low values of $K$ is analogous to the one found by \cite{aguirregabiria/mira:2002}. Inspired by the analysis in \cite{kasahara/shimotsu:2008}, Section \ref{sec:highorder} studies the high-order properties of the optimal $K$-MD estimator and rationalizes the simulation result for low values of $K$.

There are several topics excluded from this paper, such as allowing for permanent unobserved heterogeneity in the dynamic discrete choice model. In principle, this could be achieved via the developments in \citet[Section 3.5]{aguirregabiria/mira:2007} or \cite{arcidiacono/miller:2011}. We plan to address this topic in future work.

\appendix

\begin{small}
\section{Appendix}

%%%%%%%% DIVIDER %%%%%%%%%%%%

Throughout this appendix, ``s.t.'' abbreviates ``such that'', ``RHS'' abbreviates ``right hand side'', and ``PSD'' abbreviates ``positive semidefinite''.

The asymptotic distribution of the various estimators considered in this paper follows from applying an iterated version of a general result for extremum estimators, derived in Theorem \ref{thm:2stepB}. In turn, this result requires the following high-level assumption. Note that whenever this assumption is used to prove results in the main text, we first verify that it holds under the lower-level conditions.

\begin{assumptionA}[(High-level assumptions for iterated extremum estimators).]\label{ass:EE_HL}
There is a sequence of limiting criterion functions $ \{ Q_{k}:k\leq K\} $ with $Q_{k}:\Theta_{\alpha }\times \Theta_{g}\times \Theta_{P}\to\mathbb{R}$ such that:

\begin{enumerate}[(a)]
\item $\sup_{\alpha \in \Theta_{\alpha }}|{\hat{Q}}_{k}(\alpha ,\tilde{g}, \tilde{P})-Q_{k}(\alpha ,g^{\ast },P^{\ast })|=o_{p}(1)$, provided that $ ( \tilde{g},\tilde{P}) =( g^{\ast },P^{\ast }) +o_{p}(1)$.

\item $Q_{k}(\alpha ,g^{\ast },P^{\ast })$ is uniquely maximized at $\alpha ^{\ast }$.

\item $\sqrt{n}{\partial \hat{Q}_{k}( \alpha ^{\ast },g^{\ast },P^{\ast }) }/{\partial \alpha }=\Xi_{k}\sqrt{{n}}( \hat{P}-P^{\ast }) + o_{p}( 1)$, for some matrix $\Xi_k$.

\item For any $\lambda \in \{\alpha ,g,P\}$, $\partial ^{2}{\hat{Q}}_{k}( \tilde{\alpha},\tilde{g},\tilde{P})/\partial \alpha \partial \lambda ^{\prime }=\partial ^{2}Q_{k}(\alpha ^{\ast },g^{\ast },P^{\ast })/\partial \alpha \partial \lambda ^{\prime }+o_{p}(1)$, provided that $( \tilde{ \alpha},\tilde{g},\tilde{P}) =( \alpha ^{\ast },g^{\ast },P^{\ast }) +o_{p}(1)$.

\item $\partial ^{2}Q_{k}(\alpha ^{\ast },g^{\ast },P^{\ast })/\partial \alpha \partial \alpha ^{\prime }$ is non-singular.
\end{enumerate}
\end{assumptionA}

%%%%%%%% DIVIDER %%%%%%%%%%%%

\subsection{Proofs of results in the main text}

\begin{proof}[Proof of Theorem \ref{thm:ML}]
	This proof will require the following notation. For any $( \alpha ,g,P) \in \Theta_{\alpha }\times \Theta_{g}\times \Theta_{P}$ and $(a,j,x) \in A \times J\times X$, let $n_{ajx}\equiv \sum_{i=1}^{n}1[ ( a_{jt,i},x_{t,i}) =( a,x) ] $, $n_{x}\equiv \sum_{i=1}^{n}1[ x_{t,i}=x ] $, $\Psi_{ajx}( \alpha ,g,P) \equiv \prod_{j \in J}\Psi_{j}( \alpha ,g,P) ( a|x) $, $\hat{P}_{ajx}\equiv n_{ajx}/n_{x}$, $P_{ajx}^{\ast }\equiv P_{j}^{\ast }( a|x) =\sum_{x^{\prime }\in X}\Pi_{j}^{\ast }(a,x,x^{\prime })/\sum_{(\tilde{a},\tilde{x}^{\prime })\in A\times X}\Pi_{j}^{\ast }(\tilde{a},x,\tilde{x}^{\prime })$, and $ m^{\ast }( x) \equiv \sum_{(a,x^{\prime })\in A\times X}\Pi_{j}^{\ast }(a,x,x^{\prime })$. 
	Note that RHS of the last equation does not change with  $j \in J$ by the equilibrium assumption in Assumption \ref{ass:iid}. Also, Assumptions \ref{ass:Identification} and \ref{ass:Regularity} imply $P_{ajx}^{\ast }>0$ and $m^{\ast }( x)>0$ for every $( a,j,x) \in A\times J\times X$.

Theorem \ref{thm:ML} is a consequence of applying Theorem \ref{thm:2stepB} with $\hat{Q}_{k}\equiv {\hat{Q}}_{PML}$ and ${Q}_{k}\equiv {Q}_{PML}$ where, for any $ ( \alpha ,g,P) \in \Theta_{\alpha }\times \Theta_{g}\times \Theta_{P}$,
\begin{align*}
{\hat{Q}}_{PML}( \alpha ,g,P) &~\equiv ~\frac{1}{n} \sum_{i=1}^{n}\ln \Psi ( \alpha ,g,P) ( a_{i}|x_{i}) ~=~\sum_{( a,j,x) \in A\times J\times X}\frac{n_{ajx}}{n}\ln \Psi_{ajx}( \alpha ,g,P) \\
&~=~\sum_{( j,x) \in J\times X}\frac{n_{x}}{n}\left[ \sum_{a\in \tilde{A}}\hat{P}_{ajx}\ln \Psi_{ajx}( \alpha ,g,P) +\hat{P}_{0jx}\ln ( 1-\sum_{a\in \tilde{A}}\Psi_{ajx}( \alpha ,g,P) ) \right] ,
\end{align*}
and
\begin{align*}
{Q}_{PML}( \alpha ,g,P) &~\equiv~ \sum_{( a,j,x) \in A\times J\times X}m^{\ast }( x) P_{ajx}^{\ast }\ln \Psi_{ajx}( \alpha ,g,P) \\
&~=~\sum_{( j,x) \in J\times X} m^{\ast }( x) \left[ \sum_{a\in \tilde{A}}P_{ajx}^{\ast }\ln \Psi_{ajx}( \alpha ,g,P) +P_{0jx}^{\ast }\ln ( 1-\sum_{a\in \tilde{A}}\Psi_{ajx}( \alpha ,g,P) ) \right] ,
\end{align*}
and where we have used that $\Psi_{0jx}( \alpha ,g,P) =1-\sum_{a\in \tilde{A}}\Psi_{ajx}( \alpha ,g,P) $. 

For any $\lambda \in \{ \alpha ,g,P\} $, note that
\begin{align*}
\frac{\partial ^{2}{Q}_{PML}( \alpha ^{\ast },g^{\ast },P^{\ast }) }{\partial \alpha \partial \lambda ^{\prime }}
&=-\sum_{( j,x) \in J\times X} m^{\ast }( x) \left[ 
\begin{array}{c}
\frac{1}{\Psi_{ajx}( \alpha ^{\ast },g^{\ast },P^{\ast }) } \sum_{a\in \tilde{A}}\frac{\partial \Psi_{ajx}( \alpha ^{\ast },g^{\ast },P^{\ast }) }{\partial \alpha }\frac{\partial \Psi _{ajx}( \alpha ^{\ast },g^{\ast },P^{\ast }) }{\partial \lambda ^{\prime }}\\
+ \frac{1}{\Psi_{0jx}( \alpha ^{\ast },g^{\ast },P^{\ast }) }\sum_{ \check{a}\in \tilde{A}}\frac{\partial \Psi_{\check{a}jx}( \alpha ^{\ast },g^{\ast },P^{\ast }) }{\partial \alpha }\sum_{\tilde{a}\in \tilde{A}}\frac{\partial \Psi_{\tilde{a}jx}( \alpha ^{\ast },g^{\ast },P^{\ast }) }{\partial \lambda ^{\prime }}
\end{array}
\right] \\
&=-\left\{ 
\begin{array}{c}
\{ \frac{\partial \Psi_{ajx}( \alpha ^{\ast },g^{\ast },P^{\ast }) }{\partial \alpha }:( a,j,x) \in \tilde{A}\times J\times X\} ^{\prime }\times \\
diag\{ m^{\ast }( x) ( diag\{ 1/P_{ajx}^{\ast }:a\in \tilde{A}\} +\mathbf{1}_{\vert \tilde{A}\vert \times \vert \tilde{A}\vert }/P_{0jx}^{\ast }) :( j,x) \in J\times X\} \times \\
\{ \frac{\partial \Psi_{ajx}( \alpha ^{\ast },g^{\ast },P^{\ast }) }{\partial \lambda }:( a,j,x) \in \tilde{A}\times J\times X\}
\end{array}
\right\} =-\Psi_{\alpha }^{\prime }\Omega_{PP}^{-1}\Psi_{\lambda },
\end{align*}
where the first equality uses that Assumptions \ref{ass:Identification} and \ref{ass:Regularity}, the second equality follows from Assumption \ref{ass:Identification}, and the final equality follows from the following argument. By Eq.\ \eqref{eq:SampleFrequency}, $\Omega _{PP}\equiv diag\{\Sigma _{jx}:(j,x)\in J\times X\}$, and $ \Sigma _{xj}\equiv (diag\{P_{jx}^{\ast }\}-P_{jx}^{\ast }P_{jx}^{\ast \prime })/m^{\ast }(x)$ and $P_{jx}^{\ast }\equiv \{P_{j}^{\ast }(a|x):a\in \tilde{A} \} $ for all $(j,x)\in J\times X$. From here, we deduce that
\begin{align}
	\Omega_{PP}^{-1}~=~diag\{ m^{\ast }( x) ( diag\{ 1/P_{ajx}^{\ast }:a\in \tilde{A}\} + \mathbf{1}_{|\tilde{A}| \times |\tilde{A}| }/P_{0jx}^{\ast }) ~:~( j,x) \in J\times X\}.
	\label{eq:InvOmegaPP}
\end{align}

To apply Theorem \ref{thm:2stepB}, we first verify Assumption \ref{ass:EE_HL}.

\underline{Part (a).} For any $( \tilde{g},\tilde{P}) =( g^{\ast },P^{\ast }) +o_{p}(1)$,
\begin{align*}
\sup_{\alpha \in \Theta_{\alpha }}\vert {\hat{Q}}_{PML}(\alpha , \tilde{g},\tilde{P})-{Q}_{PML}(\alpha ,g^{\ast },P^{\ast })\vert &\leq \left[
\begin{array}{c}
\sup_{\alpha \in \Theta_{\alpha }}\vert \sum_{( a,j,x) \in A\times J\times X}\frac{n_{ajx}}{n}\ln ( \Psi_{ajx}( \alpha , \tilde{g},\tilde{P}) /\Psi_{ajx}( \alpha ,g^{\ast },P^{\ast }) ) \vert \\
+\sup_{\alpha \in \Theta_{\alpha }}\vert \sum_{( a,j,x) \in A\times J\times X}( \frac{n_{ajx}}{n}-m^{\ast }( x) P_{ajx}^{\ast }) \ln \Psi_{ajx}( \alpha ,g^{\ast },P^{\ast }) \vert
\end{array}
\right] \\
&\leq \sum_{( a,j,x) \in A\times J\times X}\left[ 
\begin{array}{c}
\frac{n_{ajx}}{n} \sup_{\alpha \in \Theta_{\alpha }}\vert \ln \Psi_{ajx}( \alpha ,\tilde{g},\tilde{P}) -\ln \Psi_{ajx}( \alpha ,g^{\ast },P^{\ast }) \vert \\
+ ( \frac{ n_{ajx}}{n}-m^{\ast }( x) P_{ajx}^{\ast }) \vert \ln ( \inf_{\alpha \in \Theta_{\alpha }}\vert \Psi_{ajx}( \alpha ,g^{\ast },P^{\ast }) \vert )
\end{array}
\right] =o_{p}(1),
\end{align*}
where the second inequality uses Assumptions \ref{ass:iid}, \ref{ass:Regularity}, and the intermediate value theorem.

\underline{Part (b).} Let $G$ be defined as follows:
\begin{align*}
G( \alpha ) ~\equiv~ {Q}_{PML}( \alpha ,g^{\ast },P^{\ast }) -{Q}_{PML}( \alpha ^{\ast },g^{\ast },P^{\ast }) =\sum_{(a,j,x) \in A\times J\times X}m^{\ast }( x) P_{ajx}^{\ast }\ln \left( \frac{\Psi_{ajx}( \alpha ,g^{\ast },P^{\ast }) }{\Psi_{ajx}( \alpha ^{\ast },g^{\ast },P^{\ast }) } \right),
\end{align*}
which is properly defined by Assumptions \ref{ass:Identification} and \ref{ass:Regularity}. By definition, $G( \alpha ^{\ast }) =0$. On the other hand, consider any $\alpha\neq \alpha ^{\ast }$. Assumption \ref{ass:Identification} implies that $\Psi ( \alpha ,g^{\ast },P^{\ast }) \not=\Psi ( \alpha ^{\ast },g^{\ast },P^{\ast }) $. This and Assumption \ref{ass:Regularity} then implies that $\Psi_{ajx}( \alpha ,g^{\ast },P^{\ast })/ \Psi_{ajx}( \alpha ^{\ast },g^{\ast },P^{\ast })\neq 1$ for some $( a,j,x) \in A\times J\times X$. Then,
\begin{equation*}
G( \alpha ) <\ln \left( \sum_{( a,j,x) \in A\times J\times X}m^{\ast }( x) P_{ajx}^{\ast }\frac{\Psi_{ajx}( \alpha ,g^{\ast },P^{\ast }) }{\Psi_{ajx}( \alpha ^{\ast },g^{\ast },P^{\ast }) }\right) =\ln \left( \sum_{( a,j,x) \in A\times J\times X}m^{\ast }( x) \Psi_{ajx}( \alpha ,g^{\ast },P^{\ast }) \right) =0,
\end{equation*}
where the inequality follows from Jensen's inequality, the strict convexity of the logarithm, and $\Psi_{ajx}( \alpha ,g^{\ast },P^{\ast })/ \Psi_{ajx}( \alpha ^{\ast },g^{\ast },P^{\ast })\neq 1$ for some $( a,j,x) \in A\times J\times X$, the first equality follows from Assumption \ref{ass:Regularity}, and the final equality follows from $\sum_{( a,j,x) \in A\times J\times X}m^{\ast }( x) \Psi_{ajx}( \alpha ,g^{\ast },P^{\ast }) =1$ for any $\alpha \in \Theta_{\alpha }$. Therefore, $G( \alpha ) $ and ${Q}_{PML}( \alpha ,g^{\ast },P^{\ast }) $ are uniquely maximized at $\alpha =\alpha ^{\ast }$.

\underline{Part (c).} For any $( \alpha ,g,P) \in \Theta_{\alpha }\times \Theta_{g}\times \Theta_{P}$ s.t.\ $\Psi ( \alpha ,g,P) $ is positive and differentiable, consider the following derivation.
\begin{align}
&\frac{\partial {\hat{Q}}_{PML}(\alpha ,g,P)}{\partial \alpha } =\frac{ \partial }{\partial \alpha }\left\{ \sum_{(j,x) \in J\times X} \frac{n_{x}}{n}\left[ \sum_{a\in \tilde{A}}\hat{P}_{ajx}\ln \Psi_{ajx}( \alpha ,g,P) +\hat{P}_{0jx}\ln \left( 1-\sum_{a\in \tilde{ A}}\Psi_{ajx}( \alpha ,g,P) \right)\right] \right\}\notag\\
&=\sum_{(a,j,x) \in \tilde{A} \times J \times X}\frac{n_{x}}{n} \left[ \frac{\hat{P}_{ajx}}{\Psi_{ajx}( \alpha ,g,P) }-\frac{ \hat{P}_{0jx}}{\Psi_{0jx}( \alpha ,g,P) }\right] \frac{\partial \Psi_{ajx}( \alpha ,g,P) }{\partial \alpha } \notag\\
% &=\sum_{(a,j,x) \in  \tilde{A}\times J\times X}\frac{n_{x}}{n} \left[ \frac{\hat{P}_{ajx}-\Psi_{ajx}( \alpha ,g,P) }{\Psi_{ajx}( \alpha ,g,P) }-\frac{\hat{P}_{0jx}-\Psi_{0jx}( \alpha ,g,P) }{\Psi_{0jx}( \alpha ,g,P) }\right] \frac{ \partial \Psi_{ajx}( \alpha ,g,P) }{\partial \alpha } \\
&=\sum_{(j,x) \in J\times X}\frac{n_{x}}{n}\left[ \sum_{a\in \tilde{A}}\frac{\hat{P}_{ajx}-\Psi_{ajx}( \alpha ,g,P) }{\Psi_{ajx}( \alpha ,g,P) }+\frac{\sum_{\tilde{a}\in \tilde{A}}( \hat{P}_{\tilde{a}jx}-\Psi_{\tilde{a}jx}( \alpha ,g,P) ) }{ \Psi_{0jx}( \alpha ,g,P) }\right] \frac{\partial \Psi_{ajx}( \alpha ,g,P) }{\partial \alpha }\notag \\
&=\sum_{(a,j,x) \in \tilde{A} \times J \times X}\frac{n_{x}}{n} \left[ \frac{\hat{P}_{ajx}-\Psi_{ajx}( \alpha ,g,P) }{\Psi_{ajx}( \alpha ,g,P) }\frac{\partial \Psi_{ajx}( \alpha ,g,P) }{\partial \alpha }+\frac{\hat{P}_{ajx}-\Psi_{ajx}( \alpha ,g,P) }{\Psi_{0jx}( \alpha ,g,P) }\sum_{\tilde{a}\in A} \frac{\partial \Psi_{\tilde{a}jx}( \alpha ,g,P) }{\partial \alpha }\right] ,\label{eq:derivationML}
\end{align}
where we have used that $\hat{P}_{0jx}=1-\sum_{a\in \tilde{A}}\hat{P}_{ajx}$ and $\Psi_{0jx}( \alpha ,g,P) =1-\sum_{a\in \tilde{A}}\Psi _{ajx}( \alpha ,g,P) $, and so $\partial \Psi_{0jx}( \alpha ,g,P) /\partial \alpha =-\sum_{a\in \tilde{A}}\partial \Psi _{ajx}( \alpha ,g,P) /\partial \alpha $. Then,
\begin{align*}
&\sqrt{n}\frac{\partial {\hat{Q}}_{PML}( \alpha ^{\ast },g^{\ast },P^{\ast }) }{\partial \alpha } =\sum_{(a,j,x) \in \tilde{A} \times J \times X} \frac{n_x}{n} \sqrt{n}( \hat{P}_{ajx}-P_{ajx}^{\ast }) \left[ \frac{1}{ P_{ajx}^{\ast }}\frac{\partial \Psi_{ajx}( \alpha ^{\ast },g^{\ast },P^{\ast }) }{\partial \alpha }+\frac{1}{P_{0jx}^{\ast }}\sum_{\tilde{ a}\in A}\frac{\partial \Psi_{\tilde{a}jx}( \alpha ^{\ast },g^{\ast },P^{\ast }) }{\partial \alpha }\right] \\
&=\sum_{(a,j,x) \in \tilde{A} \times J \times X}m^{\ast }( x) \sqrt{n}( \hat{P}_{ajx}-P_{ajx}^{\ast }) \left[ \frac{1}{ P_{ajx}^{\ast }}\frac{\partial \Psi_{ajx}( \alpha ^{\ast },g^{\ast },P^{\ast }) }{\partial \alpha }+\frac{1}{P_{0jx}^{\ast }}\sum_{\tilde{ a}\in A}\frac{\partial \Psi_{\tilde{a}jx}( \alpha ^{\ast },g^{\ast },P^{\ast }) }{\partial \alpha }\right] +o_{p}( 1) \\
&=\left[ 
\begin{array}{c}
\{ \frac{\partial \Psi_{ajx}( \alpha ^{\ast },g^{\ast },P^{\ast }) }{\partial \alpha }:( a,j,x) \in \tilde{A}\times J \times X\} ^{\prime }\times \\
diag\{ m^{\ast }( x) ( diag\{ 1/P_{ajx}^{\ast }:a\in \tilde{A}\} +\mathbf{1}_{|\tilde{A}| \times |\tilde{A}|}/P_{0jx}^{\ast }) :( j,x) \in J\times X\} \times \\ 
\{ \sqrt{n}( \hat{P}_{ajx}-P_{ajx}^{\ast }) :( a,j,x) \in \tilde{A}\times J\times X\}
\end{array}
\right] +o_{p}( 1) \\
&=\Psi_{\alpha }^{\prime }\Omega_{PP}^{-1}\sqrt{n}( \hat{P}-P^{\ast }) +o_{p}( 1) ,
\end{align*}
where the first equality holds by Eq.\ \eqref{eq:derivationML} and Assumption \ref{ass:Regularity}, the second equality holds by Assumption \ref{ass:iid}, and the final equality follows from Eq.\ \eqref{eq:InvOmegaPP}.

\underline{Part (d).} Consider the following derivation for any $( \alpha ,g,P) \in \Theta_{\alpha }\times \Theta_{g}\times \Theta_{P}$ s.t.\ $\Psi ( \alpha ,g,P) $ is positive and twice differentiable.
\begin{align}
\frac{\partial ^{2}{\hat{Q}}_{PML}(\alpha ,g,P)}{\partial \alpha \partial \lambda ^{\prime }} &=&\frac{\partial }{\partial \lambda ^{\prime }} \sum_{(a,j,x) \in \tilde{A} \times J\times X}\frac{n_{x}}{n}\left[ \frac{\hat{P}_{ajx}-\Psi_{ajx}( \alpha ,g,P) }{\Psi_{ajx}( \alpha ,g,P) }-\frac{\hat{P}_{a0x}-\Psi_{a0x}( \alpha ,g,P) }{\Psi_{a0x}( \alpha ,g,P) }\right] \frac{\partial \Psi_{ajx}( \alpha ,g,P) }{\partial \alpha } \notag \\
&=&\sum_{(a,j,x) \in \tilde{A} \times J \times X}\frac{n_{x}}{n}
\left\{ 
\begin{array}{c}
\left[ \frac{\hat{P}_{ajx}-\Psi_{ajx}( \alpha ,g,P) }{\Psi_{ajx}( \alpha ,g,P) }-\frac{\hat{P}_{a0x}-\Psi_{a0x}( \alpha ,g,P) }{\Psi_{a0x}( \alpha ,g,P) }\right] \frac{ \partial \Psi_{ajx}( \alpha ,g,P) }{\partial \alpha \partial \lambda ^{\prime }} \\
-\left[ \frac{\Psi_{ajx}( \alpha ,g,P) +( \hat{P}_{ajx}-\Psi_{ajx}( \alpha ,g,P) ) }{ \Psi_{ajx}( \alpha ,g,P) ^{2}}\right] \frac{\partial \Psi_{ajx}( \alpha ,g,P) }{\partial \alpha }\frac{\partial \Psi_{ajx}( \alpha ,g,P) }{\partial \lambda ^{\prime }} \\
+\left[ \frac{\Psi_{a0x}( \alpha ,g,P) +( \hat{P}_{a0x}-\Psi_{a0x}( \alpha ,g,P) ) }{\Psi_{a0x}( \alpha ,g,P) ^{2}}\right] \frac{\partial \Psi_{ajx}( \alpha ,g,P) }{\partial \alpha }\frac{\partial \Psi_{a0x}( \alpha ,g,P) }{\partial \lambda ^{\prime }}
\end{array}
\right\} .\label{eq:derivationML2}
\end{align}
Then, for any $\lambda \in \{ \alpha ,g,P\} $ and $( \tilde{ \alpha},\tilde{g},\tilde{P}) =( \alpha ^{\ast },g^{\ast },P^{\ast }) +o_{p}(1)$,
\begin{align*}
\frac{\partial ^{2}{\hat{Q}}_{PML}(\tilde{\alpha},\tilde{g},\tilde{P})}{ \partial \alpha \partial \lambda ^{\prime }} &\overset{p}{\to}\sum_{(a,j,x) \in \tilde{A} \times J \times X}m^{\ast }( x) \left[
\begin{array}{c}
\frac{1}{P_{0xj}^{\ast }}\frac{\partial \Psi_{ajx}( \alpha ^{\ast },g^{\ast },P^{\ast }) }{\partial \alpha }\frac{\partial \Psi_{a0x}( \alpha ^{\ast },g^{\ast },P^{\ast }) }{\partial \lambda ^{\prime }}- \frac{1}{P_{ajx}^{\ast }}\frac{\partial \Psi_{ajx}( \alpha ^{\ast },g^{\ast },P^{\ast }) }{\partial \alpha }\frac{\partial \Psi_{ajx}( \alpha ^{\ast },g^{\ast },P^{\ast }) }{\partial \lambda ^{\prime }}
\end{array}
\right] \\
&=-\sum_{(j,x) \in J\times X}m^{\ast }( x) \left[ 
\begin{array}{c}
\sum_{a\in \tilde{A}}\frac{1}{P_{ajx}^{\ast }}\frac{\partial \Psi_{ajx}( \alpha ^{\ast },g^{\ast },P^{\ast }) }{\partial \alpha } \frac{\partial \Psi_{ajx}( \alpha ^{\ast },g^{\ast },P^{\ast }) }{\partial \lambda ^{\prime }} +\\
\frac{1}{P_{0xj}^{\ast }}\sum_{\check{a}\in \tilde{A}}\frac{\partial \Psi_{\check{a}jx}( \alpha ^{\ast },g^{\ast },P^{\ast }) }{\partial \alpha }\sum_{\tilde{a}\in \tilde{A}}\frac{\partial \Psi_{\tilde{a}jx}( \alpha ^{\ast },g^{\ast },P^{\ast }) }{\partial \lambda ^{\prime }}
\end{array}
\right] \\
&=-\left[
\begin{array}{c}
\{ \frac{\partial \Psi_{ajx}( \alpha ^{\ast },g^{\ast },P^{\ast }) }{\partial \alpha }:( a,j,x) \in \tilde{A}\times J\times X \} ^{\prime }\times \\ 
diag\{ m^{\ast }( x) [ diag\{ 1/P_{ajx}^{\ast }:a\in \tilde{A}\} +\mathbf{1}_{|\tilde{A}| \times |\tilde{A}| }/P_{0jx}^{\ast }] :( j,x) \in J\times X \} \times \\
\{ \frac{\partial \Psi_{ajx}( \alpha ^{\ast },g^{\ast },P^{\ast }) }{\partial \lambda }:( a,j,x) \in \tilde{A}\times J\times X \}
\end{array}
\right]\\
&=-\Psi_{\alpha }^{\prime }\Omega_{PP}^{-1}\Psi_{\lambda }=\frac{ \partial ^{2}{Q}_{PML}( \alpha ^{\ast },g^{\ast },P^{\ast }) }{ \partial \alpha \partial \lambda ^{\prime }},
\end{align*}
where the first equality holds by Eq.\ \eqref{eq:derivationML2} and Assumptions \ref{ass:iid} and \ref{ass:Regularity}, the second equality holds by $\Psi_{a0x}( \alpha ,g,P) =1-\sum_{a\in \tilde{A}}\Psi_{ajx}( \alpha ,g,P) $ and so $\partial \Psi_{a0x}( \alpha ,g,P) /\partial \lambda ^{\prime }=-\sum_{a\in \tilde{A}}\partial \Psi_{ajx}( \alpha ,g,P) /\partial \lambda ^{\prime }$, and the final equality follows from Eq.\ \eqref{eq:InvOmegaPP}.

\underline{Part (e).} $\partial ^{2}{Q}_{PML}( \alpha ^{\ast },g^{\ast },P^{\ast }) /\partial \alpha \partial \alpha ^{\prime }=-\Psi_{\alpha }^{\prime }\Omega_{PP}^{-1}\Psi_{\alpha }$ is nonsingular because $\Psi_{\alpha }$ has full rank.

This completes the verification of Assumption \ref{ass:EE_HL}. Since we also assume Assumptions \ref{ass:Regularity} and \ref{ass:Baseline}, Theorem \ref{thm:2stepB} applies. In particular, Eq.\ \eqref{eq:Expansions_3} yields
\begin{align}
\sqrt{n}(\hat{\theta}_{K-PML}-\theta ^{\ast })
&=\left[ 
\begin{array}{ccc}
A_{K}+B_{K}\Upsilon_{K,P} & B_{K}\Upsilon_{K,0} & B_{K}\Upsilon
_{K,g}+C_{K} \\ 
\mathbf{0}_{d_{g}\times d_{P}} & \mathbf{0}_{d_{g}\times d_{P}} & \mathbf{I}
_{d_{g}}
\end{array}
\right] \sqrt{n}\left[ 
\begin{array}{c}
\hat{P}-P^{\ast } \\ 
\hat{P}_{0}-P^{\ast } \\ 
\hat{g}-g^{\ast }
\end{array}
\right] +o_{p}(1), \label{eq:KML_almost}
\end{align}
with $A_{K}$, $B_{K}$, and $C_{K}$ determined according to Eq.\ \eqref{eq:ABCdefns}, and $\Upsilon_{K,P}$, $\Upsilon_{K,0}$, and $\Upsilon_{K,g}$ determined according to Eq.\ \eqref{eq:Upsilondefns}. As a next step, we work out these constants.

For $k\leq K$, $\Xi_{k}=\Psi_{\alpha }^{\prime }\Omega_{PP}^{-1}$ and $\partial ^{2}{Q}_{PML}( \alpha ^{\ast },g^{\ast },P^{\ast }) /\partial \alpha \partial \lambda ^{\prime }=-\Psi_{\alpha }^{\prime }\Omega_{PP}^{-1}\Psi_{\lambda }$, and so, according to Eq.\ \eqref{eq:ABCdefns}, $A_{k}=( \Psi_{\alpha }^{\prime }\Omega_{PP}^{-1}\Psi_{\alpha }) ^{-1}\Psi_{\alpha }^{\prime }\Omega_{PP}^{-1}$, $B_{k}=-( \Psi_{\alpha }^{\prime }\Omega_{PP}^{-1}\Psi_{\alpha }) ^{-1}\Psi_{\alpha }^{\prime }\Omega_{PP}^{-1}\Psi_{P}$, and $C_{k}=-( \Psi_{\alpha }^{\prime }\Omega_{PP}^{-1}\Psi_{\alpha }) ^{-1}\Psi_{\alpha }^{\prime }\Omega_{PP}^{-1}\Psi_{g}$. In addition, according to Eq.\ \eqref{eq:Upsilondefns},  $\{ \Upsilon_{k,P}:k\leq K\} $, $\{ \Upsilon_{k,g}:k\leq K\} $, and $ \{ \Upsilon_{k,0}:k\leq K\} $ are as follows. Set $\Upsilon_{1,0}\equiv \mathbf{I}_{d_{P}}$, $\Upsilon_{1,g}\equiv \mathbf{0}_{d_{P}\times d_{g}}$, $\Upsilon_{1,P}\equiv \mathbf{0}_{d_{P}\times d_{P}}$ and, for any $k=1,\ldots ,K-1$,
\begin{align*}
\Upsilon_{k+1,P} &=( \mathbf{I}_{d_{P}}-\Psi_{\alpha }( \Psi_{\alpha }^{\prime }\Omega_{PP}^{-1}\Psi_{\alpha }) ^{-1}\Psi_{\alpha }^{\prime }\Omega_{PP}^{-1}) \Psi_{P}\Upsilon_{k,P}+\Psi_{\alpha }( \Psi_{\alpha }^{\prime }\Omega_{PP}^{-1}\Psi_{\alpha }) ^{-1}\Psi_{\alpha }^{\prime }\Omega_{PP}^{-1} \\
\Upsilon_{k+1,0} &=( \mathbf{I}_{d_{P}}-\Psi_{\alpha }( \Psi_{\alpha }^{\prime }\Omega_{PP}^{-1}\Psi_{\alpha }) ^{-1}\Psi_{\alpha }^{\prime }\Omega_{PP}^{-1}) \Psi_{P}\Upsilon_{k,0}\\
\Upsilon_{k+1,g} &=( \mathbf{I}_{d_{P}}-\Psi_{\alpha }( \Psi_{\alpha }^{\prime }\Omega_{PP}^{-1}\Psi_{\alpha }) ^{-1}\Psi_{\alpha }^{\prime }\Omega_{PP}^{-1}) \Psi_{P}\Upsilon_{k,g}+( \mathbf{I}_{d_{P}}-\Psi_{\alpha }( \Psi_{\alpha }^{\prime }\Omega_{PP}^{-1}\Psi_{\alpha }) ^{-1}\Psi_{\alpha }^{\prime }\Omega_{PP}^{-1}) \Psi_{g}.
\end{align*}
Then, Eq.\ \eqref{eq:coefficients_ML} follows from setting $\Upsilon_{k,P}\equiv \Phi_{k,P}$, and $\Upsilon_{k,g}\equiv\Phi_{k,g}\Psi_{g}$, $\Upsilon_{k,0}\equiv\Phi_{k,0}$ for all $k\leq K$.

If we plug this information into Eq.\ \eqref{eq:KML_almost} and combine with Assumption \ref{ass:Baseline}, we deduce that
\begin{equation}
\sqrt{n}( \hat{\theta}_{K-PML}-\theta ^{\ast }) ~=~\left( 
\begin{array}{c}
\sqrt{n}( \hat{\alpha}_{K-PML}-\alpha ^{\ast })  \\ 
\sqrt{n}( \hat{g}-g^{\ast }) 
\end{array}
\right) ~\overset{d}{\to}~N\left( \left( 
\begin{array}{c}
\mathbf{0}_{d_{\alpha }} \\ 
\mathbf{0}_{d_{g}}
\end{array}
\right) ,\left( 
\begin{array}{cc}
\Sigma_{K-PML}( \hat{P}_{0})  & \Sigma_{\alpha g,K-PML} \\ 
\Sigma_{\alpha g,K-PML}^{\prime } & \Omega_{gg}
\end{array}
\right) \right) ,  \label{eq:JointML}
\end{equation}
where $\Sigma_{K-PML}( \hat{P}_{0})$ is as defined in Theorem \ref{thm:ML} and 
\begin{equation*}
	\Sigma_{\alpha g,K-PML}~\equiv~( \Psi_{\alpha }^{\prime }\Omega_{PP}^{-1}\Psi_{\alpha }) ^{-1}\Psi_{\alpha }^{\prime }\Omega_{PP}^{-1}[ ( \mathbf{I}_{d_{P}}-\Psi_{P}\Phi_{K,P}) \Omega_{Pg}-\Psi_{P}\Phi_{K,0}\Omega_{0g} -( \Psi_{P}\Phi_{K,g}+\mathbf{I}_{d_{P}}) \Psi_{g}\Omega_{gg} ] .
\end{equation*}
The desired result is a corollary of Eq.\ \eqref{eq:JointML}.
\end{proof}

%%%%%%%% DIVIDER %%%%%%%%%%%%

\begin{proof}[Proof of Theorem \ref{thm:MD}]
	This result is a consequence of applying Theorem \ref{thm:2stepB} with $\hat{Q}_{k}\equiv {\hat{Q}}_{k-MD}$ and ${Q}_{k}\equiv {Q}_{k-MD}$ where, for any $ ( \alpha ,g,P) \in \Theta_{\alpha }\times \Theta_{g}\times \Theta_{P}$,
\begin{equation*}
{Q}_{k-MD}(\alpha ,g,P)~\equiv~ -( P^{\ast }-\Psi ( \alpha ,g,P) ) ^{\prime }W_{k}( P^{\ast }-\Psi ( \alpha ,g,P) ) .
\end{equation*}
For any $\lambda \in \{ \alpha ,g,P\} $, notice that ${\partial ^{2}{Q}_{k-MD}( \alpha ^{\ast },g^{\ast },P^{\ast }) }/{\partial \alpha \partial \lambda ^{\prime }}=-2\Psi_{\alpha }^{\prime }W_{k}\Psi_{\lambda }$. 

To apply this result, we first verify Assumption \ref{ass:EE_HL}.

\underline{Part (a).} For any $( \tilde{g},\tilde{P}) =( g^{\ast },P^{\ast }) +o_{p}(1)$,
\begin{align*}
&\sup_{\alpha \in \Theta_{\alpha }}\vert {\hat{Q}}_{k-MD}(\alpha , \tilde{g},\tilde{P})-{Q}_{k-MD}(\alpha ,g^{\ast },P^{\ast })\vert\leq \left[
\begin{array}{c}
\Vert \hat{W}_{k}-W_{k}\Vert +\Vert \tilde{P}-P^{\ast }\Vert ^{2}\Vert W_{k}\Vert +2\Vert \tilde{P}-P^{\ast }\Vert \Vert W_{k}\Vert \\
+2\Vert W_{k}\Vert \sup_{\alpha \in \Theta_{\alpha }}\Vert \Psi ( \alpha ,g^{\ast },P^{\ast }) -\Psi ( \alpha ,\tilde{g} ,\tilde{P}) \Vert
\end{array}
\right] =o_{p}(1),
\end{align*}
where the last equality uses Assumption \ref{ass:Weight}.

\underline{Part (b).} First, consider $\alpha =\alpha ^{\ast }$. Then, Assumption \ref{ass:Identification} implies $\Psi ( \alpha ,g^{\ast },P^{\ast }) =P^{\ast }$, and so ${Q}_{k-MD}(\alpha ^{\ast },g^{\ast },P^{\ast })=0$. Second, consider $\alpha \neq \alpha ^{\ast }$.  Then, Assumption \ref{ass:Identification} implies $\Psi ( \alpha ,g^{\ast },P^{\ast }) \neq P^{\ast }$. This and Assumption \ref{ass:Weight} imply that ${Q}_{k-MD}(\alpha ,g^{\ast },P^{\ast })<0$. Then, ${Q}_{k-MD}(\alpha ,g^{\ast },P^{\ast })$ is uniquely maximized at $\alpha =\alpha ^{\ast }$, as required.

\underline{Part (c).} Consider the following derivation: 
\begin{align*}
\sqrt{n}\frac{\partial {\hat{Q}}_{k-MD}( \alpha ^{\ast },g^{\ast },P^{\ast }) }{\partial \alpha } = 2( \hat{P}-\Psi ( \alpha ^{\ast },g^{\ast },P^{\ast }) ) ^{\prime }\hat{W}_{k}\frac{ \partial \Psi ( \alpha ^{\ast },g^{\ast },P^{\ast }) }{\partial \alpha } = 2\Psi_{\alpha }^{\prime }W_{k}\sqrt{n}( \hat{P}-P^{\ast }) +o_{p}( 1) ,
\end{align*}
where the second equality uses Assumptions \ref{ass:Regularity} and \ref{ass:Weight}.

\underline{Part (d).} For any $\lambda \in \{ \alpha ,g,P\} $ and $( \tilde{\alpha},\tilde{g},\tilde{P}) =( \alpha ^{\ast },g^{\ast },P^{\ast }) +o_{p}(1)$,
\begin{eqnarray*}
\frac{\partial ^{2}{\hat{Q}}_{k-MD}(\tilde{\alpha},\tilde{g},\tilde{P})}{ \partial \alpha \partial \lambda ^{\prime }} &=&-2\sqrt{n}\frac{\partial \Psi ( \tilde{\alpha},\tilde{g},\tilde{P}) ^{\prime }}{\partial \lambda ^{\prime }}\hat{W}_{k}\frac{\partial \Psi ( \tilde{\alpha},\tilde{g}, \tilde{P}) }{\partial \alpha }+2( \hat{P}-\Psi ( \tilde{ \alpha},\tilde{g},\tilde{P}) ) ^{\prime }\hat{W}_{k}\frac{ \partial \Psi ( \tilde{\alpha},\tilde{g},\tilde{P}) }{\partial \alpha \partial \lambda ^{\prime }} \\
&\overset{p}{\to}&\frac{\partial ^{2}{Q}_{k-MD}( \alpha ^{\ast },g^{\ast },P^{\ast }) }{\partial \alpha \partial \lambda ^{\prime }} =-2\Psi_{\alpha }^{\prime }W_{k}\Psi_{\lambda },
\end{eqnarray*}
where the convergence uses Assumptions \ref{ass:Regularity} and \ref{ass:Weight}.

\underline{Part (e).} $\partial ^{2}{Q}_{K-MD}( \alpha ^{\ast },g^{\ast },P^{\ast }) /\partial \alpha \partial \alpha ^{\prime }=-2\Psi_{\alpha }^{\prime }W_{k}\Psi_{\alpha }$ is nonsingular by Assumptions \ref{ass:Regularity} and \ref{ass:Weight}.

This completes the verification of Assumption \ref{ass:EE_HL}. Since we also assume Assumptions \ref{ass:Regularity} and \ref{ass:Baseline}, Theorem \ref{thm:2stepB} applies. In particular, Eq.\ \eqref{eq:Expansions_3} yields:
\begin{align}
\sqrt{n}(\hat{\theta}_{K-MD}-\theta ^{\ast })
&=\left[ 
\begin{array}{ccc}
A_{K}+B_{K}\Upsilon_{K,P} & B_{K}\Upsilon_{K,0} & B_{K}\Upsilon
_{K,g}+C_{K} \\ 
\mathbf{0}_{d_{g}\times d_{P}} & \mathbf{0}_{d_{g}\times d_{P}} & \mathbf{I}
_{d_{g}}
\end{array}
\right] \sqrt{n}\left[ 
\begin{array}{c}
\hat{P}-P^{\ast } \\ 
\hat{P}_{0}-P^{\ast } \\ 
\hat{g}-g^{\ast }
\end{array}
\right] +o_{p}(1), \label{eq:KMD_almost}
\end{align}
with $A_{K}$, $B_{K}$, and $C_{K}$ determined according to Eq.\ \eqref{eq:ABCdefns}, and $\Upsilon_{K,P}$, $\Upsilon_{K,0}$, and $\Upsilon_{K,g}$ determined according to Eq.\ \eqref{eq:Upsilondefns}. As a next step, we work out these constants.

For $k\leq K$, $\Xi_{k}=2\Psi_{\alpha }^{\prime }W_{k}$ and $\partial ^{2}{Q}_{k-MD}( \alpha ^{\ast },g^{\ast },P^{\ast }) /\partial \alpha \partial \lambda ^{\prime }=-2\Psi_{\alpha }^{\prime }W_{k}\Psi_{\lambda }$, and so $A_{k}=( \Psi_{\alpha }^{\prime }W_{k}\Psi_{\alpha }) ^{-1}\Psi_{\alpha }^{\prime }W_{k}$, $B_{k}=-( \Psi_{\alpha }^{\prime }W_{k}\Psi_{\alpha }) ^{-1}\Psi_{\alpha }^{\prime }W_{k}\Psi_{P}$, and $ C_{k}=-( \Psi_{\alpha }^{\prime }W_{k}\Psi_{\alpha }) ^{-1}\Psi_{\alpha }^{\prime }W_{k}\Psi_{g}$. Then, $\{ \Upsilon_{k,P}:k\leq K\} $, $\{ \Upsilon_{k,g}:k\leq K\} $, and $\{ \Upsilon_{k,0}:k\leq K\} $ are as follows. Set $\Upsilon_{1,0}\equiv \mathbf{I}_{d_{P}}$, $\Upsilon_{1,g}\equiv \mathbf{0}_{d_{P}\times d_{g}}$, $\Upsilon_{1,P}\equiv \mathbf{0}_{d_{P}\times d_{P}}$ and, for any $ k=1,\ldots ,K-1$,
\begin{align*}
\Upsilon_{k+1,P} &=( \mathbf{I}_{d_{P}}-\Psi_{\alpha }( \Psi_{\alpha }^{\prime }W_{k}\Psi_{\alpha }) ^{-1}\Psi_{\alpha }^{\prime }W_{k}) \Psi_{P}\Upsilon_{k,P}+\Psi_{\alpha }( \Psi_{\alpha }^{\prime }W_{k}\Psi_{\alpha }) ^{-1}\Psi_{\alpha }^{\prime }W_{k} \\
\Upsilon_{k+1,0} &=( \mathbf{I}_{d_{P}}-\Psi_{\alpha }( \Psi_{\alpha }^{\prime }W_{k}\Psi_{\alpha }) ^{-1}\Psi_{\alpha }^{\prime }W_{k}) \Psi_{P}\Upsilon_{k,0}\\
\Upsilon_{k+1,g} &=( \mathbf{I}_{d_{P}}-\Psi_{\alpha }( \Psi_{\alpha }^{\prime }W_{k}\Psi_{\alpha }) ^{-1}\Psi_{\alpha }^{\prime }W_{k}) \Psi_{P}\Upsilon_{k,g}+( \mathbf{I}_{d_{P}}-\Psi_{\alpha }( \Psi_{\alpha }^{\prime }W_{k}\Psi_{\alpha }) ^{-1}\Psi_{\alpha }^{\prime }W_{k}) \Psi_{g}.
\end{align*}
Then, Eq.\ \eqref{eq:coefficients_MD} follows from setting $\Upsilon_{k,P}\equiv \Phi_{k,P}$, and $\Upsilon_{k,g}\equiv\Phi_{k,g}\Psi_{g}$, $\Upsilon_{k,0}\equiv\Phi_{k,0}$ for all $k\leq K$.

If we plug this information into Eq.\ \eqref{eq:KMD_almost} and combine with Assumption \ref{ass:Baseline}, we deduce that:
\begin{equation}
\sqrt{n}( \hat{\theta}_{K-MD}-\theta ^{\ast }) ~=~\left( 
\begin{array}{c}
\sqrt{n}( \hat{\alpha}_{K-MD}-\alpha ^{\ast })  \\ 
\sqrt{n}( \hat{g}-g^{\ast }) 
\end{array}
\right) ~\overset{d}{\to}~N\left( \left( 
\begin{array}{c}
\mathbf{0}_{d_{\alpha }\times 1} \\ 
\mathbf{0}_{d_{g}\times 1}
\end{array}
\right) ,\left( 
\begin{array}{cc}
\Sigma_{K-MD}( \hat{P}_{0},\{ W_{k}:k\leq K\} ) & \Sigma_{\alpha g,K-MD} \\ 
\Sigma_{\alpha g,K-MD}^{\prime } & \Omega_{gg}
\end{array}
\right) \right) ,  \label{eq:JointMD}
\end{equation}
where $\Sigma_{K-MD}( \hat{P}_{0},\{ W_{k}:k\leq K\} )$ is as defined in Theorem \ref{thm:MD} and 
\begin{equation*}
	\Sigma_{\alpha g,K-MD}~\equiv~( \Psi_{\alpha }^{\prime }W_{K}\Psi_{\alpha }) ^{-1}\Psi_{\alpha }^{\prime }W_{K}[ ( \mathbf{I}_{d_{P}}-\Psi_{P}\Phi_{K,P}) \Omega_{Pg}-\Psi_{P}\Phi_{K,0}\Omega_{0g} -( \Psi_{P}\Phi_{K,g}+\mathbf{I}_{d_{P}}) \Psi_{g}\Omega_{gg} ] .
\end{equation*}
The desired result is a corollary of Eq.\ \eqref{eq:JointMD}.
\end{proof}

%%%%%%%% DIVIDER %%%%%%%%%%%%

\begin{proof}[Proof of Theorem \ref{thm:MD_K1}]
	The asymptotic distribution of the $1$-MD estimator follows from Theorem \ref{thm:MD}. The proof is completed by showing that $\Sigma_{1-MD}(\hat{P}_{0},W_{1})-\Sigma_{1-MD}(\tilde{P},W_{1}^*)$ is PSD and that Eq.\ \eqref{eq:Optimal_AVar_MD} holds.
	
	By combining Theorem \ref{thm:MD} and Eq.\ \eqref{eq:asyequiv_P0}, we obtain the following derivation.
	\begin{align}
	&\Sigma_{1-MD}(\tilde{P},W_{1}) \notag \\
	&=(\Psi_{\alpha }^{\prime }W_{1}\Psi_{\alpha })^{-1}\Psi_{\alpha }^{\prime }W_{1}\left[ \left(
	\begin{array}{ccc}
	\mathbf{I}_{d_{P}} & -\Psi_{P} & -\Psi_{g}
	\end{array}
	\right) \left(
	\begin{array}{ccc}
	\Omega_{PP} & \Omega_{PP} & \Omega_{Pg} \\
	\Omega_{PP} & \Omega_{PP} & \Omega_{Pg} \\
	\Omega_{Pg}^{\prime } & \Omega_{Pg}^{\prime } & \Omega_{gg}
	\end{array}
	\right) \left(
	\begin{array}{c}
	\mathbf{I}_{d_{P}} \\
	-\Psi_{P}^{\prime } \\
	-\Psi_{g}^{\prime }
	\end{array}
	\right) \right] W_{1}^{\prime }\Psi_{\alpha }(\Psi_{\alpha }^{\prime }W_{1}^{\prime }\Psi_{\alpha })^{-1} \label{eq:MD_K1_0} \\
	& = (\Psi_{\alpha }^{\prime }W_{1}\Psi_{\alpha })^{-1}\Psi_{\alpha }^{\prime }W_{1}\left[ \left(
	\begin{array}{cc}
	\mathbf{I}_{d_{P}}-\Psi_{P} & -\Psi_{g}
	\end{array}
	\right) \left(
	\begin{array}{cc}
	\Omega_{PP} & \Omega_{Pg} \\
	\Omega_{Pg}^{\prime } & \Omega_{gg}
	\end{array}
	\right) \left(
	\begin{array}{c}
	\mathbf{I}_{d_{P}}-\Psi_{P}^{\prime } \\
	-\Psi_{g}^{\prime }
	\end{array}
	\right) \right] W_{1}^{\prime }\Psi_{\alpha }(\Psi_{\alpha }^{\prime }W_{1}^{\prime }\Psi_{\alpha })^{-1}.\label{eq:MD_K1_2}
	\end{align}
	
	By Eq.\ \eqref{eq:MD_K1_0} and Lemma \ref{lem:PSD},
	\begin{equation}
	\Sigma_{1-MD}(\hat{P}_{0},W_{1})-\Sigma_{1-MD}(\tilde{P},W_{1})\text{ is PSD.} \label{eq:MD_K1_1}
	\end{equation}
	By Assumptions \ref{ass:Regularity} and \ref{ass:Baseline2}, the expression in brackets in RHS of Eq.\ \eqref{eq:MD_K1_2} is non-singular. Then, standard arguments in GMM estimation (e.g.\ \citealp[page 2165]{mcfadden/newey:1994}) imply that $ W_{1}^{\ast }$ in Eq.\ \eqref{eq:W1_optimal} is efficient. Therefore,
	\begin{equation}
	\Sigma_{1-MD}(\tilde{P},W_{1})-\Sigma_{1-MD}(\tilde{P},W_{1}^{\ast })\text{ is PSD.} \label{eq:MD_K1_3}
	\end{equation}

	By combining Eqs.\ \eqref{eq:MD_K1_1} and \eqref{eq:MD_K1_3}, we conclude that $ \Sigma_{1-MD}(\hat{P}_{0},W_{1})-\Sigma_{1-MD}(\tilde{P},W_{1}^{\ast })$ is PSD, as desired. Finally, Eq.\ \eqref{eq:Optimal_AVar_MD} follows from plugging in $ W_{1}^{\ast }$ in Eq.\ \eqref{eq:W1_optimal} in Eq.\ \eqref{eq:MD_K1_2}.
\end{proof}

%%%%%%%% DIVIDER %%%%%%%%%%%%

\begin{proof}[Proof of Theorem \ref{thm:MD_main}] As in the statement of \underline{optimality}, let $\hat{\alpha}_{K-MD}$ denote the $K$-MD estimator with arbitrary initial CCP estimator $\hat{P}_{0}$ and weight matrices $\{ W_{k}:k\leq K\} $.
By Theorem \ref{thm:MD}, $\sqrt{n}( \hat{\alpha}_{K-MD}-\alpha ^{\ast }) \overset{d}{\to} N( \mathbf{0}_{d_{\alpha }},\Sigma_{K-MD}( \hat{P}_{0},\{ W_{k}:k\leq K\} ) ) $. By combining Eq.\ \eqref{eq:asyequiv_P0} and Lemma \ref{lem:PSD},
\begin{align}
	\Sigma_{K-MD}( \hat{P}_{0},\{ W_{k}:k\leq K\} )  - \Sigma_{K-MD}( \tilde{P},\{ W_{k}:k\leq K\} )~~~ \text{is PSD.}
	\label{eq:MD_main_1}
\end{align}

As in the statement of \underline{invariance}, let $\hat{\alpha}^{*}_{K-MD}$ denote the $K$-MD estimator with initial CCP estimator $\tilde{P}$ and weight matrices $\{ W_{k}:k\leq K-1\} $ for steps $1,\ldots ,K-1$ (if $K>1$), and the corresponding optimal weight matrix in step $K$. By Theorem \ref{thm:MD}, $\sqrt{n}( \hat{\alpha}^*_{K-MD}-\alpha ^{\ast }) \overset{d}{\to} N( \mathbf{0}_{d_{\alpha }},\Sigma_{K-MD}( \tilde{P},\{\{ W_{k}:k\leq K-1\},W_{K}^{*}\} ) ) $. By definition of an optimal choice of $W_{K}$,
\begin{align}
	\Sigma_{K-MD}(\tilde{P},\{\{ W_{k}:k\leq K\}\}) - \Sigma_{K-MD}( \tilde{P},\{\{ W_{k}:k\leq K-1\},W_{K}^{*}\})~~~\text{is PSD.}
	\label{eq:MD_main_2}
\end{align}

As a next step, we provide an explicit formula for $W_{K}^{\ast }$ and we compute the resulting asymptotic variance $\Sigma _{K-MD}(\tilde{P} ,\{\{W_{k}:k\leq K-1\},W_{K}^{\ast }\})$. To this end, consider the following derivation.
\begin{align}
&\Sigma _{K-MD}(\tilde{P},\{W_{k}:k\leq K\}) \notag \\
&=\left\{ 
\begin{array}{c}
(\Psi _{\alpha }^{\prime }W_{K}\Psi _{\alpha })^{-1}\Psi _{\alpha }^{\prime
}W_{K}\times \\ 
\left[ 
\begin{array}{c}
(\mathbf{I}_{d_{P}}-\Psi _{P}\Phi _{K,P})^{\prime } \\ 
-(\Psi _{P}\Phi _{K,0})^{\prime } \\ 
-((\mathbf{I}_{d_{P}}+\Psi _{P}\Phi _{K,g})\Psi _{g})^{\prime }
\end{array}
\right] ^{\prime }\left( 
\begin{array}{ccc}
\Omega _{PP} & \Omega _{PP} & \Omega _{Pg} \\ 
\Omega _{PP}^{\prime } & \Omega _{PP} & \Omega _{Pg} \\ 
\Omega _{Pg}^{\prime } & \Omega _{Pg}^{\prime } & \Omega _{gg}
\end{array}
\right) \left[ 
\begin{array}{c}
(\mathbf{I}_{d_{P}}-\Psi _{P}\Phi _{K,P})^{\prime } \\ 
-(\Psi _{P}\Phi _{K,0})^{\prime } \\ 
-((\mathbf{I}_{d_{P}}+\Psi _{P}\Phi _{K,g})\Psi _{g})^{\prime }
\end{array}
\right] \label{eq:deriv_1} \\ 
\times W_{K}^{\prime }\Psi _{\alpha }(\Psi _{\alpha }^{\prime }W_{K}^{\prime
}\Psi _{\alpha })^{-1}
\end{array}
\right\} \\
&=(\Psi _{\alpha }^{\prime }W_{K}\Psi _{\alpha })^{-1}\Psi _{\alpha }^{\prime }W_{K}\Delta _{K}W_{K}^{\prime }\Psi _{\alpha }(\Psi _{\alpha }^{\prime }W_{K}^{\prime }\Psi _{\alpha })^{-1},\label{eq:deriv_2}
\end{align}
where Eq.\ \eqref{eq:deriv_1} follows from Eq.\ \eqref{eq:asyequiv_P0} and Eq.\ \eqref{eq:deriv_2} follows from defining
\begin{equation}
\Delta _{K}~\equiv~ \left[ 
\begin{array}{c}
(\mathbf{I}_{d_{P}}-\Psi _{P}\Phi _{K,P0})^{\prime } \\ 
-((\mathbf{I}_{d_{P}}+\Psi _{P}\Phi _{K,g})\Psi _{g})^{\prime }
\end{array}
\right] ^{\prime }\left( 
\begin{array}{cc}
\Omega _{PP} & \Omega _{Pg} \\ 
\Omega _{Pg}^{\prime } & \Omega _{gg}
\end{array}
\right) \left[ 
\begin{array}{c}
(\mathbf{I}_{d_{P}}-\Psi _{P}\Phi _{K,P0})^{\prime } \\ 
-((\mathbf{I}_{d_{P}}+\Psi _{P}\Phi _{K,g})\Psi _{g})^{\prime }
\end{array}
\right]  \label{eq:Delta_K}
\end{equation}
and $\{\Phi _{k,P0}:k\leq K\}$ is defined by $\Phi _{k,P0}\equiv \Phi_{k,P}+\Phi _{k,0}$ for $k\leq K$. Under the assumption that $\Delta _{K}$ is non-singular, standard arguments in GMM (e.g.\ \citet[Theorem 5.2]{mcfadden/newey:1994}) imply that $W_{K}^{\ast } =\Delta _{K}^{-1}$, resulting in an (optimal) asymptotic variance $\Sigma _{K-MD}(\hat{P} ,\{\{W_{k}:k\leq K-1\},W_{K}^{\ast }\})=(\Psi _{\alpha }^{\prime }\Delta _{K}^{-1}\Psi _{\alpha })^{-1}$. Note that the statement assumes that $\Delta _{K}$ is non-singular. For completeness, we now provide conditions under which this occurs. For any $K\geq 1$, some algebra shows that 
\begin{eqnarray}
	 [\begin{array}{cc}
	 \mathbf{I}_{d_{P}}-\Psi _{P}\Phi _{K,P0}& (\mathbf{I}_{d_{P}}+\Psi _{P}\Phi _{K,g}) \Psi _{g})
	 \end{array}
	 ] =
	 \Lambda _{K}
	 [\begin{array}{cc}
	 \mathbf{I }_{d_{P}}-\Psi _{P}&\Psi _{g}
	 \end{array}
	 ].
	 \label{eq:Seq_P0_and_G}
\end{eqnarray}
with $\Lambda _{K}~\equiv~ \sum_{i=0}^{K-1}( \Psi _{P}( \mathbf{I} _{d_{P}}-\Psi _{\alpha }(\Psi _{\alpha }^{\prime }\Omega _{PP}^{-1}\Psi _{\alpha })^{-1}\Psi _{\alpha }^{\prime }\Omega _{PP}^{-1}) ) ^{i}$. Therefore,
\begin{equation*}
\Delta _{K}=\Lambda _{K}\left\{ \left[ 
\begin{array}{c}
( \mathbf{I}_{d_{P}}-\Psi _{P}) ^{\prime } \\ 
-\Psi _{g}^{\prime }
\end{array}
\right] ^{\prime }\left( 
\begin{array}{cc}
\Omega _{PP} & \Omega _{Pg} \\ 
\Omega _{Pg}^{\prime } & \Omega _{gg}
\end{array}
\right) \left[ 
\begin{array}{c}
( \mathbf{I}_{d_{P}}-\Psi _{P}) ^{\prime } \\ 
-\Psi _{g}^{\prime }
\end{array}
\right] \right\} \Lambda _{K}^{\prime }.
\end{equation*}
By Assumptions \ref{ass:Regularity} and \ref{ass:Baseline2}, the matrix in braces in the last display is non-singular. Then, \citet[Proposition 2.7.3 and Corollary 2.7.6]{bernstein:2009} imply that $\Delta _{K}$ is non-singular if and only if $\Lambda _{K}$ is non-singular. This condition holds automatically in single-agent problems, and has been numerically verified in our simulations in Section \ref{sec:MCsimulations}.

Since the choice of $\{W_{k}:k\leq K-1\}$ was completely arbitrary, the proof of \underline{invariance} follows from showing that
\begin{align}
	\Psi_{\alpha }^{\prime }\Delta_{K}^{-1}\Psi_{\alpha } =(\Sigma ^{\ast })^{-1},
	\label{eq:Goal}
\end{align}
where the non-singularity of $\Sigma ^{\ast } $ follows from the non-singularity of $\Delta_{K}^{-1}$ and Assumption \ref{ass:Regularity}. To this end, define the following matrices
\begin{align}
A_{K} &\equiv ( \mathbf{I}_{d_{P}}-\Psi_{P}\Phi_{K,P0}) \Omega_{PP}( \mathbf{I}_{d_{P}}-\Psi_{P}\Phi_{K,P0}) ^{\prime } \notag \\
B_{K} &\equiv \left[ 
\begin{array}{c}
( \mathbf{I}_{d_{P}}+\Psi_{P}\Phi_{K,g}) \Psi_{g}\Omega_{gg}\Psi_{g}^{\prime }( \mathbf{I}_{d_{P}}+\Psi_{P}\Phi_{K,g}) ^{\prime } \\
-( \mathbf{I}_{d_{P}}+\Psi_{P}\Phi_{K,g}) \Psi_{g}\Omega_{Pg}^{\prime }( \mathbf{I}_{d_{P}}-\Psi_{P}\Phi_{K,P0}) ^{\prime }  \\
-( \mathbf{I}_{d_{P}}-\Psi_{P}\Phi_{K,P0}) \Omega_{Pg}\Psi_{g}^{\prime }( \mathbf{I}_{d_{P}}+\Psi_{P}\Phi_{K,g}) ^{\prime }
\end{array}
\right]  \notag \\
C_{K} &\equiv ( \mathbf{I}_{d_{P}}-\Psi_{P}\Phi_{K,P0}) ^{-1}B_{K}( \mathbf{I}_{d_{P}}-\Psi_{P}\Phi_{K,P0})'^{-1}, \label{eq:MatrixDefns}
\end{align}
where we have used that $( \mathbf{I}_{d_{P}}-\Psi_{P}\Phi_{K,P0}) $ is non-singular, which follows from Eq.\ \eqref{eq:Seq_P0_and_G}, Assumption \ref{ass:Regularity}, and that $\Lambda_{K}$ is non-singular. In turn, this implies that $A_{K}$ is non-singular and, in fact,
\begin{equation}
	A_{K}^{-1}=( \mathbf{I}_{d_{P}}-\Psi_{P}\Phi_{K,P0})'^{-1}\Omega_{PP}^{-1}( \mathbf{I}_{d_{P}}-\Psi_{P}\Phi_{K,P0}) ^{-1}.
	\label{eq:InverseOfA}
\end{equation}

The following derivation proves that $C_{K}=C_{1}$.
\begin{eqnarray*}
C_{K} &=&\left[ 
\begin{array}{c}
( \mathbf{I}_{d_{P}}-\Psi_{P}\Phi_{K,P0}) ^{-1}( \mathbf{I}_{d_{P}}+\Psi_{P}\Phi_{K,g}) \Psi_{g}\Omega_{gg}\Psi_{g}^{\prime }( \mathbf{I}_{d_{P}}+\Psi_{P}\Phi_{K,g}) ^{\prime }( \mathbf{I}_{d_{P}}-\Psi_{P}\Phi_{K,P0})'^{-1}\\
-( \mathbf{I}_{d_{P}}-\Psi_{P}\Phi_{K,P0}) ^{-1}( \mathbf{I }_{d_{P}}+\Psi_{P}\Phi_{K,g}) \Psi_{g}\Omega_{Pg}^{\prime } -\Omega_{Pg}\Psi_{g}^{\prime }( \mathbf{I}_{d_{P}}+\Psi_{P}\Phi_{K,g}) ^{\prime }( \mathbf{I}_{d_{P}}-\Psi_{P}\Phi_{K,P0})'^{-1}
\end{array}
\right] \\
&=&\left[ 
\begin{array}{c}
	( \mathbf{I}_{d_{P}}-\Psi_{P}) ^{-1}\Psi_{g}\Omega_{gg}\Psi_{g}^{\prime }( \mathbf{I}_{d_{P}}-\Psi_{P})'^{-1}\\
-( \mathbf{I}_{d_{P}}-\Psi_{P}) ^{-1}\Psi_{g}\Omega_{Pg}^{\prime }-\Omega_{Pg}\Psi_{g}^{\prime }( \mathbf{I}_{d_{P}}-\Psi_{P})'^{-1}
\end{array}
\right] =C_{1}.
\end{eqnarray*}
where the first equality uses Eq.\ \eqref{eq:MatrixDefns}, the second equality uses Lemma \ref{lem:Key_aux_MD}(b), and the final equality holds by Eq.\ \eqref{eq:MatrixDefns} with $K=1$.

Finally, Eq.\ \eqref{eq:Goal} follows immediately from the next derivation.
\begin{eqnarray*}
\Psi_{\alpha }^{\prime }\Delta_{K}^{-1}\Psi_{\alpha } &=&\Psi_{\alpha }^{\prime }( A_{K}+B_{K}) ^{-1}\Psi_{\alpha } \\
&=&\Psi_{\alpha }^{\prime }( A_{K}^{-1}-( \mathbf{I}_{d_{P}}+A_{K}^{-1}B_{K}) ^{-1}A_{K}^{-1}B_{K}A_{K}^{-1}) \Psi_{\alpha } \\
&=&\Psi_{\alpha }^{\prime }( \mathbf{I}_{d_{P}}-\Psi_{P}\Phi_{K,P0})'^{-1}\Omega_{PP}^{-1}\{ \Omega_{PP}-\Omega_{PP}( \mathbf{I}_{d_{P}}+\Omega_{PP}^{-1}C_{K}) ^{-1}\Omega_{PP}^{-1}C_{K}\} \Omega_{PP}^{-1}( \mathbf{I}_{d_{P}}-\Psi_{P}\Phi_{K,P0}) ^{-1}\Psi_{\alpha } \\
&=&\Psi_{\alpha }^{\prime }( \mathbf{I}_{d_{P}}-\Psi_{P})'^{-1}\Omega_{PP}^{-1}\{ \Omega_{PP}-\Omega_{PP}( \mathbf{ I}_{d_{P}}+\Omega_{PP}^{-1}C_{1}) ^{-1}\Omega_{PP}^{-1}C_{1}\} \Omega_{PP}^{-1}( \mathbf{I}_{d_{P}}-\Psi_{P}) ^{-1}\Psi_{\alpha } \\
&=&\Psi_{\alpha }^{\prime }\left( \left[ 
	\begin{array}{c}
	( \mathbf{I}_{d_{P}}-\Psi_{P}) ^{\prime } \\ 
	-\Psi_{g}^{\prime }
	\end{array}
	\right] ^{\prime } \left[ 
	\begin{array}{cc}
	\Omega_{PP} & \Omega_{Pg} \\ 
	\Omega_{Pg}^{\prime } & \Omega_{gg}
	\end{array}
	\right] \left[ 
	\begin{array}{c}
	( \mathbf{I}_{d_{P}}-\Psi_{P}) ^{\prime } \\ 
	-\Psi_{g}^{\prime }
	\end{array}
	\right]\right) ^{-1}\Psi_{\alpha }=\Psi_{\alpha }^{\prime }\Delta_{1}^{-1}\Psi_{\alpha } = (\Sigma ^{\ast })^{-1},
\end{eqnarray*}
where the first equality follows from $\Delta_{K}=A_{K}+B_{K}$, which is implied by combining Eqs.\ \eqref{eq:Delta_K} and \eqref{eq:MatrixDefns}, the second equality follows from $\Delta_{K}$ and $A_{K}$ being non-singular, the third equality follows from Eqs.\ \eqref{eq:MatrixDefns} and \eqref{eq:InverseOfA}, the fourth equality is based on Lemma \ref{lem:Key_aux_MD}(a) and $C_{K}=C_{1}$, the fifth equality follows from algebra and the final equality holds by Eq.\ \eqref{eq:Optimal_AVar_MD}.

Since the choice of $\{W_{k}:k\leq K-1\}$ was completely arbitrary, the proof of \underline{optimality} follows from showing the following argument. Note that
\begin{equation*}
\Sigma_{K-MD}( \hat{P}_{0},\{ W_{k}:k\leq K\} )
-\Sigma ^{\ast }=\left[ 
\begin{array}{c} 
(\Sigma_{K-MD}( \hat{P}_{0},\{ W_{k}:k\leq K\} ) -\Sigma_{K-MD}( \tilde{P},\{ W_{k}:k\leq K\} ) ) \\
+(\Sigma_{K-MD}( \tilde{P},\{ W_{k}:k\leq K\} ) -\Sigma_{K-MD}( \tilde{P},\{ \tilde{W}_{k}:k\leq K\} ) ) \\
+ \Sigma_{K-MD}( \tilde{P},\{ \tilde{W}_{k}:k\leq K\} ) -\Sigma ^{\ast }
\end{array}
\right] .
\end{equation*}
The RHS is the sum of three terms. The first term is PSD by Eq.\ \eqref{eq:MD_main_1}, the second term is PSD by Eq.\ \eqref{eq:MD_main_2}, and the third bracket is zero by Eq.\ \eqref{eq:Goal}. Then, $\Sigma_{K-MD}( \hat{P}_{0},\{ W_{k}:k\leq K\} ) -\Sigma ^{\ast }$ is PSD, as desired.
\end{proof}

%%%%%%%% DIVIDER %%%%%%%%%%%%
\subsection{Additional auxiliary results}
%%%%%%%% DIVIDER %%%%%%%%%%%%

\begin{theorem}[\textbf{General result for iterated extremum estimators}] \label{thm:2stepB}
Fix $K\in \mathbb{N}$ arbitrarily. Assume Assumptions \ref{ass:Regularity}, \ref{ass:Baseline}, and \ref{ass:EE_HL}. Then, for all $ k\leq K$,
\begin{align}
\sqrt{n}(\hat{\alpha}_{k}-\alpha ^{\ast })& ~=~\left[ 
\begin{array}{ccc}
(A_{k}+B_{k}\Upsilon_{k,P}) & B_{k}\Upsilon_{k,0} & (B_{k}\Upsilon_{k,g}+C_{k})
\end{array}
\right] \sqrt{{n}}\left[ 
\begin{array}{c}
\hat{P}-P^{\ast } \\ 
\hat{P}_{0}-P^{\ast } \\ 
\hat{g}-g^{\ast }
\end{array}
\right] +o_{p}(1)  \label{eq:Expansions_1} \\
\sqrt{n}(\hat{P}_{k-1}-P^{\ast })& ~=~\left[ 
\begin{array}{ccc}
\Upsilon_{k,P} & \Upsilon_{k,0} & \Upsilon_{k,g}
\end{array}
\right] \sqrt{{n}}\left[ 
\begin{array}{c}
\hat{P}-P^{\ast } \\ 
\hat{P}_{0}-P^{\ast } \\ 
\hat{g}-g^{\ast }
\end{array}
\right] +o_{p}(1),  \label{eq:Expansions_2} \\
\sqrt{n}(\hat{\theta}_{k}-\theta ^{\ast })~& =~\left[ 
\begin{array}{ccc}
(A_{k}+B_{k}\Upsilon_{k,P}) & B_{k}\Upsilon_{k,0} & (B_{k}\Upsilon_{k,g}+C_{k}) \\
\mathbf{0}_{d_{g\times }d_{P}} & \mathbf{0}_{d_{g\times }d_{P}} & \mathbf{I}_{d_{g}}
\end{array}
\right] \sqrt{{n}}\left[ 
\begin{array}{c}
\hat{P}-P^{\ast } \\ 
\hat{P}_{0}-P^{\ast } \\ 
\hat{g}-g^{\ast }
\end{array}
\right] +o_{p}(1),\label{eq:Expansions_3}
\end{align}
where $\hat{\theta}_{k}\equiv (\hat{\alpha}_{k},\hat{g})$, $\theta ^{\ast }\equiv (\alpha ^{\ast },g^{\ast })$, $\{A_{k}:k\leq K\}$, $\{B_{k}:k\leq K\} $, and $\{C_{k}:k\leq K\}$ are defined by:
\begin{align}
A_{k}& ~\equiv ~-\left( \frac{\partial ^{2}Q_{k}(\alpha ^{\ast },g^{\ast },P^{\ast })}{\partial \alpha \partial \alpha ^{\prime }}\right) ^{-1}\Xi_{k} \nonumber \\
B_{k}& ~\equiv ~-\left( \frac{\partial ^{2}Q_{k}(\alpha ^{\ast },g^{\ast },P^{\ast })}{\partial \alpha \partial \alpha ^{\prime }}\right) ^{-1}\left( \frac{\partial ^{2}Q_{k}(\alpha ^{\ast },g^{\ast },P^{\ast })}{\partial \alpha \partial P^{\prime }}\right) \nonumber \\
C_{k}& ~\equiv ~-\left( \frac{\partial ^{2}Q_{k}(\alpha ^{\ast },g^{\ast },P^{\ast })}{\partial \alpha \partial \alpha ^{\prime }}\right) ^{-1}\left( \frac{\partial ^{2}Q_{k}(\alpha ^{\ast },g^{\ast },P^{\ast })}{\partial \alpha \partial g^{\prime }}\right) , \label{eq:ABCdefns}
\end{align}
and $\{\Upsilon_{k,P}:k\leq K\}$, $\{\Upsilon_{k,g}:k\leq K\}$, and $ \{\Upsilon_{k,0}:k\leq K\}$ are iteratively defined as follows. Set $ \Upsilon_{1,P}\equiv \mathbf{0}_{d_{P}\times d_{P}}$, $\Upsilon_{1,g}\equiv \mathbf{0}_{d_{P}\times d_{g}}$, $\Upsilon_{1,0}\equiv \mathbf{ I}_{d_{P}}$ and, for any $k\leq K-1$,
\begin{align}
\Upsilon_{k+1,P}& ~\equiv ~(\Psi_{P}+\Psi_{\alpha }B_{k})\Upsilon_{k,P}+\Psi_{\alpha }A_{k} \nonumber \\
\Upsilon_{k+1,0}& ~\equiv ~(\Psi_{P}+\Psi_{\alpha }B_{k})\Upsilon_{k,0}. \nonumber \\
\Upsilon_{k+1,g}& ~\equiv ~(\Psi_{P}+\Psi_{\alpha }B_{k})\Upsilon_{k,g}+\Psi_{g}+\Psi_{\alpha }C_{k}. \label{eq:Upsilondefns}
\end{align}
\end{theorem}
%%%%%%%% DIVIDER %%%%%%%%%%%%
\begin{proof}
We divide the proof into three steps.

\underline{Step 1.} Show that $( \hat{\alpha}_{k},\hat{P}_{k-1}) =( \alpha ^{\ast },P^{\ast }) +o_{p}( 1) $ for any $k\leq K$. We prove this by induction.

We begin with the initial step, i.e., show that the result holds for $k=1$. First, $\hat{P}_{0}=P^{\ast }+o_{p}( 1) $ follows directly from Assumption \ref{ass:Baseline}. Assumptions \ref{ass:Baseline} and \ref{ass:EE_HL} imply that $\sup_{\alpha \in \Theta_{\alpha}}|{\hat{Q}}_{1}(\alpha ,\hat{g},\hat{P}_{0})-{Q}_{1}(\alpha ,g^{\ast },P^{\ast })|=o_{p}(1)$, ${Q}_{1}(\alpha ,g^{\ast },P^{\ast })$ is upper semi-continuous function of $\alpha$, and ${Q }_{1}(\alpha ,g^{\ast },P^{\ast })$ is uniquely maximized at $\alpha ^{\ast } $. From these conditions, $\hat{\alpha}_{1}=\alpha ^{\ast }+o_{p}( 1)$ follows from standard results for extremum estimators.
 % (e.g.\ \cite{mcfadden/newey:1994}).

We next show the inductive step, i.e., assume that the result holds for $ k\leq K-1$ and show that it holds for $k+1$. First, notice that:
\begin{eqnarray*}
\hat{P}_{k}-P^{\ast } &=&\Psi ( \hat{\alpha}_{k},\hat{g},\hat{P}_{k-1}) -\Psi ( \alpha ^{\ast },g^{\ast },P^{\ast }) \\
&=&\Psi_{\alpha }( \alpha ^{\ast },g^{\ast },P^{\ast }) ( \hat{\alpha}_{k}-\alpha ^{\ast }) +\Psi_{g}( \alpha ^{\ast },g^{\ast },P^{\ast }) ( \hat{g}-g^{\ast }) +\Psi_{P}( \alpha ^{\ast },g^{\ast },P^{\ast }) ( \hat{P}_{k-1}-P^{\ast }) +o_{p}( 1) =o_{p}( 1) ,
\end{eqnarray*}
where the second line follows from the intermediate value theorem, the inductive hypothesis, and Assumptions \ref{ass:Regularity} and \ref {ass:Baseline}. Assumptions \ref{ass:Baseline} and \ref{ass:EE_HL} imply that $\sup_{\alpha \in \Theta_{\alpha}}|{\hat{Q}}_{k+1}(\alpha ,\hat{g},\hat{ P}_{0})-{Q}_{k+1}(\alpha ,g^{\ast },P^{\ast })|=o_{p}(1)$, ${Q}_{k+1}(\alpha ,g^{\ast },P^{\ast })$ is upper semi-continuous function of $\alpha$, and ${Q }_{k+1}(\alpha ,g^{\ast },P^{\ast })$ is uniquely maximized at $\alpha ^{\ast }$. By repeating previous arguments, $\hat{\alpha}_{k+1}=\alpha ^{\ast }+o_{p}( 1) $ follows.

\underline{Step 2.} Derive an expansion for $\sqrt{n}( \hat{\alpha}_{k}-\alpha ^{\ast })$ for any $k\leq K$.

For any $k\leq K$, consider the following derivation.
\begin{align}
\mathbf{0}_{d_{\alpha }\times 1} &=\sqrt{n}\frac{\partial \hat{Q}_{k}( \hat{ \alpha}_{k},\hat{g},\hat{P}_{k-1}) }{\partial \alpha } +o_{p}( 1) \nonumber \\
&=\left[ 
\begin{array}{c}
\sqrt{n}\frac{\partial \hat{Q}_{k}( \alpha ^{\ast },g^{\ast },P^{\ast }) }{ \partial \alpha }+\frac{\partial ^{2}Q_{k}( \alpha ^{\ast },g^{\ast },P^{\ast }) }{\partial \alpha \partial \alpha ^{\prime }} \sqrt{n}( \hat{ \alpha}_{k}-\alpha ^{\ast }) + \\
\frac{\partial ^{2}Q_{k}( \alpha ^{\ast },g^{\ast },P^{\ast }) }{ \partial \alpha \partial P^{\prime }}\sqrt{n}( \hat{P}_{k-1}-P^{\ast }) +\frac{ \partial ^{2}Q_{k}( \alpha ^{\ast },g^{\ast },P^{\ast }) }{\partial \alpha \partial g^{\prime }}\sqrt{n}( \hat{g} -g^{\ast })
\end{array}
\right] +o_{p}( 1)  \nonumber \\
&=\left[ 
\begin{array}{c}
\Xi_{k}\sqrt{{n}}( \hat{P}-P^{\ast }) +\frac{\partial ^{2}Q_{k}( \alpha ^{\ast },g^{\ast },P^{\ast }) }{\partial \alpha \partial \alpha ^{\prime }} \sqrt{n}( \hat{\alpha}_{k}-\alpha ^{\ast }) + \\
\frac{\partial ^{2}Q_{k}( \alpha ^{\ast },g^{\ast },P^{\ast }) }{ \partial \alpha \partial P^{\prime }}\sqrt{n}( \hat{P}_{k-1}-P^{\ast }) +\frac{ \partial ^{2}Q_{k}( \alpha ^{\ast },g^{\ast },P^{\ast }) }{\partial \alpha \partial g^{\prime }}\sqrt{n}( \hat{g} -g^{\ast }),
\end{array}
\right] +o_{p}( 1) ,  \label{eq:PreExpansionTheta}
\end{align}
where the first line holds because $( \hat{\alpha}_{k},\hat{g},\hat{P }_{k-1}) =( \alpha ^{\ast },g^{\ast },P^{\ast }) +o_{p}( 1) $ (due to the step 1 and Assumption \ref{ass:Baseline}), $\hat{\alpha}_{k}$ is the maximizer of $\hat{Q}_{k}( \alpha ,\hat{g},\hat{P}_{k-1}) $ in $\Theta_{\alpha }$, and $\hat{\alpha}_{k}$ belongs to the interior of $\Theta_{\alpha }$ with probability approaching one (due to the preliminary result and Assumption \ref {ass:Regularity}), the second line holds by the intermediate value theorem and elementary convergence arguments based on Assumption \ref{ass:EE_HL}, and the third line holds by Assumption \ref{ass:EE_HL}.

We are now ready to derive the desired expansion.
\begin{align}
&\sqrt{n}( \hat{\alpha}_{k}-\alpha ^{\ast }) ~~= \nonumber \\ 
& - \frac{\partial ^{2}Q_{k}( \alpha ^{\ast },g^{\ast },P^{\ast }) }{ \partial \alpha \partial \alpha ^{\prime }} ^{-1}\left[ \Xi_{k}\sqrt{{n}}( \hat{P}-P^{\ast }) +\frac{\partial ^{2}\hat{Q}_{k}( \alpha ^{\ast },g^{\ast },P^{\ast }) }{\partial \alpha \partial P^{\prime }}\sqrt{n}( \hat{P}_{k-1}-P^{\ast }) +\frac{ \partial ^{2}\hat{Q}_{k}( \alpha ^{\ast },g^{\ast },P^{\ast }) }{ \partial \alpha \partial g^{\prime }}\sqrt{n}( \hat{g} -g^{\ast }) \right] \nonumber \\
& =~~ A_{k}\sqrt{{n}}( \hat{P}-P^{\ast }) + B_{k}\sqrt{n}( \hat{P }_{k-1}-P^{\ast }) + C_{k}\sqrt{n}( \hat{g}-g^{\ast }) +o_{p}( 1) , \label{eq:ExpansionTheta}
\end{align}
where the first line holds by Eq.\ \eqref{eq:PreExpansionTheta} and Assumption \ref{ass:EE_HL}, and the second line holds by Eq.\ \eqref{eq:ABCdefns} and Assumption \ref{ass:EE_HL}.

\underline{Step 3.} Show Eqs.\ \eqref{eq:Expansions_1}, \eqref{eq:Expansions_2}, and \eqref{eq:Expansions_3}. Eq. \eqref{eq:Expansions_3} follows immediately from Eq.\ \eqref{eq:Expansions_1}. Eqs.\ \eqref{eq:Expansions_1} and \eqref{eq:Expansions_2} are the result of the following inductive argument.

We begin with the initial step, i.e., show that the result holds for $k=1$. By $\Upsilon_{1,P}\equiv \mathbf{0}_{d_{P}\times d_{P}}$, $\Upsilon_{1,g}\equiv \mathbf{0}_{d_{P}\times d_{g}}$, $\Upsilon_{1,0}\equiv \mathbf{I}_{d_{P}}$, Eq.\ \eqref{eq:Expansions_2} holds for $k=1$. By the same argument and step 2, Eq.\ \eqref{eq:Expansions_1} holds for $k=1$.

We next show the inductive step, i.e., assume that the result holds for $ k\leq K-1$ and show that it holds for $k+1$. First, consider the following derivation:
\begin{align}
& \sqrt{n}(\hat{P}_{k}-P^{\ast })=\Psi_{\alpha }\sqrt{n}(\hat{\alpha}_{k}-\alpha ^{\ast })+\Psi_{g}\sqrt{n}(\hat{g}-g^{\ast })+\Psi_{P}\sqrt{n}( \hat{P}_{k-1}-P^{\ast })+o_{p}(1) \nonumber \\
& =\left[ 
\begin{array}{c}
\lbrack \Psi_{\alpha }A_{k}+(\Psi_{P}+\Psi_{\alpha }B_{k})\Upsilon_{k,P}] \sqrt{{n}}(\hat{P}-P^{\ast })+[\Psi_{\alpha }C_{k}+\Psi_{g}+(\Psi_{P}+\Psi_{\alpha }B_{k})\Upsilon_{k,g}]\sqrt{n}(\hat{g}-g^{\ast }) \\
+(\Psi_{P}+\Psi_{\alpha }B_{k})\Upsilon_{k,0}\sqrt{{n}}(\hat{P}_{0}-P^{\ast })
\end{array}
\right] +o_{p}(1)  \nonumber \\
& =\Upsilon_{k+1,P}\sqrt{{n}}(\hat{P}-P^{\ast })+\Upsilon_{k+1,g}\sqrt{n}( \hat{g}-g^{\ast })+\Upsilon_{k+1,0}\sqrt{{n}}(\hat{P}_{0}-P^{\ast })+o_{p}(1), \label{eq:Expansion_P_kp1}
\end{align}
where the first equality holds by $\hat{P}_{k}\equiv \Psi (\hat{\alpha}_{k}, \hat{g},\hat{P}_{k-1})$, $P^{\ast }=\Psi (\alpha ^{\ast },g^{\ast },P^{\ast })$, Assumption \ref{ass:Regularity}, and the intermediate value theorem, the second line holds by step 2 and the inductive hypothesis, and the last equality holds by Eq.\ \eqref{eq:Upsilondefns}. This verifies Eq.\ \eqref{eq:Expansions_2} for $k+1$. Second, consider the following derivation:
\begin{align*}
& \sqrt{n}(\hat{\alpha}_{k+1}-\alpha ^{\ast })=A_{k+1}\sqrt{{n}}(\hat{P} -P^{\ast })+B_{k+1}\sqrt{n}(\hat{P}_{k}-P^{\ast })+C_{k+1}\sqrt{n}(\hat{g} -g^{\ast })+o_{p}(1) \\
& =(A_{k+1}+B_{k+1}\Upsilon_{k+1,P})\sqrt{{n}}(\hat{P}-P^{\ast })+(C_{k+1}+B_{k+1}\Upsilon_{k+1,g})\sqrt{n}(\hat{g}-g^{\ast })+B_{k+1}\Upsilon_{k+1,0}\sqrt{{n}}(\hat{P}_{0}-P^{\ast })+o_{p}(1),
\end{align*}
where the first equality holds by step 2 and the second equality follows from Eq.\ \eqref{eq:Expansion_P_kp1}. This verifies Eq.\ \eqref{eq:Expansions_1} for $k+1$, and completes the proof.
\end{proof}

%%%%%%%% DIVIDER %%%%%%%%%%%%
\begin{lemma}\label{lem:Key_aux_MD}
Assume the conditions in Theorem \ref{thm:MD_main}. Let $\{ \Phi_{k,0}:k\leq K\} $, $\{ \Phi_{k,P}:k\leq K\} $, and $\{ \Phi_{k,g}:k\leq K\} $ defined as in Eq.\ \eqref{eq:coefficients_MD}, and let $\Phi_{k,P0}\equiv \Phi_{k,P}+\Phi_{k,0}$ for all $k\leq K$. Then,
\begin{enumerate}
\item $( \mathbf{I}_{d_{P}}-\Psi_{P}\Phi_{K,P0}) ^{-1}\Psi_{\alpha }=( \mathbf{I}_{d_{P}}-\Psi_{P}) ^{-1}\Psi_{\alpha }$.
\item $( \mathbf{I}_{d_{P}}-\Psi_{P}\Phi_{K,P0}) ^{-1}( \mathbf{I}_{d_{P}}+\Psi_{P}\Phi_{K,g}) =( \mathbf{I}_{d_{P}}-\Psi_{P}) ^{-1}$.
\end{enumerate}
\end{lemma}
%%%%%%%% DIVIDER %%%%%%%%%%%%
\begin{proof}
Throughout this proof, denote $\Pi_{k}\equiv \Psi_{\alpha }( \Psi_{\alpha }^{\prime }W_{k}\Psi_{\alpha }) ^{-1}\Psi_{\alpha }^{\prime }W_{k}$ for all $k\leq K$.

\underline{Part 1.} It suffices to show that $( \mathbf{I}_{d_{P}}-\Psi_{P}\Phi_{k,P0}) ( \mathbf{I}_{d_{P}}-\Psi_{P}) ^{-1}\Psi_{\alpha }=\Psi_{\alpha }$ for $k\leq K$. We show this by induction. The initial step follows from $\Phi_{k,P0}=\mathbf{I}_{d_{P}}$. We next show the inductive step, i.e., assume the result holds for $k\leq K-1$ and show it also holds for $k+1$. Consider the following derivation.
\begin{eqnarray*}
( \mathbf{I}_{d_{P}}-\Psi_{P}\Phi_{k+1,P0}) ( \mathbf{I}_{d_{P}}-\Psi_{P}) ^{-1}\Psi_{\alpha } 
% &=&( \mathbf{I}_{d_{P}}-\Psi_{P}\Pi_{k}-\Psi_{P}( \mathbf{I}_{d_{P}}-\Pi_{k}) \Psi_{P}\Phi_{k,P0}) ( \mathbf{I}_{d_{P}}-\Psi_{P}) ^{-1}\Psi_{\alpha } \\
&=&( \mathbf{I}_{d_{P}}-\Psi_{P}+\Psi_{P}( \mathbf{I}_{d_{P}}-\Pi_{k}) ( \mathbf{I}_{d_{P}}-\Psi_{P}\Phi_{k,P0}) ) ( \mathbf{I}_{d_{P}}-\Psi_{P}) ^{-1}\Psi_{\alpha } \\
&=&\Psi_{\alpha }+\Psi_{P}( \mathbf{I}_{d_{P}}-\Pi_{k}) ( \mathbf{I}_{d_{P}}-\Psi_{P}\Phi_{k,P0}) ( \mathbf{I}_{d_{P}}-\Psi_{P}) ^{-1}\Psi_{\alpha } \\
&=&\Psi_{\alpha }+\Psi_{P}( \mathbf{I}_{d_{P}}-\Pi_{k}) \Psi_{\alpha }=\Psi_{\alpha },
\end{eqnarray*}
as required, where the first equality follows from Eq.\ \eqref{eq:coefficients_MD} and some algebra, the second equality follows from the inductive hypothesis, and the final equality follows from $\Pi_{k}\Psi_{\alpha }=\Psi_{\alpha }$.

\underline{Part 2.} It suffices to show that $( \mathbf{I}_{d_{P}}+\Psi_{P}\Phi_{k,g}) ( \mathbf{I}_{d_{P}}-\Psi_{P}) =( \mathbf{I}_{d_{P}}-\Psi_{P}\Phi_{k,P0}) $ for $k\leq K$. We show this by induction. The initial step follows from $\Phi_{0,g}=\mathbf{0}_{d_{P}\times d_{g}}$ and $\Phi_{0,P}=\mathbf{I}_{d_{P}}$. We next show the inductive step, i.e., assume the result holds for $k\leq K-1$ and show it also holds for $k+1$. Consider the following derivation.
\begin{eqnarray*}
( \mathbf{I}_{d_{P}}+\Psi_{P}\Phi_{k+1,g}) ( \mathbf{I}_{d_{P}}-\Psi_{P}) &=&( \mathbf{I}_{d_{P}}-\Psi_{P}) +\Psi_{P}( \mathbf{I}_{d_{P}}-\Pi_{k}) ( \Psi_{P}\Phi_{k,g}+\mathbf{I}_{d_{P}}) ( \mathbf{I}_{d_{P}}-\Psi_{P}) \\
&=&( \mathbf{I}_{d_{P}}-\Psi_{P}) +\Psi_{P}( \mathbf{I}_{d_{P}}-\Pi_{k}) ( \mathbf{I}_{d_{P}}-\Psi_{P}\Phi_{k,P0}) \\
&=& -\Psi_{P}( \mathbf{I}_{d_{P}}-\Pi_{k}) \Psi_{P}\Phi_{k,P0}+\mathbf{I}_{d_{P}}-\Psi_{P}\Pi_{k} \\
&=& \mathbf{I}_{d_{P}}-\Psi_{P}\Phi_{k+1,P0},
\end{eqnarray*}
where the first and fourth equalities follows from Eq.\ \eqref{eq:coefficients_MD}, the second equality follows from the inductive hypothesis, and the third equality follows from algebra.
% \begin{eqnarray*}
% \mathbf{I}_{d_{P}}-\Psi_{P}\Phi_{k+1,P0} &=&-\Psi_{P}( \mathbf{I}_{d_{P}}-\Pi_{k}) \Psi_{P}\Phi_{k,P0}+\mathbf{I}_{d_{P}}-\Psi_{P}\Pi_{k} \\
% &=&( \mathbf{I}_{d_{P}}-\Psi_{P}) +\Psi_{P}( \mathbf{I}_{d_{P}}-\Pi_{k}) ( -\Psi_{P}\Phi_{k,P0}) -\Psi_{P}\Pi_{k}+\Psi_{P} \\
% &=&( \mathbf{I}_{d_{P}}-\Psi_{P}) +\Psi_{P}( \mathbf{I}_{d_{P}}-\Pi_{k}) ( \mathbf{I}_{d_{P}}-\Psi_{P}\Phi_{k,P0}) -\Psi_{P}( \mathbf{I}_{d_{P}}-\Pi_{k}) -\Psi_{P}\Pi_{k}+\Psi_{P} \\
% &=&( \mathbf{I}_{d_{P}}-\Psi_{P}) +\Psi_{P}( \mathbf{I}_{d_{P}}-\Pi_{k}) ( \mathbf{I}_{d_{P}}-\Psi_{P}\Phi_{k,P0}) ,
% \end{eqnarray*}
% where the equality follows from Eq.\ \eqref{eq:coefficients_MD}. Thus, $( \mathbf{I}_{d_{P}}+\Psi_{P}\Phi_{k+1,g}) ( \mathbf{I}_{d_{P}}-\Psi_{P}) =( \mathbf{I}_{d_{P}}-\Psi_{P}\Phi_{k+1,P0}) $, as required.
\end{proof}
%%%%%%%% DIVIDER %%%%%%%%%%%%

\begin{lemma}\label{lem:PSD} 
Under Assumptions \ref{ass:Baseline} and \ref{ass:Baseline2},
\begin{equation*}
\left( 
\begin{array}{ccc}
\Omega_{PP} & \Omega_{P0} & \Omega_{Pg} \\ 
\Omega_{P0}^{\prime } & \Omega_{00} & \Omega_{0g} \\ 
\Omega_{Pg}^{\prime } & \Omega_{0g}^{\prime } & \Omega_{gg}
\end{array}
\right) - \left( 
\begin{array}{ccc}
\Omega_{PP} & \Omega_{PP} & \Omega_{Pg} \\ 
\Omega_{PP} & \Omega_{PP} & \Omega_{Pg} \\ 
\Omega_{Pg}^{\prime } & \Omega_{Pg}^{\prime } & \Omega_{gg}
\end{array}
\right)\text{~~~~is PSD.}
\end{equation*}
\end{lemma}
%%%%%%%% DIVIDER %%%%%%%%%%%%
\begin{proof}
First, note that Assumption \ref{ass:Baseline} implies  
\[
\sqrt{{n}}\left( 
\begin{array}{c}
\hat{P}_{0}-P^{\ast } \\ 
\hat{g}-g^{\ast }
\end{array}
\right) \overset{d}{\to}N\left(\left( 
\begin{array}{c}
\mathbf{0}_{d_{P}\times 1} \\ 
\mathbf{0}_{d_{g}\times 1}
\end{array}
\right) ,\left( 
\begin{array}{cc}
\Omega_{00} & \Omega_{0g} \\ 
\Omega_{0g}^{\prime } & \Omega_{gg}
\end{array}
\right) \right) .
\]
Second, note that Assumption \ref{ass:Baseline2} implies that  
\[
\left( 
\begin{array}{cc}
\Omega_{00} & \Omega_{0g} \\ 
\Omega_{0g}^{\prime } & \Omega_{gg}
\end{array}
\right) -\left( 
\begin{array}{cc}
\Omega_{PP} & \Omega_{Pg} \\ 
\Omega_{Pg}^{\prime } & \Omega_{gg}
\end{array}
\right)
\]
is PSD, i.e, for any $\gamma_{P}\in \mathbb{R} ^{d_{P}}$ and $\gamma_{g}\in \mathbb{R} ^{d_{g}}$, \begin{equation} \gamma_{P}^{\prime }( \Omega_{00}-\Omega_{PP}) \gamma_{P}+2\gamma_{P}^{\prime }( \Omega_{0g}-\Omega_{Pg}) \gamma_{g}\geq 0. \label{eq:PSD1}
\end{equation}

The remainder of this proof follows arguments similar to those used to show \citet[Lemma 2.1]{hausman:1978}. Fix $r\in \mathbb{R}$ and $A\in \mathbb{R} ^{d_{P}\times d_{P}}$ arbitrarily. By Assumption \ref{ass:Baseline},
\[
\sqrt{{n}}\left( 
\begin{array}{c}
\hat{P}+rA(\hat{P}_{0}-\hat{P})-P^{\ast } \\ 
\hat{g}-g^{\ast }
\end{array}
\right) =\sqrt{{n}}\left( 
\begin{array}{c}
(\mathbf{I}_{d_{P}}-rA)(\hat{P}-P^{\ast })+rA(\hat{P}_{0}-P^{\ast }) \\ 
\hat{g}-g^{\ast }
\end{array}
\right) \overset{d}{\rightarrow }N\left( \left( 
\begin{array}{c}
\mathbf{0}_{d_{P}} \\ 
\mathbf{0}_{d_{g}}
\end{array}
\right) ,\Sigma \right) ,
\]
and 
\[
\Sigma \equiv \left( 
\begin{array}{cc}
\left( 
\begin{array}{c}
r^{2}A\Omega_{00}A^{\prime }+(\mathbf{I}_{d_{P}}-rA)\Omega_{PP}(\mathbf{I}_{d_{P}}-rA^{\prime }) \\
+r(\mathbf{I}_{d_{P}}-rA)\Omega_{P0}A^{\prime }+rA\Omega_{P0}^{\prime }( \mathbf{I}_{d_{P}}-rA^{\prime })
\end{array}
\right)  & rA(\Omega_{0g}-\Omega_{Pg})+\Omega_{Pg} \\ 
r(\Omega_{0g}-\Omega_{Pg})^{\prime }A^{\prime }+\Omega_{Pg}^{\prime } & \Omega_{gg}
\end{array}
\right) .
\]
Assumption \ref{ass:Baseline2} implies that 
\[
\Sigma -\left( 
\begin{array}{cc}
\Omega_{PP} & \Omega_{Pg} \\ 
\Omega_{Pg}^{\prime } & \Omega_{gg}
\end{array}
\right) =\left( 
\begin{array}{cc}
\left( 
\begin{array}{c}
r^{2}A\Omega_{00}A^{\prime }+(\mathbf{I}_{d_{P}}-rA)\Omega_{PP}(\mathbf{I}_{d_{P}}-rA^{\prime })+ \\
r(\mathbf{I}_{d_{P}}-rA)\Omega_{P0}A^{\prime }+rA\Omega_{P0}^{\prime }( \mathbf{I}_{d_{P}}-rA^{\prime })-\Omega_{PP}
\end{array}
\right) & rA(\Omega_{0g}-\Omega_{Pg}) \\
r(\Omega_{0g}-\Omega_{Pg})^{\prime }A^{\prime } & \mathbf{0}_{d_{g}\times d_{g}}
\end{array}
\right) 
\]
is PSD, i.e., for any $\lambda_{P}\in \mathbb{R}^{d_{P}}$ and $\lambda_{g}\in \mathbb{R}^{d_{g}}$,
\begin{equation}
H(r)\equiv \left( 
\begin{array}{c}
\lambda_{P}^{\prime }(r^{2}A\Omega_{00}A^{\prime }+(\mathbf{I}_{d_{P}}-rA)\Omega_{PP}(\mathbf{I}_{d_{P}}-rA^{\prime })+r(\mathbf{I}_{d_{P}}-rA)\Omega_{P0}A^{\prime }+rA\Omega_{P0}^{\prime }(\mathbf{I}_{d_{P}}-rA^{\prime })-\Omega_{PP})\lambda_{P} \\ +2r\lambda_{P}^{\prime }(\Omega_{0g}-\Omega_{Pg})^{\prime }A^{\prime }\lambda_{g}
\end{array}
\right) \geq 0.  \label{eq:PSD2}
\end{equation}
Note that $H(0)=0$, i.e., $H(r)$ achieves a minimum at $r=0$. Then, the first order condition for a minimization has to be satisfied at $r=0$, which implies
\begin{equation}
H^{\prime }(0)=\lambda_{P}^{\prime }(\Omega_{P0}A^{\prime }+A\Omega_{P0}^{\prime }-A\Omega_{PP}-\Omega_{PP}A^{\prime })\lambda_{P}+2\lambda_{g}^{\prime }(\Omega_{0g}-\Omega_{Pg})^{\prime }A^{\prime }\lambda_{P}=0.
\label{eq:PSD3}
\end{equation}
Since Eq.\ \eqref{eq:PSD3} has to hold for $\lambda_{g}=\mathbf{0}_{d_{g}}$, $A=\mathbf{I}_{d_{P}}$, and all $\lambda_{P}\in \mathbb{R}^{d_{P}}$, we deduce that
\begin{equation}
2\Omega_{PP}=\Omega_{P0}+\Omega_{P0}^{\prime }.  \label{eq:PSD4}
\end{equation}
Plugging this information into Eq.\ \eqref{eq:PSD3} yields
\begin{equation}
H^{\prime }(0)=\lambda_{P}^{\prime }((\Omega_{P0}-\Omega_{P0}^{\prime
})A^{\prime }+A(\Omega_{P0}^{\prime }-\Omega_{P0}))\lambda_{P}/2+2\lambda
_{P}^{\prime }(\Omega_{0g}-\Omega_{Pg})^{\prime }A^{\prime }\lambda_{g}=0.
\label{eq:PSD5}
\end{equation}
Since Eq.\ \eqref{eq:PSD5} has to hold for $\lambda_{g}=\mathbf{0}_{d_{g}}$, $ A=\Omega_{P0}-\Omega_{P0}^{\prime }$ and all $\lambda_{P}\in \mathbb{R}^{d_{P}}$, we deduce that $ \Omega_{P0}'=\Omega_{P0}$. If we combine this with Eq.\ \eqref{eq:PSD4}, we conclude that
\begin{equation}
\Omega_{PP}=\Omega_{P0}=\Omega_{P0}^{\prime }.  \label{eq:PSD6}
\end{equation}
Plugging this information into Eq.\ \eqref{eq:PSD5} yields
\begin{equation}
H^{\prime }(0)=2\lambda_{P}^{\prime }(\Omega_{0g}-\Omega_{Pg})^{\prime }A^{\prime }\lambda_{g}=0. \label{eq:PSD7}
\end{equation}
Since Eq.\ \eqref{eq:PSD7} has hold for $A=\mathbf{I}_{d_{P}}$ and all $\lambda_{P}\in \mathbb{R}^{d_{P}}$ and $\lambda_{g}\in \mathbb{R}^{d_{g}}$, we conclude that
\begin{equation}
\Omega_{0g}=\Omega_{Pg}. \label{eq:PSD8}
\end{equation}

For any $\mu =(\mu_{P}^{\prime },{\mu}_{0}^{\prime },\mu_{g}^{\prime })^{\prime }$ with $\mu_{P}, \mu_{0}\in \mathbb{R} ^{d_{P}}$ and $\mu_{g}\in \mathbb{R}^{d_{g}}$, consider the following argument.
\begin{align*}
\mu ^{\prime }\left[ \left( 
\begin{array}{ccc}
\Omega_{PP} & \Omega_{P0} & \Omega_{Pg} \\ 
\Omega_{P0}^{\prime } & \Omega_{00} & \Omega_{0g} \\ 
\Omega_{Pg}^{\prime } & \Omega_{0g}^{\prime } & \Omega_{gg}
\end{array}
\right) -\left( 
\begin{array}{ccc}
\Omega_{PP} & \Omega_{PP} & \Omega_{Pg} \\ 
\Omega_{PP} & \Omega_{PP} & \Omega_{Pg} \\ 
\Omega_{Pg}^{\prime } & \Omega_{Pg}^{\prime } & \Omega_{gg}
\end{array}
\right) \right] \mu  
={\mu}_{0}^{\prime }(\Omega_{00}-\Omega_{PP}){\mu}_{0}\geq 0,
\end{align*}
where the equality uses Eqs.\ \eqref{eq:PSD6} and \eqref{eq:PSD8}, and the inequality uses Eq.\ \eqref{eq:PSD1} for $\gamma_{P}={\mu}_{0}$ and $\gamma_g = {\bf 0}_{d_g \times 1}$. Since the choice of $\mu$ was arbitrary, the desired result follows.
\end{proof}
%%%%%%%% DIVIDER %%%%%%%%%%%%

\subsection{Comparison with maximum likelihood estimator}\label{sec:MLE}

If we could abstract from the complications of computational complexity, the parameters of the model in Section \ref{sec:Model} could be estimated via MLE.\footnote{For example, the MLE could be computed via MPEC method proposed in \cite{su/judd:2012} or approximated via the one-step MLE in \citet[Section 3.6]{aguirregabiria/mira:2007}.} The goal of this section is to derive the asymptotic distribution of the MLE of $\alpha ^{\ast }$, and compare it with that of the optimal $1$-MD estimator considered in Section \ref{sec:MD}. Under reasonable conditions, we show that the MLE of $ \alpha ^{\ast }$ is asymptotically normal and its variance-covariance matrix coincides with that of the optimal $1$-MD estimator. By the well-known results in econometrics and statistics, the MLE is an efficient estimator for general classes of estimators (see, e.g., \citet[Section 5]{mcfadden/newey:1994}). As a corollary, under these conditions, the optimal $1$-MD estimator is also the efficient estimator of $ \alpha ^{\ast }$.

The MLE estimator has been studied in \citet[Section 3.2]{aguirregabiria/mira:2007} in the context of dynamic games, and in \citet[Section 3]{aguirregabiria/mira:2002} and \citet[Section 3]{kasahara/shimotsu:2008} in the context of dynamic single-agent problems. Relative to some of these references, our analysis takes into account the effect of the transition probabilities and the marginal state distribution on the likelihood function. To achieve this goal, we introduce the following notation. We use $ \Pi _{\theta }=\{\Pi _{\theta }(\vec{a},x^{\prime },x):(\vec{a},x^{\prime },x)\in A\times X\times X\}$, where $\Pi _{\theta }(\vec{a},x,x^{\prime })$ denotes the model-implied probability that players choose action $\vec{a}$ and the current state evolves from $x$ to $x^{\prime }$ for a generic parameter $\theta =(\alpha ,g)\in \Theta =\Theta _{\alpha }\times \Theta _{g} $. As the notation suggests, $\Pi _{\theta }$ is the model-implied analog of the DGP $\Pi ^{\ast }$ introduced in Assumption \ref{ass:iid}. This allows us to deduce the model-implied CCPs, transition probabilities, and marginal state distribution. For all $(\vec{a},j,x,x^{\prime })\in A^{|J|}\times J\times X\times X$ and $\vec{a}=(a,\vec{a}_{-j})$, we denote the model-implied CCPs, transition probabilities, and marginal state distribution for each $\theta $ as follows:
\begin{align}
P_{\theta ,j}(a|x)~& \equiv ~\frac{\sum_{(\vec{a}_{-j},x^{\prime })\in A^{|J|-1}\times X}\Pi _{\theta }((a,\vec{a}_{-j}),x,x^{\prime })}{\sum_{( \vec{a},x^{\prime })\in A^{|J|}\times X}\Pi _{\theta }(\vec{a},x,x^{\prime }) } \notag \\
\Lambda _{\theta }(x^{\prime }|x,\vec{a})~& \equiv ~\frac{\Pi _{\theta }( \vec{a},x,x^{\prime })}{\sum_{(\vec{a},x)\in A^{|J|}\times X}\Pi _{\theta }( \vec{a},x,x^{\prime })}, \notag \\
m_{\theta }(x)~& \equiv ~\sum_{(\vec{a},x^{\prime })\in A^{|J|}\times X}\Pi _{\theta }(\vec{a},x,x^{\prime }). \label{eq:ModelImpliedDistributions}
\end{align}
Our notation in Eq.\ \eqref{eq:ModelImpliedDistributions} requires that model-implied CCPs are uniquely defined for each $\theta $. Note that this is stronger than assuming that the data are sampled from a unique equilibrium, as imposed in Assumption \ref{ass:iid}. In principle, this restriction could be removed by considering the ideas in \citet[Eq.\ (26)]{aguirregabiria/mira:2007}, but we do not pursue this extension here for simplicity.

Under these restrictions, the MLE of $\theta ^{\ast }$, denoted by $\hat{ \theta}_{MLE}$, is given by
\begin{equation}
\hat{\theta}_{MLE}~=~\underset{\theta \in \Theta}{\arg \max}~\frac{1}{n}
\sum_{i=1}^{n}\left\{ 
\begin{array}{c}
\sum_{(j,a,x)\in J\times A\times X}\ln P_{\theta ,j}(a|x)1[(x_{i},a_{j,i})=(x,a)]+ \\
\sum_{(\vec{a},x,x^{\prime })\in A^{|J|}\times X\times X}\ln \Lambda _{\theta }(x^{\prime }|x,\vec{a})1[(x_{i}^{\prime },x_{i},\vec{a} _{i})=(x^{\prime },x,\vec{a})] \\
+\sum_{x\in X}\ln m_{\theta }(x)1[x_{i}=x]
\end{array}
\right\}   \label{eq:avg_log_lik}
\end{equation}
Under standard regularity conditions (see, e.g., \citet[page 120]{amemiya:1985}), $\hat{\theta}_{MLE}$ is the solution to the first-order condition of the MLE problem, which can be expressed as follows:
\begin{equation*}
\frac{1}{n}\sum_{i=1}^{n}[T_{i,1}(\hat{\theta}_{MLE})+T_{i,2}(\hat{\theta}_{MLE})+T_{i,3}(\hat{\theta}_{MLE})]=\mathbf{0}_{d_{\theta }\times 1},
\end{equation*}
where 
\begin{align*}
T_{i,1}(\theta )& \equiv \sum_{(j,a,x)\in J\times \tilde{A}\times X}\frac{ \partial P_{\theta ,j}(a|x)}{\partial \theta }\left( \frac{ 1[(x_{i},a_{j,i})=(x,a)]}{P_{\theta ,j}(a|x)}-\frac{1[(x_{i},a_{j,i})=(x,0)] }{P_{\theta ,j}(0|x)}\right) \\
T_{i,2}(\theta )& \equiv \sum_{(\tilde{x},\vec{a},x)\in \tilde{X}\times A^{|J|}\times X}\frac{\partial \Lambda _{\theta }(x^{\prime }|x,\vec{a})}{ \partial \theta }\left( \frac{1[(x_{i}^{\prime },x_{i},\vec{a} _{i})=(x^{\prime },x,\vec{a})]}{\Lambda _{\theta }(x^{\prime }|x,\vec{a})}- \frac{1[(x_{i}^{\prime },x_{i},\vec{a}_{i})=(1,x,\vec{a})]}{\Lambda _{\theta }(1|x,\vec{a})}\right) \\
T_{i,3}( \theta ) & \equiv \sum_{x\in \tilde{X}}\frac{ \partial m_{\theta }( x) }{\partial \theta }\left( \frac{1[ x_{i}=x] }{m_{\theta }( x) }-\frac{1[ x_{i}=1] }{ m_{\theta }( 1) }\right) .
\end{align*}

Furthermore, under these regularity conditions, $\hat{\theta}_{MLE}$ has the following asymptotic distribution:
\begin{equation}
\sqrt{n}(\hat{\theta}_{MLE}-\theta ^{\ast })\overset{d}{\rightarrow }N( \mathbf{0}_{d_{\theta }\times 1},\Sigma _{\theta ,MLE}^{\ast }), \label{eq:MLE_theta}
\end{equation}
where 
\begin{equation*}
\Sigma _{\theta ,MLE}^{\ast }\equiv \left[ (\frac{\partial P_{\theta ^{\ast }}}{\partial \theta })^{\prime }\Omega _{PP}^{-1}(\frac{\partial P_{\theta ^{\ast }}}{\partial \theta })+(\frac{\partial \Lambda _{\theta ^{\ast }}}{ \partial \theta })^{\prime }\Omega _{\Lambda \Lambda }^{-1}(\frac{\partial \Lambda _{\theta ^{\ast }}}{\partial \theta })+(\frac{\partial m_{\theta ^{\ast }}}{\partial \theta })^{\prime }\Omega _{mm}^{-1}(\frac{\partial m_{\theta ^{\ast }}}{\partial \theta })\right] ^{-1},
\end{equation*}
with 
\begin{align*}
\partial P_{\theta ^{\ast }}/\partial \theta & \equiv \{\partial P_{\theta ^{\ast },j}(a|x)/\partial \theta :(j,a,x)\in J\times \tilde{A}\times X\} \\ \partial \Lambda _{\theta ^{\ast }}/\partial \theta & \equiv \{\partial \Lambda _{\theta ^{\ast }}(x^{\prime }|x,\vec{a})/\partial \theta :(x^{\prime },\vec{a},x)\in \tilde{X}\times A^{|J|}\times X\} \\ \partial m_{\theta ^{\ast }}/\partial \theta & \equiv \{\partial m_{\theta ^{\ast }}(x)/\partial \theta :x\in \tilde{X}\} \\
% \Omega _{PP}^{-1}& \equiv diag\{m^{\ast }(x)\{diag\{1/P_{j}^{\ast }(a|x):a\in \tilde{A}\}+1_{|\tilde{A}|\times |\tilde{A}|}1/P_{j}^{\ast }(0|x)\}:(x,j)\in J\times X\} \\
\Omega _{\Lambda \Lambda }^{-1}& \equiv diag\{\Pi ^{\ast }(\vec{a} ,x)\{diag\{1/\Lambda ^{\ast }(x^{\prime }|x,\vec{a}):x^{\prime }\in \tilde{X}\}+1_{|\tilde{X}|\times |\tilde{X}|}1/\Lambda ^{\ast }(1|x,\vec{a})\}:(\vec{a} ,x)\in A^{|J|}\times X\} \\
\Omega _{m m }^{-1}& \equiv diag\{ 1/m^*( x) :x\in \tilde{X}\} +1_{\vert \tilde{X}\vert \times \vert \tilde{X}\vert }1/m^*( 1).
\end{align*}
Eq.\ \eqref{eq:MLE_theta} follows from showing that
\begin{align*}
\frac{1}{\sqrt{n}}\sum_{i=1}^{n}T_{i,1}(\hat{\theta}_{MLE})& =(\frac{ \partial P_{\theta ^{\ast }}}{\partial \theta })^{\prime }\Omega _{PP}^{-1}L_{n,1}+o_{p}(1) \\
\frac{1}{\sqrt{n}}\sum_{i=1}^{n}T_{i,2}(\hat{\theta}_{MLE})& =(\frac{ \partial \Lambda _{\theta ^{\ast }}}{\partial \theta })^{\prime }\Omega _{\Lambda \Lambda }^{-1}L_{n,2}+o_{p}(1) \\
\frac{1}{\sqrt{n}}\sum_{i=1}^{n}T_{i,3}(\hat{\theta}_{MLE})& =(\frac{ \partial m_{\theta ^{\ast }}}{\partial \theta })^{\prime }\Omega _{mm}^{-1}L_{n,3}+o_{p}(1)
\end{align*}
with 
\begin{align*}
\left( 
\begin{array}{c}
L_{n,1} \\ 
L_{n,2} \\ 
L_{n,3}
\end{array}
\right) 
% & \equiv \left(
% \begin{array}{c}
% \sqrt{n}\{(\hat{P}_{j}(a|x)-P_{j}^{\ast }(a|x))\hat{m}(x)/m^{\ast }(x):(j,a,x)\in J\times \tilde{A}\times X\} \\
% \{\sqrt{n}(\hat{\Lambda}(x^{\prime }|\vec{a},x)-\Lambda ^{\ast }(x^{\prime }|x,\vec{a}))\hat{\Pi}(\vec{a},x)/\Pi ^{\ast }(\vec{a},x):(x^{\prime },\vec{a },x)\in \tilde{X}\times A^{|J|}\times X\} \\
% \{\sqrt{n}(\hat{m}(x)-m^{\ast }(x)):x\in \tilde{X}\}
% \end{array}
% \right)  \\
& \overset{d}{\rightarrow }N\left( \left( 
\begin{array}{c}
\mathbf{0}_{|J\times \tilde{A}\times X| \times 1} \\ 
\mathbf{0}_{|\tilde{X}\times A^{|J|}\times X| \times 1} \\ 
\mathbf{0}_{|\tilde{X}| \times 1}
\end{array}
\right) ,\left( 
\begin{array}{ccc}
\Omega _{PP} & \mathbf{0}_{|J\times \tilde{A}\times X|\times |\tilde{X} \times A^{|J|}\times X|} & \mathbf{0}_{|J\times \tilde{A}\times X|\times | \tilde{X}|} \\
\mathbf{0}_{|\tilde{X}\times A^{|J|}\times X|\times |J\times \tilde{A}\times X|} & \Omega _{\Lambda \Lambda } & \mathbf{0}_{|\tilde{X}\times A^{|J|}\times X|\times |\tilde{X}|} \\
\mathbf{0}_{|\tilde{X}|\times |J\times \tilde{A}\times X|} & \mathbf{0}_{| \tilde{X}|\times |\tilde{X}\times A^{|J|}\times X|} & \Omega _{mm}
\end{array}
\right) \right) .
\end{align*}

Note that $\hat{\theta}_{MLE}=(\hat{\alpha}_{MLE},\hat{g}_{MLE})$, and our interest lies in the asymptotic properties of $\hat{\alpha}_{MLE}$. To make progress on characterizing the asymptotic distribution of $\hat{\alpha} _{MLE}$, we require several restrictions. First, we assume that the model-implied transition probabilities and marginal state distributions are solely a function of $g$ and not $\alpha $, i.e., $\Lambda _{\theta }=\Lambda _{g}$ and $m_{\theta }=m_{g}$. This implies that
\begin{align}
\frac{\partial \Lambda _{g^{\ast }}}{\partial \alpha } =\mathbf{0}_{| \tilde{X}\times A^{|J|}\times X|\times d_{\alpha }} ~~~~\text{ and }~~~~
\frac{\partial m_{g^{\ast }}}{\partial \alpha } =\mathbf{0}_{|\tilde{X} |\times d_{\alpha }}.\label{eq:condition_1}
\end{align}
In addition, we also assume that the following matrices are non-singular:
\begin{align}
&M_{g} \equiv (\frac{\partial \Lambda _{g^{\ast }}}{\partial g})^{\prime }\Omega _{\Lambda \Lambda }^{-1}\frac{\partial \Lambda _{g^{\ast }}}{ \partial g}+(\frac{\partial m_{g^{\ast }}}{\partial g})^{\prime }\Omega _{mm}^{-1}\frac{\partial m_{g^{\ast }}}{\partial g}, \notag\\
&M_{g}+(\frac{\partial P_{\theta ^{\ast }}}{\partial g})^{\prime }\Omega _{PP}^{-1}\frac{\partial P_{\theta ^{\ast }}}{\partial g}, \notag\\
&(\frac{\partial P_{\theta ^{\ast }}}{\partial \alpha })^{\prime }\Omega _{PP}^{-1}\frac{\partial P_{\theta ^{\ast }}}{\partial \alpha }-(\frac{ \partial P_{\theta ^{\ast }}}{\partial \alpha })^{\prime }\Omega _{PP}^{-1} \frac{\partial P_{\theta ^{\ast }}}{\partial g}\left( M_{g}+(\frac{\partial P_{\theta ^{\ast }}}{\partial g})^{\prime }\Omega _{PP}^{-1}\frac{\partial P_{\theta ^{\ast }}}{\partial g}\right) ^{-1}(\frac{\partial P_{\theta ^{\ast }}}{\partial g})^{\prime }\Omega _{PP}^{-1}\frac{\partial P_{\theta ^{\ast }}}{\partial \alpha }.\label{eq:condition_4}
\end{align}
For example, we note if there is no preliminary estimator $\hat{g}$, i.e., if $ \theta^{\ast }=\alpha^{\ast }$, then Eq.\ \eqref{eq:condition_1} is satisfied, and the non-singularity of the matrices in Eq.\ \eqref{eq:condition_4} reduces to the non-singularity of $(\frac{\partial P_{\theta ^{\ast }}}{\partial \alpha })^{\prime }\Omega _{PP}^{-1}\frac{\partial P_{\theta ^{\ast }}}{\partial \alpha }$, which follows from Assumptions \ref{ass:Regularity} and \ref{ass:Baseline2}. Under these conditions,
\begin{equation*}
\sqrt{n}(\hat{\alpha}_{MLE}-\alpha ^{\ast })\overset{d}{\rightarrow }N( \mathbf{0}_{d_{\alpha }\times 1},\Sigma _{\alpha ,MLE}^{\ast }),
\end{equation*}
where 
\begin{align}
\Sigma _{\alpha ,MLE}^{\ast }& ~=\left[ (\frac{\partial P_{\theta ^{\ast }}}{ \partial \alpha })^{\prime }\Omega _{PP}^{-1}\frac{\partial P_{\theta ^{\ast }}}{\partial \alpha }-(\frac{\partial P_{\theta ^{\ast }}}{\partial \alpha } )^{\prime }\Omega _{PP}^{-1}\frac{\partial P_{\theta ^{\ast }}}{\partial g} \left( M_{g}+(\frac{\partial P_{\theta ^{\ast }}}{\partial g})^{\prime }\Omega _{PP}^{-1}\frac{\partial P_{\theta ^{\ast }}}{\partial g}\right) ^{-1}(\frac{\partial P_{\theta ^{\ast }}}{\partial g})^{\prime }\Omega _{PP}^{-1}\frac{\partial P_{\theta ^{\ast }}}{\partial \alpha }\right] ^{-1} \notag \\
& =\left[ \Psi _{\alpha }^{\prime }\left[
\begin{array}{c}
((\mathbf{I}_{d_{P}}-\Psi _{P})\Omega _{PP}(\mathbf{I}_{d_{P}}-\Psi _{P})^{\prime })^{-1}-((\mathbf{I}_{d_{P}}-\Psi _{P})\Omega _{PP}(\mathbf{I}_{d_{P}}-\Psi _{P})^{\prime })^{-1}\Psi _{g}\times \\
(\Psi _{g}^{\prime }((\mathbf{I}_{d_{P}}-\Psi _{P})\Omega _{PP}(\mathbf{I}_{d_{P}}-\Psi _{P})^{\prime })^{-1}\Psi _{g}+M_{g})^{-1}\Psi _{g}^{\prime }((\mathbf{I}_{d_{P}}-\Psi _{P})\Omega _{PP}(\mathbf{I}_{d_{P}}-\Psi _{P})^{\prime })^{-1}
\end{array}
\right] \Psi _{\alpha }\right] ^{-1}  \notag \\
& =[ \Psi _{\alpha }^{\prime }[ (\mathbf{I}_{d_{P}}-\Psi _{P})\Omega _{PP}(\mathbf{I}_{d_{P}}-\Psi _{P})^{\prime }+\Psi _{g}M_{g}^{-1}\Psi _{g}^{\prime }] ^{-1}\Psi _{\alpha }] ^{-1}, \label{eq:MLE_alpha_var}
\end{align}
where the first equality follows from the inverse of a partitioned matrix (e.g.\ \citet[Proposition 2.8.7]{bernstein:2009}), the second equality uses that $\partial P_{\theta ^{\ast }}/\partial \alpha =(\mathbf{I}_{d_{P}}-\Psi _{P})^{-1}\Psi _{\alpha }$ and $\partial P_{\theta ^{\ast }}/\partial g=(\mathbf{I}_{d_{P}}-\Psi _{P})^{-1}\Psi _{g}$, both of which follow from Eq.\ \eqref{eq:FP}, and the third equality follows from the matrix inversion lemma (e.g.\ \citet[Corollary 2.8.8]{bernstein:2009}).

To conclude the section, we note that Eq.\ \eqref{eq:MLE_alpha_var} coincides with the $\Sigma ^{\ast }$ in Eq.\ \eqref{eq:Optimal_AVar_MD} whenever the first-stage estimator $\hat{g}$ is such that Assumption \ref{ass:Baseline} holds with
\begin{align}
\Omega _{gg} =M_{g}^{-1} ~~~~\text{ and }~~~~
\Omega _{gP} =\mathbf{0}_{d_{g}\times d_{P}}.  \label{eq:g_restrictions}
\end{align}
Once again, these restrictions are satisfied automatically if there is no preliminary estimator $\hat{g}$. In the presence of this estimator, these restrictions are satisfied if $\hat{g}$ is the partial MLE for $g^{*}$ based on the model for the state transitions and marginal distributions, i.e.,
\begin{equation}
\hat{g}~=~\underset{g\in \Theta _{g}}{\arg \max}~\left\{
\sum_{(\vec{a},x,x^{\prime })\in A^{|J|}\times X\times X}\ln \Lambda _{g}(x^{\prime }|x,\vec{a})1[(x_{i}^{\prime },x_{i},\vec{a}_{i})=(x^{\prime },x,\vec{a})] +\sum_{x\in X}\ln m_{g}(x)1[x_{i}=x]
\right\} .  \label{eq:g_partialMLE}
\end{equation}
This is a natural choice for $\hat{g}$ in this context. For example, if the state transition distribution is non-parametrically specified as in Eq.\ \eqref{eq:g_star_example} (i.e., $g=\Lambda _{g}$) and the marginal distribution is exogenous to the model, then $\hat{g}$ in Eq.\ \eqref{eq:g_hat_example} coincides in Eq.\ \eqref{eq:g_partialMLE}.

\subsection{High-order properties of the optimal $K$-MD estimator}\label{sec:highorder}

The invariance result in Theorem \ref{thm:MD_main} indicates that the asymptotic distribution of the optimal $K$-MD estimator is invariant to the number of iterations $K$. This result follows from studying the asymptotic expansion based on the first order condition that defines the estimator. On the other hand, our Monte Carlo simulations in Section \ref{sec:MCsimulations} suggest that the performance of the optimal $K$-MD estimator may improve slightly with the first few iterations, i.e., $K\leq 3$. The objective of this section is to provide a theoretical explanation of this phenomenon.

Theorem \ref{thm:MD_main} focuses on the leading term of the asymptotic distribution and ignores the higher-order terms in the asymptotic expansion, as these are asymptotically irrelevant relative to the leading term. While the invariance result implies that the asymptotic distribution of the leading term does not change with $K$, the higher-order terms can vary systematically with $K$. This effect could be noticeable in finite samples, especially since the asymptotic distribution of the leading term is invariant to $K$.

It is relevant to point out that our findings for the optimal $K$-MD estimator applied to dynamic discrete choice games are analogous to those found by \cite{aguirregabiria/mira:2002} for the $K$-PML estimator applied to single-agent problems. That is, \cite{aguirregabiria/mira:2002} predicted the asymptotic distribution of the $K$-PML estimator to be invariant to $K$, yet found in Monte Carlos that its distribution improved relative to the MLE with the first few iterations. Motivated by these results, \cite{kasahara/shimotsu:2008} studied the high-order properties of the $K$-PML estimator in the context of single-agent problems. They showed that the MLE and $K$-PML share the structure of their first order conditions and use this to prove that the high-order difference between the two estimators decreases with $K$. In this sense, \cite{kasahara/shimotsu:2008} provide a formal justification of the Monte Carlo evidence in \cite{aguirregabiria/mira:2002}.  Unfortunately, the arguments in \cite{kasahara/shimotsu:2008} do not apply to our optimal $K$-MD estimator because the MLE and the optimal $K$-MD estimator do not share a common structure of their first order conditions.\footnote{Section \ref{sec:MLE} shows that MLE and the optimal $K$-MD estimator have the same asymptotic distribution under some conditions, but this does not follow from having a common structure in their first order condition.} In addition, the analysis in \cite{kasahara/shimotsu:2008} requires the zero Jacobian property which does not necessarily hold in our econometric model. For this reason, we develop a different approach in this section.

In a related analysis, \cite{newey/smith:2004} investigate the high-order properties of the GMM and other estimators. While our $K$-MD estimator is of the GMM type, Theorem \ref{thm:HighOrder} does not directly follow from results in \cite{newey/smith:2004} for reasons that we now explain. First, the moment conditions that define the $K$-MD estimator include a plug-in estimator given by the $(K-1)$-step estimator of the CCPs. The GMM framework in \cite{newey/smith:2004} allows for a plug-in estimated weight matrix but does not seem to allow a plug-in estimator in the moment condition. Relatedly, this plug-in estimator has an asymptotic distribution that varies in a specific manner with the number of iterations $K$. As shown in Theorem \ref{thm:HighOrder}, this feature generates high-order terms on the expansion that vary with $K$. Our proof deals with this complication by using an inductive argument, which is not required to derive the analogous expansion in \citet[Lemma A.4]{newey/smith:2004}.

The main result in this section is Theorem \ref{thm:HighOrder}, which characterizes the high-order terms of the optimal $K$-MD estimator. Two features of these high-order terms are worth highlighting. First, these terms depend non-trivially on $K$. Second, some of these terms vary sharply between the first iteration (i.e.\ $K=1$) and additional ones (i.e.\ $K>1$). These features arguably help rationalize the Monte Carlo evidence in Section \ref{sec:MCsimulations}. We provide additional discussion of Theorem \ref{thm:HighOrder} and its implications after its proof.

Our analysis in this section will make several simplifying assumptions, mostly for reasons of tractability. First, we assume that there is no preliminary estimator $\hat{g}$, i.e., $ \theta^{\ast }=\alpha^{\ast }$. Note that this is the case in the Monte Carlo simulations in Section \ref{sec:MCsimulations}. Second, we assume that $\hat{P}_{0} = \hat{P}$, which is the preliminary estimator of the CCPs used in Section \ref{sec:MCsimulations}. Third, we assume that the sequence of optimal weight matrices is chosen optimally at each iteration step, which we denote by $\{ W_{k}^{\ast }:k\leq K\} $. As we explain in the paper, there are many optimal sequences of weight matrices, but $\{ W_{k}^{\ast }:k\leq K\} $ is a natural one to consider. In fact, our Monte Carlo simulations were implemented with the sample analog of $\{ W_{k}^{\ast }:k\leq K\} $. Fourth, we assume that the optimal $K$-MD estimator is implemented with an unfeasible optimal weight matrix sequence $\{ W_{k}^{\ast }:k\leq K\} $ instead of its sample analog. We denote this version of the optimal $K$-MD estimator by $\hat{\alpha}_{K-MD}^{\ast }$. We adopted this last simplification for two reasons. First, our analysis of this section can be extended to the sample analog optimal weight matrix at the expense of having a slightly longer proof and adding one term to the high order expansion in Theorem \ref{thm:HighOrder}. This additional term is a direct result of the error in estimating $\{ W_{k}^{\ast }:k\leq K\} $ and has a similar structure for each $k=1,\dots,K$. In this sense, this additional term would not change the qualitative conclusions of Theorem \ref{thm:HighOrder}. Second, while the simulations in Section \ref{sec:MCsimulations} were implemented with the sample analog of the optimal weight matrix sequence, we have also conducted the simulations with the infeasible optimal weight matrix, and we obtained similar results.

Our result requires the following additional notation. Let $q_{\alpha }\equiv -\Psi _{\alpha }^{\prime }( ( \mathbf{I} _{d_{P}}-\Psi _{P}) ^{\prime }) ^{-1}\Omega _{PP}^{-1}( \mathbf{I}_{d_{P}}-\Psi _{P}) ^{-1}\Psi _{\alpha }$.
Also, for any $j=1,\ldots ,d_{P}$, let $H[ j] \in \mathbb{R} ^{( d_{P}+d_{\alpha }) \times ( d_{P}+d_{\alpha }) }$ denote the matrix of second derivatives of $\Psi _{[ j] }( \alpha ,P) /2$ evaluated at $( \alpha ^{\ast },P^{\ast }) $ . In addition, we simultaneously define the sequences $\{ W_{k}^{\ast }:k\leq K\} $ and $\{ \Phi _{k,0P}^{\ast }:k\leq K\} $ as follows. For any $k\leq K$ and given $\Phi _{k,P0}^{\ast }$, $W_{k}^{\ast }$ is defined as optimal weight matrix for the $k$-MD estimator when $\Phi _{k,P0}=\Phi _{k,P0}^{\ast }$. From Eq.\ \eqref{eq:Sigma_MD_example}, this is given by
\begin{equation}
W_{k}^{\ast }=( ( \mathbf{I}_{d_{P}}-\Psi _{P}\Phi _{k,P0}^{\ast }) ^{\prime }) ^{-1}\Omega _{PP}^{-1}( \mathbf{I} _{d_{P}}-\Psi _{P}\Phi _{k,P0}^{\ast }) ^{-1}. \label{eq:Wopt_in_k}
\end{equation}
In turn, we define $\Phi _{1,0P}^{\ast }\equiv \mathbf{I}_{d_{P}}$ and, for any $k\leq K-1$ and given $W_{k}^{\ast }$, $\Phi _{k+1,0P}^{\ast }\equiv \Phi _{k+1,P}^{\ast }+\Phi _{k+1,0}^{\ast }$ with $\Phi _{k+1,P}^{\ast }$ and $\Phi _{k+1,0}^{\ast }$ as in Eq.\ \eqref{eq:coefficients_ML} when $ W_{k}=W_{k}^{\ast }$. By replacing Eq.\ \eqref{eq:Wopt_in_k} on Eq.\ \eqref{eq:coefficients_ML}, we conclude that
\begin{equation}
\Phi _{k+1,0P}^{\ast }=\Psi _{P}\Phi _{k,0P}^{\ast }-\Psi _{\alpha }q_{\alpha }^{-1}\Psi _{\alpha }^{\prime }(  \mathbf{I} _{d_{P}}-\Psi _{P} ^{\prime }) ^{-1}\Omega _{PP}^{-1}.
\label{eq:phi_kp1_appendix}
\end{equation}
Finally, for $\lambda =( \alpha ^{\prime },P^{\prime }) ^{\prime } $, $a=1,\ldots ,d_{\alpha }$, and $u,v=1,\ldots ,d_{\alpha }+d_{P}$, let $ U_{a,k}\in \mathbb{R} ^{( d_{P}+d_{\alpha }) \times ( d_{P}+d_{\alpha }) }$ be a matrix whose $( u,v) $ entry is
\begin{equation*}
U_{a,k}[ u,v] \equiv \left\{
\begin{array}{c}
[ \partial ^{2}\Psi /( \partial \alpha _{a}\partial \lambda _{u}) ] ^{\prime }( ( \mathbf{I}_{d_{P}}-\Psi _{P}\Phi _{k,0P}^{\ast }) ^{\prime }) ^{-1}\Omega _{PP}^{-1}( \mathbf{I}_{d_{P}}-\Psi _{P}\Phi _{k,0P}^{\ast }) ^{-1}[ \partial \Psi /\partial \lambda _{v}] \\
+[ \partial ^{2}\Psi /( \partial \alpha _{a}\partial \lambda _{v}) ] ^{\prime }( ( \mathbf{I}_{d_{P}}-\Psi _{P}\Phi _{k,0P}^{\ast }) ^{\prime }) ^{-1}\Omega _{PP}^{-1}( \mathbf{I}_{d_{P}}-\Psi _{P}\Phi _{k,0P}^{\ast }) ^{-1}[ \partial \Psi /\partial \lambda _{u}] \\
+[ \partial \Psi /\partial \alpha _{a}] ^{\prime }( ( \mathbf{I}_{d_{P}}-\Psi _{P}\Phi _{k,0P}^{\ast }) ^{\prime }) ^{-1}\Omega _{PP}^{-1}( \mathbf{I}_{d_{P}}-\Psi _{P}\Phi _{k,0P}^{\ast }) ^{-1}[ \partial ^{2}\Psi /( \partial \lambda _{u}\partial \lambda _{v}) ]
\end{array}
\right\} /2.
\end{equation*}
With this notation in place, we are now ready to state and prove the main result of this section.

%%%%%%%% DIVIDER %%%%%%%%%%%%
\begin{theorem} \label{thm:HighOrder}
Fix $K \in \mathbb{N}$ arbitrarily, Assume Assumptions \ref{ass:iid}-\ref{ass:Weight}, and that $\Psi ( \alpha ,P) $ is three-times continuously differentiable at $( \alpha ^{\ast },P^{\ast }) $. Let $\hat{\alpha}_{K-MD}^{\ast }$ denote the optimal $K$-MD estimator with $\hat{P}_{0}=\hat{P}$ and optimal weight matrices $\{ W_{k}^{\ast }:k\leq K\} $. Then,
\begin{equation}
\sqrt{n}( \hat{\alpha}_{K-MD}^{\ast }-\alpha ^{\ast }) ~=~-q_{\alpha }^{-1}\Psi _{\alpha }^{\prime }(  \mathbf{I} _{d_{P}}-\Psi _{P} ^{\prime }) ^{-1}\Omega _{PP}^{-1}\sqrt{n} ( \hat{P}-P^{\ast }) +n^{-1/2}R_{K,n}+o_{p}( n^{-1/2}) , \label{eq:Expansion}
\end{equation}
where
\begin{equation}
R_{K,n}~\equiv~ R_{K,n}^{1}+R_{K,n}^{2}+R_{K,n}^{3},  \label{eq:sumRK}
\end{equation}
with $R_{K,n}^{j}=O_{p}( 1) $ for $j=1,2,3$, and defined by
\begin{align}  
& R_{K,n}^{1}\equiv q_{\alpha }^{-1}\left\{ 
\begin{array}{c}
\left( 
\begin{array}{c}
-q_{\alpha }^{-1}\Psi _{\alpha }^{\prime }( ( \mathbf{I}
_{d_{P}}-\Psi _{P}) ^{\prime }) ^{-1}\Omega _{PP}^{-1}\sqrt{n}
( \hat{P}-P^{\ast })  \\ 
\Phi _{K,0P}^{\ast }\sqrt{n}( \hat{P}-P^{\ast }) 
\end{array}
\right) ^{\prime }U_{a,K} \\ 
\times \left( 
\begin{array}{c}
-q_{\alpha }^{-1}\Psi _{\alpha }^{\prime }( ( \mathbf{I}
_{d_{P}}-\Psi _{P}) ^{\prime }) ^{-1}\Omega _{PP}^{-1}\sqrt{n}
( \hat{P}-P^{\ast })  \\ 
\Phi _{K,0P}^{\ast }\sqrt{n}( \hat{P}-P^{\ast }) 
\end{array}
\right) 
\end{array}
\right\} _{a=1}^{d_{\alpha }}  \label{eq:defnR1} \\
& R_{K,n}^{2}\equiv q_{\alpha }^{-1}\left\{ \sum_{j=1}^{d_{P}}\left\{ 
\begin{array}{c}
[ ( ( \mathbf{I}_{d_{P}}-\Psi _{P}\Phi _{K,0P}^{\ast })
^{\prime }) ^{-1}\Omega _{PP}^{-1}( \mathbf{I}_{d_{P}}-\Psi
_{P}\Phi _{K,0P}^{\ast }) ^{-1}\sqrt{n}( \hat{P}-P^{\ast }) 
] _{[ j] } \\ 
\times\left\{ 
\begin{array}{c}
-\frac{\partial ^{2}\Psi [ j] }{\partial \alpha \partial
P^{\prime }}\Phi _{K,0P}^{\ast }\sqrt{n}( \hat{P}-P^{\ast }) + \\ 
\frac{\partial ^{2}\Psi [ j] }{\partial \alpha \partial \alpha
^{\prime }}[ q_{\alpha }^{-1}\Psi _{\alpha }^{\prime }( ( 
\mathbf{I}_{d_{P}}-\Psi _{P}) ^{\prime }) ^{-1}\Omega _{PP}^{-1}
\sqrt{n}( \hat{P}-P^{\ast }) ] 
\end{array}
\right\} 
\end{array}
\right\} \right\}  \label{eq:defnR2}\\
& R_{K,n}^{3}\equiv \left\{ 
\begin{array}{c}
1[ K=1] \times \mathbf{0}_{d_{\alpha }\times 1}+ \\ 
1[ K>1] \times q_{\alpha }^{-1}\Psi _{\alpha }^{\prime }(
( \mathbf{I}-\Psi _{P}) ^{\prime }) ^{-1}\Omega
_{PP}^{-1}( \mathbf{I}-\Psi _{P}\Phi _{K,0P}^{\ast }) ^{-1}\Psi
_{P}\times  \\ 
\sum_{k=1}^{K-1}\Psi _{P}^{K-1-k}\left[ \Psi _{\alpha }R_{k,n}+\left\{ 
\begin{array}{c}
\left( 
\begin{array}{c}
-q_{\alpha }^{-1}\Psi _{\alpha }^{\prime }( ( \mathbf{I}-\Psi
_{P}) ^{\prime }) ^{-1}\Omega _{PP}^{-1}\sqrt{n}( \hat{P}
-P^{\ast })  \\ 
\Phi _{k,0P}^{\ast }\sqrt{n}( \hat{P}-P^{\ast }) 
\end{array}
\right) ^{\prime }H[ j]  \\ 
\times \left( 
\begin{array}{c}
-q_{\alpha }^{-1}\Psi _{\alpha }^{\prime }( ( \mathbf{I}-\Psi
_{P}) ^{\prime }) ^{-1}\Omega _{PP}^{-1}\sqrt{n}( \hat{P}
-P^{\ast })  \\ 
\Phi _{k,0P}^{\ast }\sqrt{n}( \hat{P}-P^{\ast }) 
\end{array}
\right) 
\end{array}
\right\} _{j=1}^{d_{P}}\right] 
\end{array}
\right\} .\label{eq:defnR3}
\end{align}
\end{theorem}
%%%%%%%% DIVIDER %%%%%%%%%%%%
\begin{proof}
	Throughout this proof, we simplify notation by omitting the argument $( \alpha ,P) $ in any function whenever this argument is equal to $( \alpha ^{\ast },P^{\ast }) $. The first-order condition that defines $\hat{\alpha}_{K-MD}^{\ast }$ can be expressed as follows:
\begin{equation}
\hat{q}^{K}( \hat{\alpha}_{K-MD}^{\ast },\hat{P}_{K-1}) ~=~\Psi _{\alpha }( \hat{\alpha}_{K-MD}^{\ast },\hat{P}_{K-1}) ^{\prime }W_{K}^{\ast }( \hat{P}-\Psi ( \hat{\alpha}_{K-MD}^{\ast },\hat{P} _{K-1}) ) ~=~\mathbf{0}_{d_{\alpha }}, \label{eq:FOC_reexpressed}
\end{equation}
where $\hat{q}^{K}( \alpha ,P) \equiv \Psi _{\alpha }( \alpha ,P) ^{\prime }W_{K}^{\ast }( \hat{P}-\Psi ( \alpha ,P) ) $. By combining this with a second-order expansion of $ \hat{q}^{K}( \alpha ,P) $ evaluated at $( \alpha ^{\ast },P^{\ast }) $, we deduce that
\begin{equation}
\hat{q}^{K}=-\hat{q}_{\alpha }^{K}( \hat{\alpha}_{K-MD}^{\ast }-\alpha
^{\ast }) -\hat{q}_{P}^{K}( \hat{P}_{K-1}-P^{\ast })
+\left\{ \left( 
\begin{array}{c}
\hat{\alpha}_{K-MD}^{\ast }-\alpha ^{\ast } \\ 
\hat{P}_{K-1}-P^{\ast }
\end{array}
\right) ^{\prime }\hat{U}_{a,K}\left( 
\begin{array}{c}
\hat{\alpha}_{K-MD}^{\ast }-\alpha ^{\ast }\\ 
\hat{P}_{K-1}-P^{\ast }
\end{array}
\right) \right\} _{a=1}^{d_{\alpha }},  \label{eq:FOCexpansion}
\end{equation}
where $\hat{q}_{\xi }^{K}$ denote the derivatives of $\hat{q}^{K}$ with respect to $\xi \in \{\alpha,P\}$, respectively, evaluated at $( \alpha ^{\ast },P^{\ast }) $, i.e.,
\begin{align}
\hat{q}_{\xi }^{K} &\equiv -\Psi _{\alpha }^{\prime }W_{K}^{\ast }\Psi _{\xi}+\sum_{j=1}^{d_{P}}[ W_{K}^{\ast }( \hat{P}-P^{\ast }) ] _{[ j] }\frac{\partial ^{2}\Psi [ j] }{\partial \alpha \partial \xi'}, \label{eq:q_lambda}
\end{align}
and $\hat{U}_{a,K}$ is $-1/2$ times the Hessian matrix of $\hat{q}_{\alpha }^{K}( \alpha ,P) \equiv [ \partial \Psi ( \alpha ,P) /\partial \alpha _{a}] ^{\prime }W_{K}^{\ast }( \hat{P} -\Psi ( \alpha ,P) ) $ evaluated at $( \tilde{\alpha} _{K},\tilde{P}_{K-1}) $, located between $( \alpha ^{\ast },P^{\ast }) $ and $( \hat{\alpha}_{K-MD}^{\ast },\hat{P} _{K-1}) $. For $\lambda =( \alpha ^{\prime },P^{\prime }) ^{\prime }$, $a=1,\ldots ,d_{\alpha }$, and $u,v=1,\ldots ,d_{\alpha }+d_{P}$, let $\hat{U}_{a,K}\in \mathbb{R} ^{( d_{P}+d_{\alpha }) \times ( d_{P}+d_{\alpha }) }$ be a matrix whose $( u,v) $ entry is
\begin{equation*}
\hat{U}_{a,K}[ u,v] =\left\{ 
\begin{array}{c}
-[ \partial ^{3}\Psi ( \tilde{\alpha}_{K},\tilde{P}_{K-1}) /( \partial \alpha _{a}\partial \lambda _{u}\partial \lambda _{v}) ] ^{\prime }W_{K}^{\ast }( \hat{P}-\Psi ( \tilde{ \alpha}_{K},\tilde{P}_{K-1}) ) \\
+[ \partial ^{2}\Psi ( \tilde{\alpha}_{K},\tilde{P}_{K-1}) /( \partial \alpha _{a}\partial \lambda _{u}) ] ^{\prime }W_{K}^{\ast }[ \partial \Psi ( \tilde{\alpha}_{K},\tilde{P} _{K-1}) /\partial \lambda _{v}] \\
+[ \partial ^{2}\Psi ( \tilde{\alpha}_{K},\tilde{P}_{K-1}) /( \partial \alpha _{a}\partial \lambda _{v}) ] ^{\prime }W_{K}^{\ast }[ \partial \Psi ( \tilde{\alpha}_{K},\tilde{P} _{K-1}) /\partial \lambda _{u}] \\
+[ \partial \Psi ( \tilde{\alpha}_{K},\tilde{P}_{K-1}) /\partial \alpha _{a}] ^{\prime }W_{K}^{\ast }[ \partial ^{2}\Psi ( \tilde{\alpha}_{K},\tilde{P}_{K-1}) /( \partial \lambda _{u}\partial \lambda _{v}) ]
\end{array}
\right\} /2.
\end{equation*}
Note that $\hat{U}_{a,K}=U_{a,K}+o_{p}( 1) $ for all $a=1,\ldots ,d_{\alpha }$, since $( \hat{\alpha}_{K-MD}^{\ast },\hat{P} _{K-1}) =( \alpha ^{\ast },P^{\ast }) +o_{p}( 1) $ and $\Psi ( \alpha ^{\ast },P^{\ast }) =P^{\ast }$.

With minor abuse of notation, let $( \hat{q}_{\alpha }^{K}) ^{-1}$ denote the inverse of $\hat{q}_{\alpha }^{K}$ if this matrix is invertible, and $ \mathbf{I}_{d_{\alpha }}$ otherwise. From Eq.\ \eqref{eq:FOCexpansion}, we deduce that
\begin{equation}
\sqrt{n}( \hat{\alpha}_{K-MD}^{\ast }-\alpha ^{\ast })
=-q_{\alpha }^{-1}\Psi _{\alpha }^{\prime }(  \mathbf{I}
_{d_{P}}-\Psi _{P} ^{\prime }) ^{-1}\Omega _{PP}^{-1}\sqrt{n}
( \hat{P}-P^{\ast }) +n^{-1/2}\hat{R}_{K,n},
\label{eq:alpha_expansion}
\end{equation}
where $\hat{R}_{K,n}\equiv \hat{R}_{K,n}^{1}+\hat{R}_{K,n}^{2}+\hat{R}
_{K,n}^{3}+\hat{R}_{K,n}^{4}$ with
\begin{eqnarray*}
\hat{R}_{K,n}^{1} &\equiv &( \hat{q}_{\alpha }^{K}) ^{-1}\left\{
\left( 
\begin{array}{c}
\sqrt{n}( \hat{\alpha}_{K-MD}^{\ast },\hat{P}_{K-1}) \\ 
\sqrt{n}( \hat{P}_{K-1}-P^{\ast })
\end{array}
\right) ^{\prime }\hat{U}_{a,K}\left( 
\begin{array}{c}
\sqrt{n}( \hat{\alpha}_{K-MD}^{\ast },\hat{P}_{K-1}) \\ 
\sqrt{n}( \hat{P}_{K-1}-P^{\ast })
\end{array}
\right) \right\} _{a=1}^{d_{\alpha }} \\
\hat{R}_{K,n}^{2} &\equiv &\sqrt{n}( -\hat{q}_{\alpha }^{K}) ^{-1}
[ \sqrt{n}\hat{q}^{K}+\hat{q}_{P}^{K}\sqrt{n}( \hat{P}
_{K-1}-P^{\ast }) ] +\sqrt{n}q_{\alpha }^{-1}[ \sqrt{n}\hat{
q}^{K}+q_{P}^{K}\sqrt{n}( \hat{P}_{K-1}-P^{\ast }) ] \\
\hat{R}_{K,n}^{3} &\equiv &\sqrt{n}q_{\alpha }^{-1}[ \Psi _{\alpha
}^{\prime }( ( \mathbf{I}_{d_{P}}-\Psi _{P}) ^{\prime
}) ^{-1}\Omega _{PP}^{-1}\sqrt{n}( \hat{P}-P^{\ast }) -
[ \sqrt{n}\hat{q}^{K}+q_{P}^{K}\sqrt{n}( \hat{P}_{K-1}-P^{\ast
}) ] ] ,
\end{eqnarray*}
and $\hat{R}_{K,n}^{4}$ represents a remainder term that absorbs the expression in Eq.\ \eqref{eq:alpha_expansion} whenever $ \hat{q}_{\alpha }^{K}$ is not invertible, i.e., $\hat{R}_{K,n}^{4} = {\bf 0}_{d_{\alpha}\times 1}$ if $ \hat{q}_{\alpha }^{K}$ is invertible. To complete the proof, it suffices to show that $\hat{R}_{K,n}=R_{K,n}+o_{p}( 1) $. In turn, this follows from showing that $\hat{R}_{K,n}^{j}=R_{K,n}^{j}+o_{p}( 1) $ for $j=1,2,3$, with $R_{K,n}^{1}$, $R_{K,n}^{2}$, and $ R_{K,n}^{3} $ defined by Eqs.\ \eqref{eq:defnR1}, \eqref{eq:defnR2}, and \eqref{eq:defnR3}, respectively, and $\hat{R}_{K,n}^{4}=o_{p}( 1) $.

Note that $\hat{R}_{K,n}^{1}=R_{K,n}^{1}+o_{p}( 1) $ follows immediately from $\hat{q}_{\alpha }^{K}=q_{\alpha }+o_{p}( 1) $, $\hat{U}_{a,K}=U_{a,K}+o_{p}( 1) $ for all $ a=1,\ldots ,d_{\alpha }$, and Theorem \ref{thm:2stepB}. Also, $ \hat{R}_{K,n}^{4}= {\bf 0}_{d_{\alpha}\times 1}$ is implied by $\hat{q}_{\alpha }^{K}$ being non-singular, and this occurs with probability approaching one by $ \hat{q}_{\alpha }^{K}=q_{\alpha }+o_{p}( 1) $. From here, $\hat{R} _{K,n}^{4}=o_{p}( 1) $ follows. Next, consider the following derivation.
\begin{eqnarray*}
\hat{R}_{K,n}^{2} =( \hat{q}_{\alpha }^{K})
^{-1}\sum_{j=1}^{d_{P}}[ W_{K}^{\ast }\sqrt{n}( \hat{P}-P^{\ast
}) ] _{[ j] }\left\{ 
\begin{array}{c}
-\frac{\partial ^{2}\Psi [ j] }{\partial \alpha \partial
P^{\prime }}\sqrt{n}( \hat{P}_{K-1}-P^{\ast }) + \\ 
-\frac{\partial ^{2}\Psi [ j] }{\partial \alpha \partial \alpha
^{\prime }}[ ( -q_{\alpha }) ^{-1}[ \sqrt{n}\hat{q}
^{K}+q_{P}^{K}\sqrt{n}( \hat{P}_{K-1}-P^{\ast }) ] ]
\end{array}
\right\} = R_{K,n}^{2}+o_{p}( 1) ,
\end{eqnarray*}
where the first equality follows from Eq.\ \eqref{eq:q_lambda}, and the second inequality follows from $\hat{q}_{\alpha }^{K}=q_{\alpha }+o_{p}( 1) $, $\sqrt{n}( \hat{\alpha}_{K-MD}^{\ast }-\alpha ^{\ast }) =-q_{\alpha }^{-1}[ \sqrt{n}\hat{q}^{K}+q_{P}^{K}\sqrt{n }( \hat{P}_{K-1}-P^{\ast }) ] +o_{p}( 1) $, and Theorem \ref{thm:2stepB}.

To conclude the proof, it suffices to show $\hat{R} _{K,n}^{3}=R_{K,n}^{3}+o_{p}( 1) $. As a preliminary step, note that
\begin{equation}
\hat{R}_{K,n}^{3}~=~q_{\alpha }^{-1}\Psi _{\alpha }^{\prime }( ( \mathbf{I}_{d_{P}}-\Psi _{P}) ^{\prime }) ^{-1}\Omega _{PP}^{-1}( \mathbf{I}_{d_{P}}-\Psi _{P}\Phi _{K,0P}^{\ast }) ^{-1}\Psi _{P}\hat{R}_{K,n}^{5}. \label{eq:R3_and_R5}
\end{equation}
where
\begin{equation}
\hat{R}_{K,n}^{5}~\equiv~ \sqrt{n}[ \sqrt{n}( \hat{P}_{K-1}-P^{\ast }) -\Phi _{K,0P}^{\ast }\sqrt{n}( \hat{P}-P^{\ast }) ] . \label{eq:R5_expansion}
\end{equation}
The desired result then follows immediately from showing that $\hat{R} _{K,n}^{5}=R_{K,n}^{5}+o_{p}( 1) $ for $R_{K,n}^{5}=O_{p}( 1) $ that is consistent with Eq.\ \eqref{eq:defnR3}. We prove this by induction. 
% First, we show that $\hat{R}_{1,n}^{5}=R_{1,n}^{5}+o_{p}( 1) $, $\hat{R}_{1,n}^{3}=R_{1,n}^{3}+o_{p}( 1) $, and $\hat{ R}_{1,n}=R_{1,n}+o_{p}( 1) $ with $R_{1,n}^{5}= \mathbf{0}_{d_{P}\times 1}$, $R_{1,n}^{3}= \mathbf{0}_{d_{\alpha}\times 1}$ and $R_{1,n}=R_{1,n}^{1}+R_{1,n}^{2}$. Second, we show that if $\hat{R}_{K,n}^{5}=R_{K,n}^{5}+o_{p}( 1) $, $\hat{ R}_{K,n}^{3}=R_{K,n}^{3}+o_{p}( 1) $, and $\hat{R} _{K,n}=R_{K,n}+o_{p}( 1) $, then $\hat{R} _{K+1,n}^{5}=R_{K+1,n}^{5}+o_{p}( 1) $ with
% Having shown this, $\hat{R}_{K+1,n}^{3}=R_{K+1,n}^{3}+o_{p}( 1) $ follows immediately. In turn, this and our previous results imply that $\hat{ R}_{K+1,n}=R_{K+1,n}+o_{p}( 1) $ with $R_{K+1,n}$ is determined by Eq.\ \eqref{eq:sumRK}.

We begin with the initial step, i.e., $K=1$. Eq.\ \eqref{eq:R5_expansion}, $ \sqrt{n}( \hat{P}_{0}-P^{\ast }) =\sqrt{n}( \hat{P}-P^{\ast }) $, and $\Phi _{1,0P}^{\ast }=\mathbf{I}_{d_{P}}$ implies that $\hat{ R}_{1,n}^{5}=R_{1,n}^{5}=\mathbf{0}_{d_{P}\times 1}$. This and Eq.\ \eqref{eq:R3_and_R5} then imply that $R_{K,n}^{3}=\mathbf{0}_{d_{\alpha }\times 1}$. This and Eq.\ \eqref{eq:sumRK} then implies that $ R_{K,n}=R_{K,n}^{1}+R_{K,n}^{2}$, concluding the initial step.

We now proceed with the inductive step. Consider the following argument.
\begin{align*}
\hat{R}_{K+1,n}^{5} &=\sqrt{n}[ [ -\Psi _{P}\Phi _{K,0P}^{\ast
}+\Psi _{\alpha }q_{\alpha }^{-1}\Psi _{\alpha }^{\prime }( ( 
\mathbf{I}_{d_{P}}-\Psi _{P}) ^{\prime }) ^{-1}\Omega _{PP}^{-1}
] \sqrt{n}( \hat{P}-P^{\ast }) +\sqrt{n}( \hat{P}
_{K}-P^{\ast }) ] \\
&=\Psi _{P}\hat{R}_{K,n}^{5}+\sqrt{n}[ \sqrt{n}( \hat{P}
_{K}-P^{\ast }) -\Psi _{\alpha }[ -q_{\alpha }^{-1}\Psi _{\alpha
}^{\prime }( ( \mathbf{I}_{d_{P}}-\Psi _{P}) ^{\prime
}) ^{-1}\Omega _{PP}^{-1}\sqrt{n}( \hat{P}-P^{\ast }) 
] -\Psi _{P}\sqrt{n}( \hat{P}_{K-1}-P^{\ast }) ] \\
&=\Psi _{P}\hat{R}_{K,n}^{5}+\Psi _{\alpha }\hat{R}_{K,n}+\sqrt{n}[ 
\sqrt{n}( \hat{P}_{K}-P^{\ast }) -\Psi _{\alpha }\sqrt{n}( 
\hat{\alpha}_{K-MD}^{\ast }-\alpha ^{\ast }) -\Psi _{P}\sqrt{n}( 
\hat{P}_{K-1}-P^{\ast }) ] \\
&=\Psi _{P}\hat{R}_{K,n}^{5}+\Psi _{\alpha }\hat{R}_{K,n}+\left\{ \left( 
\begin{array}{c}
\sqrt{n}( \hat{\alpha}_{K-MD}^{\ast }-\alpha ^{\ast }) \\ 
\sqrt{n}( \hat{P}_{K-1}-P^{\ast })
\end{array}
\right) ^{\prime }\hat{H}_{K}[ j] \left( 
\begin{array}{c}
\sqrt{n}( \hat{\alpha}_{K-MD}^{\ast }-\alpha ^{\ast }) \\ 
\sqrt{n}( \hat{P}_{K-1}-P^{\ast })
\end{array}
\right) \right\} _{j=1}^{d_{P}} \\
&=R_{K+1,n}^{5}+o_{p}( 1) ,
\end{align*}
where the first equality holds by Eq.\ \eqref{eq:phi_kp1_appendix}, the second equality holds by Eq.\ \eqref{eq:R5_expansion}, the third equality holds by Eq.\ \eqref{eq:alpha_expansion}, the fourth equality follows from a second-order expansion of Eq.\ \eqref{eq:k-PkDefn} centered at $( \alpha ^{\ast },P^{\ast })$ where for $\lambda =( \alpha ^{\prime },P^{\prime }) ^{\prime }$ and $j=1,\ldots ,d_{P}$, $\hat{H}_{K}[ j] $ denotes $1/2$ times the second derivative of $j$'th coordinate of $\Psi $ evaluated at $( \breve{\alpha}_{K},\breve{P}_{K-1}) $ (located between $( \alpha ^{\ast },P^{\ast }) $ and $( \hat{\alpha}_{K-MD}^{\ast },\hat{P} _{K-1}) $), and the fifth equality follows from the inductive assumption, Eq.\ \eqref{eq:alpha_expansion}, and the following iterative definition:
\begin{equation}
R_{K+1,n}^{5} \equiv \Psi _{P}R_{K,n}^{5}+\Psi _{\alpha }R_{K,n}+\left\{ 
\begin{array}{c}
\left( 
\begin{array}{c}
-q_{\alpha }^{-1}\Psi _{\alpha }^{\prime }( ( \mathbf{I}
_{d_{P}}-\Psi _{P}) ^{\prime }) ^{-1}\Omega _{PP}^{-1}\sqrt{n}
( \hat{P}-P^{\ast }) \\ 
\Phi _{K,0P}^{\ast }\sqrt{n}( \hat{P}-P^{\ast })
\end{array}
\right) ^{\prime }H[ j] \\ 
\times \left( 
\begin{array}{c}
-q_{\alpha }^{-1}\Psi _{\alpha }^{\prime }( ( \mathbf{I}
_{d_{P}}-\Psi _{P}) ^{\prime }) ^{-1}\Omega _{PP}^{-1}\sqrt{n}
( \hat{P}-P^{\ast }) \\ 
\Phi _{K,0P}^{\ast }\sqrt{n}( \hat{P}-P^{\ast })
\end{array}
\right)
\end{array}
\right\} _{j=1}^{d_{P}}.  
\label{eq:R5_expansion2}
\end{equation}
To complete the proof, we note that Eq.\ \eqref{eq:R3_and_R5}, $R_{1,n}^{5}=\mathbf{0}_{d_{P}\times 1}$, and Eq.\ \eqref{eq:R5_expansion2} imply Eq.\ \eqref{eq:defnR3}.
\end{proof}
%%%%%%%% DIVIDER %%%%%%%%%%%%

Theorem \ref{thm:HighOrder} provides a detailed characterization of the asymptotic distribution of the optimal $K$-MD estimator. Eq.\ \eqref{eq:Expansion} shows that the estimator converges in distribution to a leading term (which does not depend on $K$, as expected from Theorem \ref{thm:MD_main}), and a high-order term, denoted $n^{-1/2}R_{K,n}$. Several observations about the high-order term are in order. First, $R_{K,n}$ is bounded in probability, and so Theorem \ref{thm:HighOrder} provides an exact rate of convergence of the high-order term, equal to $n^{-1/2}$. Second, unlike the leading term, $R_{K,n} $ depends non-trivially on $K$. Third, Eq.\ \eqref{eq:sumRK} decomposes $R_{K,n} $ into three terms, and the third term changes sharply between $K=1$ and $K>1 $. These observations could help explain the Monte Carlo evidence in Section \ref{sec:MCsimulations}. To provide further evidence of this, Table \ref{tab:MC_H0} shows the bias, variance, and mean squared error of $n^{-1/2}R_{K,n}$ in the designs used in our Monte Carlo simulations when $n=500$.\footnote{The results for $n=1,000$ and $n=2,000$ are analogous and available upon request.} In all designs, the magnitude of the mean squared error of $n^{-1/2}R_{K,n}$ is relatively large for the first few iterations (i.e.\ $K\leq 3$), and decreases sharply with additional iterations.

The main results of this paper (based on the first order asymptotic approximation) imply that iterations do not affect the asymptotic distribution of the optimal $K$-MD estimator. Theorem \ref{thm:HighOrder} indicate that iterations could have an impact on the high-order approximation of this asymptotic distribution. Given these new findings, a natural question is whether Theorem \ref{thm:HighOrder} can be used to make a better choice of $K$, or even reconsider the definition of optimality for the $K$-MD estimator. We consider that this is a hard idea to put into practice, as the terms that compose $R_{K,n}$ appear to be complicated and have a non-trivial pattern of dependence on $K$. In fact, while Table \ref{tab:MC_H0} suggests that few iterations could reduce the mean squared error of the high-order terms, this is not a general result that we can establish beyond our specific Monte Carlo setting.

\begin{table}[h]
	\begin{center}
	\scalebox{0.79}{
	\begin{tabular}{r|rrrrrrrr}
  \hline\hline
 \multicolumn{1}{c|}{Statistic} & \multicolumn{1}{c}{$K=1$} & \multicolumn{1}{c}{$K=2$} & \multicolumn{1}{c}{$K=3$} & 
\multicolumn{1}{c}{$K=4$} & \multicolumn{1}{c}{$K=5$} & \multicolumn{1}{c}{$K=10$} & \multicolumn{1}{c}{$K=15$} & \multicolumn{1}{c}{$K=20$} \\
                            \hline\hline  \multicolumn{1}{r}{} & \multicolumn{8}{c}{Design 1: $(\lambda _{RN}^{\ast },\lambda _{EC}^{\ast },\lambda _{RS}^{\ast },\lambda _{FC,1}^{\ast },\lambda _{FC,2}^{\ast },\beta^{\ast })
=(2.8,0.8,0.7,0.6,0.4,0.95)$} \\   \hline
\multicolumn{1}{c|}{Bias} &  1.77 &  2.60 &  1.24 &  0.50 &  0.46 &  0.39 &  0.39 &  0.39 \\ 
  \multicolumn{1}{c|}{Var} & 11.23 &  6.69 &  1.91 &  1.24 &  1.04 &  0.93 &  0.93 &  0.93 \\ 
  \multicolumn{1}{c|}{MSE} & 14.35 & 13.46 &  3.44 &  1.49 &  1.25 &  1.09 &  1.08 &  1.08 \\ 
   \hline\hline \multicolumn{1}{r}{} & \multicolumn{8}{c}{Design 2: $(\lambda _{RN}^{\ast },\lambda _{EC}^{\ast },\lambda _{RS}^{\ast },\lambda _{FC,1}^{\ast },\lambda _{FC,2}^{\ast },\beta^{\ast })
                        =(2,1.8,0.2,0.01,0.03,0.95)$} \\  \hline
\multicolumn{1}{c|}{Bias} &  2.98 &  3.83 &  2.30 &  1.47 &  1.20 &  0.90 &  0.87 &  0.87 \\ 
  \multicolumn{1}{c|}{Var} & 20.57 & 16.34 &  6.04 &  3.41 &  3.02 &  2.60 &  2.61 &  2.61 \\ 
  \multicolumn{1}{c|}{MSE} & 29.43 & 31.03 & 11.34 &  5.57 &  4.45 &  3.40 &  3.37 &  3.37 \\ 
   \hline\hline \multicolumn{1}{r}{} & \multicolumn{8}{c}{Design 3: $(\lambda _{RN}^{\ast },\lambda _{EC}^{\ast },\lambda _{RS}^{\ast },\lambda _{FC,1}^{\ast },\lambda _{FC,2}^{\ast },\beta^{\ast })
                        =(2.2,1.45,0.45,0.22,0.29,0.95)$} \\  \hline
\multicolumn{1}{c|}{Bias} &  2.59 &  3.32 &  1.95 &  1.15 &  0.94 &  0.73 &  0.72 &  0.72 \\ 
  \multicolumn{1}{c|}{Var} & 16.30 & 10.82 &  3.81 &  2.26 &  2.02 &  1.83 &  1.84 &  1.84 \\ 
  \multicolumn{1}{c|}{MSE} & 23.00 & 21.84 &  7.61 &  3.58 &  2.90 &  2.36 &  2.36 &  2.36 \\
   \hline
\hline
\end{tabular}}
	\end{center}
	\caption{\small Bias, variance, and mean squared error of the asymptotic high-order terms for the optimal $K$-MD estimator of $\lambda _{RN}^{\ast }$ when $n=500$. ``Bias'', ``Var'', and ``MSE'' denotes the average empirical bias, variance, and mean squared error of $n^{-1/2}R_{K,n}$ based on Theorem \ref{thm:HighOrder}, where the average is computed over $S=10,000$ simulations.}
	\label{tab:MC_H0}
\end{table}
\end{small}

\bibliography{BIBLIOGRAPHY}

\end{document}